\documentclass[preprint,11pt]{article}
\usepackage{amsthm}
\usepackage{amsmath}
\usepackage{dsfont}
\usepackage{amssymb}
\usepackage{graphicx}
\usepackage{hyperref}
\usepackage{tikz}
\usepackage{cases}
\usepackage{empheq}
\usepackage{enumerate}
\usepackage{enumitem} 
\usetikzlibrary{arrows,positioning,fit,shapes,calc}

\newtheorem{Thm}{Theorem}[section]
\newtheorem{Lem}[Thm]{Lemma}
\newtheorem{Def}[Thm]{Definition}
\newtheorem{Not}[Thm]{Notation}
\newtheorem{Pro}[Thm]{Proposition}
\newtheorem{Cor}[Thm]{Corollary}
\newtheorem{Rem}[Thm]{Remark}
\numberwithin{equation}{section}

\hypersetup{colorlinks,linkcolor={red},citecolor={darkgreen},urlcolor={red}} 
\definecolor{darkgreen}{rgb}{0.1,0.5,0.1}

\newcommand{\beqa}{\begin{eqnarray}}
\newcommand{\eeqa}[1]{\label{#1}\end{eqnarray}}
\newcommand{\beq}{\begin{equation}}
\newcommand{\eeq}[1]{\label{#1}\end{equation}}

\newcommand{\rmd}{{\mathrm{ d}}}
\newcommand{\rme}{{\mathrm{ e}}}
\newcommand{\rmi}{{\mathrm{ i}}}
\newcommand{\R}{{\mathbb{ R}}}
\newcommand{\bbR}{{\mathbb{R}}}
\newcommand{\bbC}{{\mathbb{C}}}
\newcommand{\bbA}{{\mathbb{A}}}

\newcommand{\N}{{\mathbb{N}}}

\newcommand{\bk}{\mathbf{k}}
\newcommand{\ba}{\mathbf{a}}
\newcommand{\bb}{\mathbf{b}}
\newcommand{\be}{\mathbf{e}}
\newcommand{\bh}{\mathbf{h}}
\newcommand{\bp}{\mathbf{p}}
\newcommand{\bm}{\mathbf{m}}
\newcommand{\bx}{\mathbf{x}}

\newcommand{\bU}{\mathbf{U}}

\newcommand{\bu}{\mathbf{u}}
\newcommand{\bv}{\mathbf{v}}

\newcommand{\bC}{\mathbf{C}}
\newcommand{\bE}{\mathbf{E}}
\newcommand{\bV}{\mathbf{V}}
\newcommand{\bH}{\mathbf{H}}
\newcommand{\bP}{\mathbf{P}}
\newcommand{\bM}{\mathbf{M}}

\newcommand{\bL}{\mathbf{L}}
\newcommand{\bG}{\mathbf{G}}

\newcommand{\bbE}{\mathbb{E}}
\newcommand{\bbF}{\mathbb{F}}
\newcommand{\bbH}{\mathbb{H}}
\newcommand{\bbP}{\mathbb{P}}
\newcommand{\bbM}{\mathbb{M}}
\newcommand{\bbU}{\mathbb{U}}

\newcommand{\bbX}{\mathbb{X}}

\newcommand{\calP}{{\mathcal{P}}}

\newcommand{\calZ}{{\mathcal{Z}}}

\newcommand{\ds}{\displaystyle}

\begin{document}
\vspace{-1in}

\title{Long time behaviour of the solution of Maxwell's equations in dissipative generalized Lorentz materials (II) A  modal approach}

\author{Maxence Cassier$^{a}$, Patrick Joly$^{b}$ and Luis Alejandro Rosas Mart\'inez$^{b}$  \\ \ \\
{
\footnotesize $^a$ Aix Marseille Univ, CNRS,  Centrale Med, Institut Fresnel, Marseille, France }\\ 
{\footnotesize $^b$ POEMS$^1$, CNRS, Inria, ENSTA Paris, Institut Polytechnique de Paris, 91120 Palaiseau, France}\\ 
{\footnotesize (maxence.cassier@fresnel.fr, patrick.joly@inria.fr, alejandro.rosas@ensta-paris.fr)}}

	\maketitle

	\begin{abstract}
This work concerns the analysis of  electromagnetic dispersive media modelled by generalized Lorentz models. More precisely, this paper is the second  of two articles dedicated to the long time behaviour of solutions of Maxwell's equations in dissipative Lorentz media, via  the long time decay rate of the electromagnetic energy for the corresponding Cauchy problem. In opposition to the frequency dependent Lyapunov functions  approach  used in \cite{cas-jol-ros-22}, we develop a method based on the spectral analysis of the underlying non selfadjoint operator of the model. Although more involved, this approach is closer to physics, as it uses the dispersion relation of the model, and has the advantage to provide more precise and more optimal results, leading to distinguish the notion of weak and strong dissipation. 
\end{abstract}

\section{Introduction}
\footnotetext[1]{POEMS (Propagation d'Ondes: Etude Math\'ematique et Simulation) is a mixed research team (UMR 7231) between CNRS (Centre National de la Recherche Scientifique), ENSTA Paris (Ecole Nationale Sup\'erieure de Techniques Avanc\'ees) and INRIA (Institut National de Recherche en Informatique et en Automatique).}
\subsection{Motivation} \label{sec_motivation}
This paper is devoted to the study of  the long time behaviour, of solutions of the Cauchy problem for Maxwell's equations in a large class dissipative and dispersive materials: the generalized Lorentz media. We refer for instance to \cite{Ziol-01,Pad-07,Li-2013,cas-jol-kach-17,cas-mil-17,cas-jol-haz-22} for the interest of considering such media with respect to applications, in particular in the context of metamaterials. We are more precisely interested in the decay of the electromagnetic energy. The evolution problem we consider is the following 
\begin{align*}
	\mbox{Find} \quad \quad \left\{
	\begin{array}{ll}
		\textbf{E}(\bx,t):\R^3\times\R^+\longrightarrow\R^3 &\textbf{H}(\bx,t):\R^3\times\R^+\longrightarrow\R^3\\ [6pt]
		\mathbf{P}_j(\bx,t):\R^3\times\R^+\longrightarrow\R^3, \ 1 \leq j \leq N_e, \quad  &\mathbf{M}_\ell(\bx,t):\R^3\times\R^+\longrightarrow\R^3,  \ 1 \leq \ell \leq N_m,
	\end{array}
	\right.
\end{align*}
such that (for all $1 \leq j \leq N_e, 1 \leq \ell \leq N_m$)
\begin{subequations}\label{planteamiento Lorentz} 
	\begin{empheq}[left=\empheqlbrace]{align}
		&\varepsilon_0\,\partial_t\,\textbf{E}-\nabla\times\textbf{H}+\varepsilon_0\,\sum_{j=1}^{N_e}\,\Omega_{e,j}^2\,\partial_t\,\mathbf{P}_j=0, &(\bx,t)\in\R^3\times\R^{+,*},\label{E Lorentz}\\
		&\mu_0\,\partial_t\,\textbf{H}+\nabla\times\textbf{E}+\mu_0\,\sum_{\ell=1}^{N_m}\,\Omega_{m,\ell}^2\,\partial_t\,\mathbf{M}_\ell=0, & (\bx,t)\in\R^3\times\R^{+,*},\label{H Lorentz}\\[6pt]
		&\partial_t^2\,\mathbf{P}_j+\alpha_{e,j}\,\partial_t\,\mathbf{P}_j+\omega_{e,j}^2\,\mathbf{P}_j=\mathbf{E}, & (\bx,t)\in\R^3\times\R^{+,*},\label{P Lorentz}\\[12pt]
		&\partial_t^2\,\mathbf{M}_\ell+\alpha_{m,\ell}\,\partial_t\,\mathbf{M}_\ell+\omega_{m,\ell}^2\,\mathbf{M}_\ell=\mathbf{H}, & (\bx,t)\in\R^3\times\R^{+,*},\label{M Lorentz}
	\end{empheq}
\end{subequations}
completed by the divergence free initial conditions
\begin{equation} \label{CI}
	\left\{\begin{array}{l}
		\mathbf{E}(\cdot, 0) =  \mathbf{E}_0, \  \mathbf{H}(\cdot, 0) =  \mathbf{H}_0, \ \bP(\cdot, 0)=\bP_0,  \ \partial_t \mathbf{P}(\cdot, 0) =\dot{\bP}_0,   \ \bM(\cdot, 0)=\bM_0,  \ \partial_t \mathbf{M}(\cdot, 0) =\dot{\bM}_0  \\ [8pt] 
		\mbox{ with }  \mathbf{E}_0, \,  \bH_0 \in \bL^2(\mathbb{R}^3),   \quad  \bP_0,  \ \dot{\bP}_0 \in \bL^2(\mathbb{R}^3)^{N_e}  \mbox{ and }  \  \bM_0,   \ \dot{\bM}_0 \in \bL^2(\mathbb{R}^3)^{N_m},  
\end{array} \right.
\end{equation}
where  $\bL^2(\mathbb{R}^3)=L^2(\mathbb{R}^3)^3$ and, for simplicity, we  have used the notations
\begin{equation*} \label{notPM} {\mathbf{P}}=(\mathbf{P}_j)\quad \mbox{and} \quad {\mathbf{M}}=(\mathbf{M}_\ell) \quad \mbox{ with } \quad  (\mathbf{P}_j):=(\mathbf{P}_j)_{j=1}^{N_e} \ \mbox{ and } \ (\mathbf{M}_\ell):=(\mathbf{M}_\ell)_{\ell=1}^{N_m}.
\end{equation*}
In  \eqref{planteamiento Lorentz}, the fields $\mathbf{E}$ and $\mathbf{H}$  are respectively the electric and magnetic fields while the $\mathbf{P}_j$ and the $\mathbf{M}_\ell$ are auxiliary unknowns (polarization and magnetization fields respectively). The coefficients $(\omega_{e,j},\omega_{m,\ell})$ and $( \Omega_{e,j},\Omega_{m,\ell}) $ are supposed to satisfy 
\begin{equation} \label{hypomegabis}
	\begin{array}{l}
		0 < \Omega_{e,1} \leq \cdots \leq \Omega_{e,N_e}, \quad 0 < \Omega_{m,1} \leq \cdots \leq \Omega_{m,N_m},  \\ [12pt]
		\omega_{e,j} > 0, \quad 1 \leq j \leq N_e, \quad  \omega_{m,\ell} > 0,\quad 1 \leq \ell \leq N_m.
	\end{array}
\end{equation}
These coefficients are responsible for dispersion effects. Finally, the coefficients $(\alpha_{e,\ell}, \alpha_{m,\ell})$ must be non negative
\begin{equation} \label{hypalpha}
	\alpha_{e,j} \geq 0, \quad 1 \leq j \leq N_e,  \quad \alpha_{m,\ell} \geq 0\quad 1 \leq \ell \leq N_m.
\end{equation}
The presence of positive coefficients $(\alpha_{e,\ell}, \alpha_{m,\ell})$ is responsible for dissipation effects. The reader will notice that one can assume without any loss of generality that the couples $(\alpha_{e,j}, \omega_{e,j})$ (resp. $(\alpha_{m,\ell}, \omega_{m,\ell})$) are all distinct the ones from the others. 
\begin{Rem}
Usually in physical applications, one assumes also that the intial conditions  \eqref{CI} of the evolution system  \eqref{planteamiento Lorentz} satisfy also 
\begin{equation} \label{CI-Phys}
 \nabla \cdot \mathbf{E}_0=\nabla \cdot \mathbf{H}_0=0, \quad  \mathbf{P}_0 =\dot{\mathbf{P}}_0= 0 \ \mbox{ and } \ \mathbf{M}_0 =\dot{\mathbf{M}}_0= 0.
\end{equation}
In other words, one considers  that the initial fields are  divergence free   and  that the medium has no polarization, no magnetization and induced currents at the initial time.
\end{Rem}

\noindent Let us introduce the two rational functions 
\begin{equation}\label{eq.permmitivity-permeabiity}
	\varepsilon(\omega) = \varepsilon_0 \, \Big( 1 - \sum_{j=1}^{N_e} \frac{\Omega_{e,j}^2}{q_{e,j}(\omega)} \Big) \quad \mbox{ and } \quad \mu(\omega) = \mu_0 \, \Big( 1 - \sum_{\ell=1}^{N_m} \frac{\Omega_{m,\ell}^2}{q_{m,\ell}(\omega)}\Big),
\end{equation}
where the second order polynomials  $q_{e,j}$ and $q_{m,\ell}$ are defined by
\begin{equation}\label{eq.polynom}
	q_{e,j}(\omega)=\omega^2+ \rmi\, \alpha_{e,j} \, \omega- \omega_{e, j}^2 \  \mbox{ and } \ q_{m,\ell}(\omega)=\omega^2+ \rmi \, \alpha_{m, \ell} \, \omega- \omega_{m, \ell}^2 .
\end{equation}
As functions of the frequency $\omega$, $\varepsilon$ and $\mu$  are respectively the permittivity and permeability of the generalized Lorentz media  associated to the Maxwell's system \eqref{planteamiento Lorentz} (see Section 1.1 of \cite{cas-jol-ros-22} for more details about their physical interpretation). 
We denote   by $\mathcal{P}_e$ (resp. $\mathcal{P}_m$)  the set of poles of the function $\varepsilon$ (resp. $\mu$). We also introduce   $\mathcal{Z}_e$   (resp. $\mathcal{Z}_m$) the set of zeros of  $ \varepsilon$ (resp. $ \mu$).
Generalized Lorentz media are causal and passive  electromagnetic materials, as defined in \cite{cas-jol-ros-22} for instance. The counterpart of causality and  passivity  in the frequency domain is that  $\omega \mapsto \omega  \,\varepsilon(\omega)$ and $\omega \mapsto \omega  \,\mu(\omega)$  are Herglotz  functions of the frequency.  Herglotz functions are analytic functions $f$  in the upper half-plane  $\bbC^+:=\{\omega \in \mathbb{C}\mid \operatorname{Im}\omega>0\}$ whose imaginary part $\operatorname{Im}f$ is non-negative on $\bbC^+$. For more details about passive electromagnetic media and Herglotz functions, we refer  to \cite{cas-jol-kach-17}, \cite{cas-mil-17} and \cite{cas-jol-ros-22}. 
\\ [12pt]
\noindent We make the two following assumptions which are linked to the irreducibility of the dispersion relation associated to the system \eqref{planteamiento Lorentz} that involves the product $\varepsilon(\omega) \mu(\omega)$ (see Section \ref{sec_dispersion-1}):
\begin{itemize}
	\item $(\mathrm{H}_1)$:  the electric polynomials  $q_{e,j}$ (see \eqref{eq.polynom})  with  distinct indices $j$  do not have common roots. The same holds for the magnetic  polynomials $q_{m,\ell}$  with  distinct indices $\ell$. 
	\item  $(\mathrm{H}_2)$:  The zeros of the permittivity  $\varepsilon$ are not poles of the permittivity $\mu$ and vice versa. Namely, one assumes that  $\calP_e \cap \calZ_m= \varnothing$ and $\calP_m\cap \calZ_e=\varnothing $. 
\end{itemize}
\begin{Rem}\label{Rem-roots}
	It is easy to see that, as $\alpha_{e,j} \geq 0$,   the roots of $q_{e,j}$  have a non positive imaginary part. More  precisely, as the discriminant of $q_{e,j}$ is $\delta_{e,j} = 4 \,\omega_{e,j}^2 - \alpha_{e,j}^2$ \\ [-12pt]
	\begin{itemize}
		\item If $\alpha_{e,j} <  2 \, \omega_{e,j}$, the two roots of $q_{e,j}$ are $-\rmi \, \alpha_{e,j}/2 \pm \, \delta_{e,j}^{\frac{1}{2}} /2\notin \rmi \, \R^-$,
		\item  If $\alpha_{e,j}\geq 2 \, \omega_{e,j} $, the two roots of $q_{e,j}$ are $-\rmi \, \big(\alpha_{e,j} \pm \, |\delta_{e,j}|^{\frac{1}{2}}\big) /2 \in \rmi \, \R^-$.
	\end{itemize}
Thus, when $\alpha_{e,j} <  2 \, \omega_{e,j}$, the two roots of $q_{e,j}$, $\omega_*$ and $-\overline{\omega_*}\notin \rmi \R^-$, are distinct. Moreover, as  $\alpha_{e,j}=-2\operatorname{Im}(\omega_*)$ and $\omega_{e,j}= |\omega_*|$, $q_{e,j}$  can not share a common root with an other electric polynomial $q_{e,j'}$ since by assumption $(\alpha_{e,j},\omega_{e,j})\neq (\alpha_{e,j'},\omega_{e,j'})$ for $j\neq j'$. Therefore, one only needs to assume $(\mathrm{H}_1)$ for electric polynomials $q_{e,j'}$ for which $\alpha_{e,j'}\geq 2 \, \omega_{e,j'} $.
The same properties hold for the magnetic polynomials $q_{m,\ell}$ with obvious changes.  
\end{Rem}
The electromagnetic energy is defined as 
		\begin{equation} \label{def_energyEM0}
	\mathcal{E}(t)\equiv \mathcal{E}(\mathbf{E}, \mathbf{H}, t):=\frac{1}{2}\Big(\varepsilon_0\, \int_{\R^3} |{\mathbf{E}(\bx,t)}|^2  \rmd \bx+\mu_0\, \int_{\R^3}|{\mathbf{H}(\bx,t)}|^2 \rmd \bx\Big).
\end{equation}
The present paper is the sequel of \cite{cas-jol-ros-22} in which we addressed the question of long time decay of $\mathcal{E}(t)$ with a (frequency dependent) Lyapunov function approach. 
We were able to prove some {\it polynomial stability} results under the so-called {\it strong dissipation assumption}: 
\begin{Def} [Strong Dissipation for Lorentz models] \label{StrongDissipation}
	\begin{equation} \label{SD} \forall \; 1 \leq j \leq N_e, \quad \alpha_{e,j} >0, \quad  
		\forall \; 1 \leq \ell \leq N_m, \quad \alpha_{m,\ell} >0. \end{equation} 
\end{Def}
\noindent By polynomial stability, we mean the energy decays more rapidly than a negative power of $t$:
\begin{equation} \label{Polynomial-Stability}
	\mathcal{E}(t) \leq C \; t^{-s} \;  \mbox{ for some } s >0.
\end{equation}
In this paper, we address the same question as in \cite{cas-jol-ros-22}  via a {\it modal approach}. Admittedly, this approach is technically more involved
than the one in \cite{cas-jol-ros-22}, but it presents many advantages:
\begin{itemize}
	\item it is by many aspects more {\it physical} (in particular it refers explicitly to the dispersion relation of the medium)
	thus more easily understandable by physicists,
	\item it leads to {\it optimal} results in the sense that upper bounds of the type \eqref{Polynomial-Stability} can be completed by corresponding lower bounds,
	\item it allows us to obtain polynomial stability results under a condition on the dissipation coefficients which is strictly weaker than \eqref{SD},  namely the 
	{\it weak dissipation} assumption:	
\begin{Def} [Weak Dissipation for Lorentz models] \label{WeakDissipation}
		\begin{equation} \label{WD} 
			\sum_{j=1}^{N_e} \alpha_{e,j}  + \sum_{\ell=1}^{N_m} \alpha_{m,\ell} > 0.
		\end{equation}
	\end{Def}
The reader will notice that the {\it strong dissipation} assumption imposes that all the coefficients $(\alpha_{e,j},\alpha_{m,\ell})$ are positive while 
the {\it weak dissipation} assumption means only that at least one of them is  positive.
\end{itemize}
In \cite{cas-jol-ros-22}, we made in Section 1.2 a rather extensive analysis of the literature addressing problems similar to the one we consider here, with application in viscoelasticity \cite{Daf-70,FabrizioMorro87,Dan-2015} or electromagnetism \cite{FabrizioMorro97,Riv-04}.  We discussed in particular various existing techniques for obtaining {\it polynomial} or {\it exponential stability} results (Sections 1.2 and 1.3 ). The closest works  to ours are the paper by Figotin-Schechter \cite{Fig-05} and  the papers  \cite{Nicaise2012,Nicaise2020,Nicaise2021}  by S. Nicaise and her collaborator C. Pignotti.  In \cite{Fig-05}, the authors  prove that the  electromagnetic energy of the solutions of a large class of linear dissipative models tends to zero but without addressing the question of the decay rate.  In  \cite{Nicaise2012,Nicaise2020,Nicaise2021}, the question of the decay rate of the solution is studied in detail in bounded domains for perfectly conducting materials. However, when specified for Maxwell's equations in generalized Lorentz media, the results of \cite{Fig-05} and \cite{Nicaise2012,Nicaise2020,Nicaise2021} only applied when the strong dissipation assumption \eqref{SD} is satisfied. Furthermore, under the weak dissipation assumption \eqref{WD}, we enlighten for some cases (when \eqref{SD} does not hold)  new  decay rates which are still  polynomial but with a smaller  exponent  in $1/t$  than the ones observed in the literature  \cite{Nicaise2012,Nicaise2020,Nicaise2021,cas-jol-ros-22}. \\ [12pt] 
The outline of the rest of this paper is as follows. First in Section  \ref{sec-schr-evolution-problem}, we rewrite the Maxwell's time-dependent system   as an evolution problem which involves a  maximal dissipative operator  $\bbA$.
Then in Section \ref{sec-main-results}, we state the main results of the article. In Section \ref{sec-reduction}, we show how to reduce the analysis of the evolution problem to the study of an infinity of Ordinary Differential Equations (ODE) systems in time parametrized by one real parameter. This is done by exploiting the fact that we work with constant coefficients which
allows us to use the Fourier transform in space and the isotropy of the model.  The equations are then parametrized  by $|{\bf k}|$, the modulus of the wave vector $\bf k$, namely the dual variable of the space variable. Each of these ODE systems involves a finite dimensional dissipative operator $\bbA_{|\bk|,\perp}$ and the rest of the analysis is devoted to estimates, for large $t$ of the exponentials
\begin{equation} \label{exponentials}
\rme^{-\rmi \bbA_{\bk,\perp }\, t} \mbox{ for } |\bk| \in \R^+,
\end{equation} 
which obviously depends on the eigenvalues of $\bbA_{|\bk|,\perp}$ and more precisely on their imaginary parts.
Section  \ref{sec-modal-analysis} is devoted to general spectral properties of $\bbA_{|\bk|,\perp}$ at fixed $|\bk|$: we  characterise its eigenvalues  as the   solutions of a rational equation parametrised by $|\bk|$, namely the dispersion relation of the medium,  whose analysis allows us 
 to provide a  diagonalizability criterion  for $\bbA_{|\bk|,\perp}$. Then, we need to distinguish is Sections \ref{sec.HF},   \ref{sec.LF} and  \ref{sec_mid-frequencies}  respectively the large, small and intermediate values of $|\bk|$. In Section \ref{sec.HF}, devoted to large $|\bk|$, we provide  asymptotic expansions (when $|\bk| \rightarrow + \infty $) of the solutions of dispersion relation, which shows the diagonalisability of  $\bbA_{|\bk|,\perp}$ for $|\bk|$ large enough.  After diagonalization of the exponential  \eqref{exponentials},  we obtain  optimal decay estimates in time of the high spatial frequency components (i. e. large values of $|\bk|$) of the solution of the evolution problem, based on  the behaviour  of the  imaginary part of the eigenvalues of $\bbA_{|\bk|,\perp}$ and uniform bounds (in $|\bk|$)  of the associated (non-orthogonal)  spectral projectors.  A similar strategy is adopted in Section \ref{sec.LF} to derive  optimal estimates for the low frequency components (i.e. small values of $|\bk|$) of the solution.  Section \ref{sec_mid-frequencies} deals with  estimates for intermediate values of $|\bk|$. The method  is quite different and somewhat less technical and precise than the method of Sections \ref{sec.HF} and \ref{sec.LF} but it is sufficient to show (in Section \ref{proof-general-cases}) that, the corresponding part of the solution decays exponentially fast in time. Sections  \ref{sec-final-step} and  \ref{sec-optim} are devoted to the proofs of our main results, after regrouping  the estimates of the three previous sections and deriving estimates in space variable norms via Plancherel's theorem: this is where the polynomial decay due to large and small values of $|\bk|$ is put in evidence. Moreover, we show that the corresponding  polynomial decay  estimates are optimal. In Section \ref{sec-extension-results},  we detail some extensions of these results. Finally, the appendix Section \ref{sec-appendix} contains technical results used for the spectral analysis of $\bbA_{|\bk|,\perp}$. In particular, in Section \ref{sec-asymptotic},  Lemma \ref{Lem-implicte-function} provides a useful  implicit function-like  result for some (scalar) nonlinear equations in the complex plane.   This lemma is used for the analysis of the  dispersion relation.  

\subsection{Maxwell's equations in dissipative generalized Lorentz media}\label{sec-schr-evolution-problem}
The (Cauchy) problem (\ref{planteamiento Lorentz},\ref{CI})  can be rewritten as an evolution problem
\begin{equation}\label{eq.schro}
\frac{\rmd \, \bU}{\rmd\, t} + \rmi\, \bbA \, \bU=0 \ \mbox{ with }  \ \bU(0)=\bU_0 \, \in \mathcal{H},
\end{equation}
where $\bbA$ is an unbounded operator on the Hilbert-space
\begin{equation}\label{eqdefHilbertspace}
\mathcal{H}:=\bL^2(\mathbb{R}^3)^N  =\bL^2(\mathbb{R}^3)\times \bL^2(\mathbb{R}^3) \times \bL^2(\mathbb{R}^3)^{N_e}
	\times \bL^2(\mathbb{R}^3)^{N_e} \times \bL^2(\mathbb{R}^3)^{N_m}  \times \bL^2(\mathbb{R}^3)^{N_m}
\end{equation}
where  $N=2+2\, (N_e+ N_m)$. $\mathcal{H}$ is
endowed by the following inner product: \\ [12pt]
for any $\bU=(\bE, \bH, \bP, \dot{\bP},\bM , \dot{\bM})\in \mathcal{H}$ and $\bU^{\prime}=(\bE^{\prime}, \bH^{\prime}, \bP^{\prime},\bM^{\prime}, \dot{\bP}^{\prime}, \dot{\bM}^{\prime})\in \mathcal{H}$, where $(\bP,\bM)$ is defined as in  \eqref{notPM}, (and the same for $(\bP^{\prime},\bM^{\prime})$, $(\dot{\bP}, \dot{\bM})$ and $(\dot{\bP}^{\prime}, \dot{\bM}^{\prime})$) 
\begin{equation} \label{defPS}
\left| \; \begin{array}{lllll}
(\bU, \bU')_{\mathcal{H}} &   =   &    \ds \frac{\varepsilon_0}{2}\ (\bE,\bE^{\prime})_{\bL^2}\,  + \, \frac{\mu_0}{2}\ (\bH,\bH^{\prime})_{\bL^2} \\ [12pt]  &  +  &   \displaystyle\frac{ \varepsilon_0}{2}  \sum_{j=1}^{N_e} \omega_{e,j}^2\, \Omega_{e,j}^2\, ( \bP_{j},\bP_{j}^{\prime})_{\bL^2} +
\frac{\varepsilon_0}{2}  \sum_{j=1}^{N_e} \Omega_{e,j}^2\, (\dot{\bP}_{j},\dot{\bP}_{j}^{\prime})_{\bL^2}  \\ [18pt]
&  +  &  \ds \frac{\mu_0}{2} \sum_{\ell=1}^{N_m} \omega_{m,\ell}^2\, \Omega_{m,\ell}^2\, (\bM_{\ell},\bM_{\ell}^{\prime})_{\bL^2}+ \frac{\mu_0}{2} \sum_{\ell=1}^{N_m}  \Omega_{m,\ell}^2\, ( \dot{\bM}_{\ell},\dot{\bM}_{\ell}^{\prime})_{\bL^2}\ ,
\end{array} \right.
\end{equation}\
where $\ds (\bu,\bv)_{\bL^2}=\int_{\mathbb{R}^3} \bu \cdot \overline{\bv} \,  \rmd \bx$ denotes the standard inner product of $\bL^2(\mathbb{R}^3)$ with $\cdot$  the operation \\[8pt]   defined on $\bbC^3 \times \bbC^3$ by  $\ba \cdot \bb :=\ba_1 \bb_1+\ba_2\bb_2+ \ba_3 \bb_3$.
\\ [12pt]
The operator  $\bbA: D(\bbA)\subset \mathcal{H}\to \mathcal{H}$  is defined by 
\begin{equation}\label{eq.defHamil}
\forall \; \bU =(\bE, \bH, \bP, \dot{\bP},\bM , \dot{\bM})\in D(\bbA), \quad 	\bbA \bU := -\rmi \; 	\begin{pmatrix}- \varepsilon_0^{-1} \nabla \times \bH+\sum \; \Omega_{e,j}^2\dot {\bP}_j \\[10pt] 
 \mu_0^{-1} \nabla \times \bE+ \sum \; \Omega_{m,\ell}^2 \dot{\bM}_{\ell}  \\[8pt]- \dot{\bP} \\[6pt]
  \alpha_{e,j} \dot{\bP}_j+\omega_{e,j}^2\,\bP_j-\bE\\[6pt] 
 - \dot{\bM} \\[6pt]  
 \alpha_{m,\ell}\, \dot{\bM}_\ell+\omega_{m,\ell}^2\,\bM_\ell-\bH
	\end{pmatrix} 
\end{equation}
with $D(\bbA)$ (dense in $\mathcal{H}$) being given by 
\begin{equation}\label{eq.defDomain}
D(\bbA):=H(\mbox{rot};\R^3) \times H(\mbox{rot};\R^3) \times \bL^2(\mathbb{R}^3)^{N_e} \times \bL^2(\mathbb{R}^3)^{N_m}  \times \bL^2(\mathbb{R}^3)^{N_m} , \end{equation}
where $H(\mbox{rot};\R^3):=\{\bu \in \bL^2(\mathbb{R}^3)\mid \nabla \times \bu \in  \bL^2(\mathbb{R}^3)  \}$. 
\begin{Rem}\label{Rem-not} Let us  point out  that, in the sequel, as in the formula \eqref{eq.defHamil}, we shall often   omit the summation bounds  $1$ and $N_e$ (resp. $1$ and $N_m$)  for sum on $j$ (resp. $\ell$) indices.\end{Rem}
\begin{Rem}\label{Rem-IC-Phys}
\noindent We notice that one can rewrite the evolution problem (\ref{planteamiento Lorentz}) with initial conditions satisfying (\ref{CI},\ref{CI-Phys}) as \eqref{eq.schro}  with  
\begin{equation}\label{eq.CI2}\bU_0=(\bE_0, \bH_0, 0, 0, 0,0)\in  \mathcal{H} \quad \mbox{ satisfying} \quad \nabla \cdot \bE_0=\nabla \cdot \bH_0=0.
\end{equation}
\end{Rem}

 \noindent We prove in Lemma  A.2 of \cite{cas-jol-ros-22} that  the operator $-\rmi \,\bbA$ is maximal dissipative.  Thus, it generates a contraction semi-group $\{\mathcal{S}(t)\}_{t\geq 0}$ of class $\mathcal{C}^0$ (see e.g. Theorem 4.3 page 14 of  \cite{Paz-83}). Hence, the Cauchy problem  \eqref{eq.schro} is well-posed and stable.
More precisely
(see propositions A.3 of \cite{cas-jol-ros-22}), it admits a unique mild solution $$\bU=\mathcal{S}(t)\bU_0 \in C^{0}(\R^+, \mathcal{H}),$$ which is a strong solution, 
$$\bU \in C^{1}(\R^+, \mathcal{H})\cap C^{0}(\R^+, \mathrm{D}(\bbA)),$$ 
as soon as $\bU_0\in D(\bbA)$, with  $D(\bbA)$ is endowed with the usual graph norm $\|\cdot\|_{D(\bbA)}$ defined by 
\begin{equation}\label{eq.normdomain}
\|\bU\|_{D(\bbA)}^2=\|\bU\|^2_{\mathcal{H}}+\|\bbA \bU\|^2_{\mathcal{H}}, \   \forall \; \bU \in D(\bbA).
\end{equation} 
\noindent
We introduce now the  well-known Hodge  decomposition in electromagnetism  which consists in writing any  $\bL^2(\mathbb{R}^3)$-fields as the  orthogonal sum of a gradient (i.e. a curl free field)  and a divergence free field (see  e.g. \cite{Dau-1990}).
Namely, one has
\begin{equation}\label{eq.hodgedecomp}
\bL^2(\mathbb{R}^3)= \nabla W^1(\mathbb{R}^3)  \overset{\perp}{\oplus} \bL^2_{\operatorname{div}, \, 0},
 \end{equation}
where $\nabla W^1(\mathbb{R}^3)  =\{ \nabla u \mid u \in  W^1(\mathbb{R}^3) \}$ with
$$W^1(\mathbb{R}^3)=\big\{u \in L^2_{loc}(\mathbb{R}^3) \mid (1+|\bx|^2)^{-1/2} u \in L^2(\mathbb{R}^3) \mbox{ and } \nabla u \in  \bL^2(\mathbb{R}^3) \big \}$$ the standard Beppo-Levy space on $\mathbb{R}^3$, and $\bL^2_{\operatorname{div}, \, 0} =\{ \bu \in  \bL^2(\mathbb{R}^3) \mid  \nabla \cdot \bu =0\}$ the space of $3D$ free divergence fields in $\mathbb{R}^3$. \\ [12pt]
 Following the decomposition \eqref{eq.hodgedecomp} for each copy $\bL^2(\mathbb{R}^3)$ of   $\mathcal{H}=\bL^2(\mathbb{R}^3)^N$,  one deduces immediately that the Hilbert $\mathcal{H}$ admits the following orthogonal decomposition:
\begin{equation}\label{eq.decompositionHilb}
 \mathcal{H} =\mathcal{H}_{\parallel} \overset{\perp}{\oplus}  \mathcal{H}_{\perp}
\mbox{ where }  \mathcal{H}_{\parallel} =\nabla W^1(\mathbb{R}^3)^N \mbox{ and }   \mathcal{H}_{\perp}=(\bL^2_{\operatorname{div},\, 0})^N.
\end{equation}
Let  us now introduce $D(\bbA_{\parallel})=D(\bbA)\cap \mathcal{H}_{\parallel} $ and $D(\bbA_{\perp})=D(\bbA)\cap \mathcal{H}_{\perp}$. 
As an immediate consequence of \eqref{eq.defHamil},  \eqref{eq.defDomain}, \eqref{eq.normdomain} and  \eqref{eq.decompositionHilb}, one has 
\begin{equation}\label{eq.stabAperp}
\bbA\big(D(\bbA_{\parallel})\big)\subset  \mathcal{H}_{\parallel},  \  \bbA\big(D(\bbA_{\perp}\big)\subset  \mathcal{H}_{\perp} \mbox{ and } D(\bbA)=D(\bbA_{\parallel})  \overset{\perp}{\oplus} D(\bbA_{\perp}).
\end{equation}
Thus one can reduce the operator $\mathbb{A}$ to a sum of two closed and densely defined operators $ \bbA_{\parallel}: D(\bbA_{\parallel}) \subset  \mathcal{H}_{\parallel} \mapsto   \mathcal{H}_{\parallel}$ and   $ \bbA_{\perp}: D(\bbA_{\perp}) \subset  \mathcal{H}_{\perp} \mapsto   \mathcal{H}_{\perp}$ in the sense that
\begin{equation}\label{eq.reducA}
\bbA= \bbA_{\parallel} \oplus  \bbA_{\perp} \   \mbox{ with } \,  \bbA_{\parallel} \bU= \bbA \bU, \, \forall \; \bU\in D(\bbA_{\parallel}) \mbox{ and } \, \bbA_{\perp} \bU= \bbA \bU, \,  \forall \; \bU\in D(\bbA_{\perp}) .
\end{equation}
Moreover, since any element of $\mathcal{H}_{\parallel}$ is made of curl free vector fields, it is readily seen on \eqref{eq.defHamil}  and \eqref{eq.defDomain} that $D(\bbA_{\parallel})= \mathcal{H}_{\parallel}$ and that $\bbA_{\parallel}$ is bounded. \\ [12pt]
\noindent From \eqref{eq.reducA}, we deduce  using  the Hille-Yosida approximation of $\mathbb{A}$ (see Corollary 3.5 page 11 of  \cite{Paz-83}) that  a similar reduction holds for the semi-group $\{\mathcal{S}(t)\}_{t\geq 0}$. Namely, $ \bbA_{\parallel}$ and $ \bbA_{\perp}$ generate two contraction semi-groups $\{\mathcal{S}_{\parallel}(t)\}_{t\geq 0}$  and $\{\mathcal{S}_{\perp}(t)\}_{t\geq 0}$ on the Hilbert spaces $\mathcal{H}_{\parallel}$ and $\mathcal{H}_{\perp}$ such that for all $t\geq 0$:
\begin{equation}\label{eq.reducsemigroup}
S(t)=S_{\parallel}(t) \oplus  \mathcal{S}_{\perp}(t), \ S_{\parallel}(t)  \bU_0=S(t) \bU_0 , \, \forall \; \bU_0 \in  \mathcal{H}_{\parallel}, \  S_{\perp}(t)  \bU_0=S(t) \bU_0 , \, \forall \; \bU_0 \in  \mathcal{H}_{\perp}  .
\end{equation}
In more physical terms, it means that if the components of the initial conditions $\bU_0\in \mathcal{H}$ are curl free (i.e. if $\bU_0\in \mathcal{H}_{\parallel}$),  then the solution $\bU(t)$ of \eqref{eq.schro} remains curl free (i.e.  $\bU(t)\in \mathcal{H}_{\parallel}$ for all $t\geq 0$). In this case, the corresponding dynamics is trivial 
since it corresponds to a damped harmonic oscillator in finite dimension (these are  non propagative solutions of the Maxwell's equations \eqref{planteamiento Lorentz}).  \\ [12pt]
Oppositely, the interesting case is when $\bU_0\in   \mathcal{H}_{\perp} $, in other words when the  initial conditions are divergence free. Then, the solution $\bU(t)$ of \eqref{eq.schro} is propagative and  by \eqref{eq.reducsemigroup}, it satisfies  \begin{equation*}\bU(t)=S(t)\bU_0= S_{\perp}(t)  \bU_0 \in   \mathcal{H}_{\perp} , \quad \forall \; t\geq 0.
\end{equation*}
In other words, the solution $\bU(t)$ remains  divergence free for all positive time.
\subsection{Statement of the main results}\label{sec-main-results}
\subsubsection{The long-time energy decay rate}
To expose our main result (Theorem \ref{thm_Lorentz}), we first need to define  an adapted functional framework:
\begin{enumerate}
\item  For each  $p\geq0$, we introduce the Banach space
\begin{equation} \label{delLcalp}
	{\cal  L}_p (\R^3) =  \big \{ v \in \mathcal{S}'(\R^3) \; /  \; |\bk|^{-p} (1+|\bk|^{p}) \; \widehat v \in L^\infty(\R^3) \big \}
\end{equation} 
 (where $\bk \mapsto \widehat v(\bk)$ denotes here the Fourier transform of a scalar tempered distributions  $v\in\mathcal{S}'(\R^3)$)  equipped  with the norm
\begin{equation} \label{normcalp}
		\|v\|_{{\cal  L}_p} := \big\| \,  |\bk|^{-p} (1+|\bk|^{p}) \; \widehat v  \, \big\|_{L^\infty(\R^3)} .
\end{equation}
We point out that if $v\in{\cal  L}_p (\R^3)$, then one has  $|\widehat{v}(\bk)|\leq  \|v\|_{{\cal  L}_p}  |\bk|^p$ for a.e. $\bk\in \R^3$. The reader will also note that  for $q > p$, ${\cal  L}_q(\R^3) \subset {\cal  L}_p(\R^3)$.\\[12pt]
 Identifying a function in ${\cal  L}_p (\R^3)$ is not  trivial since it requires to look at the same time at the function and its Fourier transform. See however Remark \ref{remcalLp} for some more concrete examples of spaces of functions included in ${\cal  L}_p (\R^3)$. \\ [6pt]
We shall introduce $\boldsymbol{\cal  L}_p (\R^3)={\cal  L}_p (\R^3)^3$ the vector valued  version of ${\cal  L}_p (\R^3)$  endowed with the norm	$\| \cdot \|_{\boldsymbol{\cal  L}_p}$     defined by \\
$$
\forall \; \bu=(u_1,u_2,u_3)\in \boldsymbol{\cal  L}_p (\R^3), \quad \| \bu \|_{\boldsymbol{\cal  L}_p}^2=\|u_1\|_{{\cal  L}_p}^2 + \|u_2\|_{{\cal  L}_p}^2 +\|u_3\|_{{\cal  L}_p} ^2.
$$
\\
\noindent Then, we define on  $\boldsymbol{\cal  L}_p (\R^3)^N$   the  norm  $\| \cdot \|_{\boldsymbol{\cal  L}_p^N}$ by:  $ \; \forall \; \bU=(\bE, \bH, \bP, \dot{\bP},  \bM, \dot{\bM}) \in  \boldsymbol{\cal  L}_p (\R^3)^N$,
\begin{equation*} \label{Lnorm} 
\|\mathbf{U} \big\|_{\boldsymbol{\cal  L}_p^N}^2=\ds \|\bE\|_{\boldsymbol{\cal  L}_p}^2+\|\bH\|_{\boldsymbol{\cal  L}_p}^2 +  \ds \sum_{j} \big(\|\bP_{j}\|_{\boldsymbol{\cal  L}_p}^2+ \|\dot{\bP}_{j}\|^2_{\boldsymbol{\cal  L}_p}\big)
 +  \ds \sum_{\ell}\big(\|\bM_{ \ell} \|^2_{\boldsymbol{\cal  L}_p}+ \|\dot{\bM}_{\ell}\|^2_{\boldsymbol{\cal  L}_p}\big ).
\end{equation*}
\item For  $m \geq 0$, we introduce the space $\bH^m(\R^3)=H^m(\R^3)^3$ where $H^m(\R^3)$ is the standard Sobolev space defined on scalar functions, endowed with the norm $\|\cdot\|_{\bH^m}$ given by
$$ \forall \; \bu=(u_1,u_2,u_3)\in \bH^m(\mathbb{R}^3), \quad \|\bu\|_{\bH^m(\R^3)}^2=\|u_1\|_{H^m(\R^3)}^2+ \|u_2\|_{H^m(\R^3)}^2+\|u_3\|_{H^m(\R^3)}^2,$$ with $\|\cdot\|_{H^m(\R^3)}$  the usual Sobolev norm on $H^m(\R^3)$ (see e.g.  \cite{Dau-1990} page 500).
Then, one defines  the norm $\|\cdot\|_{\bH^m(\bbR^3)^N}$ for any $\bU=(\bE, \bH, \bP, \dot{\bP},  \bM, \dot{\bM}) \in \bH^m(\bbR^3)^N$  by
\begin{equation} \label{Hmnorm} \left| \begin{array}{lll}
\|\mathbf{U} \big\|_{\bH^m(\R^3)^N}^2&=&\ds \|\bE\|^2_{\bH^m(\bbR^3)}+\|\bH\|^2_{\bH^m(\bbR^3)} \\ [12pt] & + & \ds \sum_{j} \big(\|\bP_{j}\|^2_{\bH^m(\bbR^3)}+ \|\dot{\bP}_{j}\|^2_{\bH^m(\bbR^3)}\big)\\ [15pt]
& + & \ds \sum_{\ell}\big(\|\bM_{ \ell} \|^2_{\bH^m(\bbR^3)}+ \|\dot{\bM}_{\ell}\|^2_{\bH^m(\bbR^3)} \big ).
\end{array} \right. 
\end{equation}
\end{enumerate}
\begin{Rem} \label{remcalLp} [~On the space ${\cal  L}_p (\R^3)$~]
	A first trivial observation is that for any $p\geq 0$, ${\cal  L}_p (\R^3)$ contains all integrable functions whose Fourier transform is supported outside a ball centered at the origin.\\ [12pt]
	Moreover, for $p \in \N$,  let us introduce the weighted $L^1$ space (already used in \cite{cas-jol-ros-22}):
		\begin{equation} \label{spacesL1weighted} 
		\displaystyle L^1_{p}(\R^3) := \big\{ v \in L^1(\R^3) \; / \; (1+|\bx|)^p \, v \in L^1(\R^3) \big\} \quad ( L^1_{0}(\R^3)= L^1(\R^3)),
	\end{equation}
endowed with the norm
\begin{equation} \label{normL1p} 
		\|v\big\|_{L^1_{p}(\R^3)} :=  \big \|(1+|\bx|)^p \,v  \big\|_{L^1(\R^3)}.
\end{equation}
In particular, functions of $L^1_{p}(\R^3)$ have existing moments up to order $p$. Let us introduce, for $p \geq 1$,   the closed subspace of $L^1_{p}(\R^3)$ of functions whose moments up to order $p-1$ vanish:
	\begin{equation} \label{spacesL1weighted2} 
		\displaystyle L^1_{p,0}(\R^3) := \big\{ u \in L^1_{p}(\R^3) \; / \; \forall \; |\alpha| \leq p-1, \; \int \bx^\alpha \, u \, d{\bx} = 0  \big\}, \
	\end{equation}
where $\alpha = (\alpha_1, \alpha_2, \alpha_3)\in \mathbb{N}^3$ denotes a multi-index with ``length" $|\alpha| = \alpha_1+ \alpha_2+ \alpha_3$ and  where $\bx^\alpha:=\bx_1^{\alpha_1}\, \bx_2^{\alpha_1}\, \bx_3^{\alpha_1}$ for all $\bx=(x_1,x_2,x_3)\in \R^3$. Moreover, by convention, $ L^1_{0,0}(\R^3) := L^1(\R^3)$.\\[6pt]
\noindent Note that since any function $v$ in $L^1_{p}(\R^3)$ belongs in particular to $L^1(\R^3)$, its Fourier transform  $\widehat v(\bk)$ is well defined and belongs to $C^0_0(\R^3)$, namely the closed subspace of $L^\infty(\R^3)$ made of continuous functions that tend to $0$ at infinity. We claim that 
	\begin{equation} \label{inclusion} 
\forall \; p \in \N, 	\quad \displaystyle L^1_{p,0}(\R^3) \subset {\cal  L}_p .
\end{equation}
Indeed, from well-known properties of the Fourier transform, for all $v \in L^1_{p,0}(\R^3)$, one has 
	\begin{equation} \label{L1pFourier}
  \partial^\alpha \widehat v(\bk) \in C^0_0(\R^3) \mbox{ for }  |\alpha| \leq p \quad \mbox{ and, if $p>0$, } \quad  \partial^\alpha \widehat v(0) = 0, \; \; \mbox{for }  |\alpha| \leq p-1.
\end{equation} 
Moreover, from  the Taylor expansion of $\widehat v(\bk)$ at $0$, one  gets (by the Taylor-Lagrange formula):
\begin{equation}\label{eq.taylorlagrange}
\exists\; C(p)>0 \ \mid  \  |\widehat v(\bk)| \leq C(p) \, \|v\big\|_{L^1_{p}(\R^3)} \; |\bk|^p, \quad \, \forall \; \bk\in \R^3 \ \mbox{ and } \  \forall \; v \in L^1_{p,0}(\R^3) .
\end{equation}
Thus, using the fact that $ \|\widehat v\|_{L^{\infty}(\R^3)}\leq \|v\|_{L^1(\R^3)}\leq \|v\|_{L^1_p(\R^3)} $, one deduces with \eqref{eq.taylorlagrange} that
\begin{equation}\label{eq.estimcontinjec}
      |\widehat v(\bk)| \lesssim    \min(1,  |\bk|^{p})  \, \|v\big\|_{L^1_{p}(\R^3)}, \quad \, \forall\; \bk\in \R^3  \ \mbox{ and } \  \forall \; v \in L^1_{p,0}(\R^3) .
\end{equation}
Then, it follows from \eqref{normcalp}  and \eqref{eq.estimcontinjec} that the embedding \eqref{inclusion}  is continuous, more precisely:
\begin{equation} \label{eqnorm}
 \exists\; C(p)>0,\ \forall \; v \in L^1_{p,0}(\R^3),\quad 	 \| v \|_{{\cal  L}_p} \leq C(p) \,  \|v\big\|_{L^1_{p}(\R^3)} .
	\end{equation}
It also follows from \eqref{inclusion} that, in particular, ${\cal  L}_p$ contains all compactly supported integrable functions (in the $\bx$ variable) whose moments up to order $p-1$ vanish. 
\end{Rem}	
Next, we introduce the notion  of critical configuration  associated to the Maxwell's system \eqref{planteamiento Lorentz}, which will influence the long-time decay rate of the electromagnetic energy. It will be enlightened in our proof via the analysis  of the dispersion curves   associated to this system in the high frequency regime (see Section \ref{sec.HFdispcurv}).
\begin{Def} \label{Critical_cases}
	We say that   the Maxwell's system \eqref{planteamiento Lorentz} is in a critical  configuration if  the weak dissipation condition \eqref{WD}  holds and if one of the following conditions is satisfied:
	\begin{enumerate}
		\item   $\forall \,  \ell \in \{ 1, \ldots, N_m \} , \alpha_{m,\ell}=0$ and  $\exists \,j \in \{ 1, \ldots, N_e\} \mid \alpha_{e,j}=0\,   \mbox{  and } \omega_{e,j} \notin \big\{ \omega_{m,\ell} \big\}.$
		\item $\forall \,  j \in \{ 1, \ldots, N_e \} , \alpha_{e,j}=0$ and  $\exists \, \ell \in \{ 1, \ldots, N_m\} \mid \alpha_{m,\ell}=0\,   \mbox{  and } \omega_{m,\ell}  \notin \big\{ \omega_{e,j} \big\}.$
	\end{enumerate}
When none of these condition holds,  the system \eqref{planteamiento Lorentz} is in a  non-critical configuration.
\end{Def}
\noindent We point out that, as the reader will easily check, under the weak dissipation assumption \eqref{WD}, the critical configuration can occur only if $N_e\geq 2$ or $N_m \geq 2$.
\\ [12pt]
 We are now able to state our main result  under the weak dissipation assumption concerning the decay of the total $\|\bU(t)\|^2_{\mathcal{H}}$ associated to the   evolution problem \eqref{eq.schro} with divergence-free initial conditions $\bU_0\in\mathcal{H}_{\perp}$ which contain  in particular  the initial conditions given by (\ref{CI},\ref{CI-Phys}). Thus, this result applies  to the Maxwell's system  \eqref{planteamiento Lorentz} with initial conditions  \eqref{CI} in $\mathcal{H}_{\perp}$. \\ [12pt]
We shall use below, as well as in the rest of this paper,  the following 
\begin{Not}\noindent To compare two non-negative functions $f(y)$ and $g(y)$,  where $y \in Y$ and $y = \bx, \bk,t$ or any combination of the variables, we introduce the notation:
\begin{equation*} \label{notation2}
	f \lesssim g  \ \Longleftrightarrow \ \exists\;  C > 0  \quad  \mid \  f(y) \leq C \;  g(y), \quad  \forall \; y \in Y,
\end{equation*} 
where the constant  $C$ depends only on $N_e$, $N_m$ and  the coefficients  of  the system \eqref{planteamiento Lorentz}. 
\end{Not} 
\begin{Thm} \label{thm_Lorentz}
Let assume  that  the Maxwell's system  \eqref{planteamiento Lorentz} satisfies  the weak dissipation assumption \eqref{WD} and the irreducibility assumptions $(\mathrm{H}_1)$ and $(\mathrm{H}_2)$. Then, for any initial condition $\bU_0\in  \mathcal{H}_{\perp}$, one has 
	\begin{equation} \label{convergenceL}
	\lim_{t \rightarrow + \infty} \|\bU(t)\|^2_{\mathcal{H}}= 0.
	\end{equation}
	Moreover if for some real numbers $m >0$ and $p \geq 0$,  $\bU_0 \in  \bH^m(\R^3)^N \cap   \boldsymbol{\cal L}_{p}^N  \cap  \mathcal{H}_{\perp}$
then the decay rate of $\|\bU(t)\|^2$ is polynomial. However,  the associated  exponents in $1/t$ depend  on   the configuration (see definition \ref{Critical_cases}) of the  corresponding Maxwell's system  \eqref{planteamiento Lorentz}. More precisely, 
\begin{enumerate}
\item If the Maxwell system is in a non-critical configuration, then one has 
	\begin{equation} \label{polynomial_decayncr}
	\|\bU(t)\|^2_{\mathcal{H}} \leq \frac{C_{\mathrm{HF}}^m(\bU_0)}{t^m} + \frac{C_{\mathrm{LF}}^p(\bU_0)}{t^{p+\frac{3}{2}}},  \quad  \forall \; t>0,
	\end{equation}
	\item If the Maxwell system is in a critical configuration, then one has
	\begin{equation} \label{polynomial_decaycr}
	\|\bU(t)\|^2_{\mathcal{H}} \leq \frac{C_{\mathrm{HF}}^m(\bU_0)}{t^\frac{m}{2}} + \frac{C_{\mathrm{LF}}^p(\bU_0)}{t^{p+\frac{3}{2}}} ,  \quad  \forall \; t>0,
	\end{equation}
\end{enumerate}
	where in \eqref{polynomial_decayncr} and \eqref{polynomial_decaycr} the constants satisfy
	\begin{equation}\label{eq.constantHF-LF}
	C_{\mathrm{HF}}^m(\bU_0)\lesssim \|\bU_0 \|_{ \bH^m(\R^3)^N}^2  \quad \mbox{ and } \quad  C_{\mathrm{LF}}^p(\bU_0)\lesssim \|\bU_0 \|_{ \boldsymbol{\cal L}_{p}^N}^2.
	\end{equation}
\end{Thm}
\begin{Rem}\label{Rem-IC-Phys2} \noindent When  the initial data $\bU_0$ satisfy also \eqref{CI-Phys}, the estimates \eqref{polynomial_decayncr} and\eqref{polynomial_decaycr} give  estimates involving the electromagnetic energy only since $\mathcal{E}(t) \leq  \|\bU(t) \|_{\cal H}^2$ while $\mathcal{E}(0) =  \|\bU_0 \|_{ \cal H}^2$.
\end{Rem}
\begin{Rem}\label{Rem.HF-LF}
The energy $\|\bU(t)\|^2_{\mathcal{H}}$  is dominated in \eqref{polynomial_decayncr} and \eqref{polynomial_decaycr}  by a sum of two terms:
\begin{itemize}
\item
The first one involves the Sobolev regularity of the initial condition and is  linked, as it will appear in the proof, to  the decay of the high spatial frequency components of the energy. 
\item The second is related to the condition $\bU_0\in  \boldsymbol{\cal L}_{p}^N$ and thus to the decay of the low spatial frequency  components of the energy.
\end{itemize} The presence of the second term is directly  related to the fact that the domain of propagation  is unbounded and  does not appear in the control of the energy  in bounded domains  for perfectly conducting materials  (see, \cite{Nicaise2012,Nicaise2020,Nicaise2021}). 
\end{Rem}
\begin{Rem} For generalized Lorentz media,  this Theorem   generalizes under the weak  dissipation assumption \eqref{WD} the results  obtained under  the more restrictive strong dissipation assumption \eqref{SD} in \cite{Nicaise2012,Nicaise2020,Nicaise2021} for bounded domains  and in  \cite{cas-jol-ros-22} for the whole space $\mathbb{R}^3$.\\ [8pt] As \eqref{SD} excludes the  critical configuration, the weaker decay   \eqref{polynomial_decaycr} could not be observed in \cite{cas-jol-ros-22,Nicaise2012,Nicaise2020,Nicaise2021}.
\end{Rem}

\subsubsection{Optimality of the bounds}
In Theorem \ref{thm_Lorentz}, we prove  a  long-time polynomial decay rate for $\|\bU(t)\|$ (cf. \eqref{polynomial_decayncr}   and \eqref{polynomial_decaycr})  for solutions  $\bU(t)$  of the Cauchy problem  \eqref{eq.schro}  which satisfies divergence free initial conditions 
$
\bU_0 \in \mathcal{H}_{\perp}.
$
In this section, we analyse the question of the optimality of this  decay rate.\\

\noindent The estimates \eqref{polynomial_decayncr} and \eqref{polynomial_decaycr} involved a sum of two terms which are related to the decay of the   low and high spatial frequency components of the solution (see Remark \ref{Rem.HF-LF}).  Thus,  to speak about optimal polynomial exponents, we need to decouple these two terms by separating the high and low frequency behaviors of the solution $\bU$.  Therefore,   we introduce  the 3D spatial Fourier transform defined  by
$$
\mathbb{G}(\bk)=\mathcal{F}(\bG)(\bk)=\frac{1}{(2\pi)^{\frac{3}{2}} } \int_{\R^3}\bG(\bx)\,  \mathrm{e}^{-\rmi \bk \cdot \bx} \,\rmd \bx  \quad \forall \; \bG \in \bL^1(\R^3) \cap  \bL^2(\R^3),
$$
which extends to a unitary transformation from $\bL^2(\bbR^3_{\bx})$ to $\bL^2(\bbR^3_{\bk}).$ Applying $\mathcal{F}$ to each  copy $\bL^2(\mathbb{R}^3)$ of   $\mathcal{H}=\bL^2(\mathbb{R}^3)^N$ defines a unitary transform, still denoted $\mathcal{F}$, in $\mathcal{H}$.\\[6pt]
 \noindent We now  introduce  for $(m,p)\in  \mathbb{R}^{+,*}\times  \mathbb{R}^+ $ the spaces 
\begin{eqnarray}
\mathcal{H}^{m}_{\perp,\operatorname{HF}}&=&\{\bV\in  \mathcal{H}_{\perp} \cap \bH^m(\bbR^3)^N \mid \exists \; k_+>0 \mid   \operatorname{supp}(\mathcal{F}(\bV))\subset \bbR^3 \setminus B(0, k_+)  \}, \label{eq.HperpHF}\\[4pt]
\mathcal{H}^{p}_{\perp,\operatorname{LF}}&=&\{\bV\in  \mathcal{H}_{\perp}  \cap \boldsymbol{\cal  L}_p (\R^3)^N   \mid \exists \; k_->0 \mid   \operatorname{supp}(\mathcal{F}(\bV))\subset \overline{  B(0, k_-)}  \} . \label{eq.HperpLF}
 \end{eqnarray}
In the spirit of \cite{Nicaise2012}, we define  for any $m>0$  the high frequency optimal  polynomial exponent  of  solutions $\bU(t)$  of the Cauchy problem  \eqref{eq.schro}  with initial conditions $\bU_0\in \mathcal{H}^{m}_{\perp,\operatorname{HF}} $ as follows:
\begin{equation}\label{eq.gammaHF}
\gamma_m^{\operatorname{HF}}=\sup \{\gamma \in \bbR^+ \mid   \forall \; \bU_0\in \mathcal{H}^{m}_{\perp,\mathrm{HF}},\, \exists \; C(\bU_0)>0 \mid \|\bU(t)\|^2_{\mathcal{H}} \ \lesssim \frac{C(\bU_0)}{t^{\gamma}}, \quad \forall t\geq1 \}.
\end{equation}
Similarly, we define for any  $p\geq 0$  the low frequency optimal  polynomial  exponent:
\begin{equation}\label{eq.gammaLF}
\gamma_p^{LF}=\sup \{\gamma \in \mathbb{R}^+ \mid   \forall \; \bU_0\in \mathcal{H}^{p}_{\perp,\mathrm{LF}},\, \exists  \; C(\bU_0)>0 \mid  \|\bU(t)\|^2_{\mathcal{H}} \ \lesssim \frac{C(\bU_0)}{t^{\gamma}}, \quad \forall t\geq1 \}.
\end{equation}
We are now able to state our  result concerning the low frequency and high  frequency optimal polynomial  decay rates.
\begin{Thm}\label{thm-optm}
Let $ (m,p)\in   \mathbb{R}^{+,*}\times  \mathbb{R}^+$.  The  exponents $\gamma_m^{\operatorname{HF}}$ and  $\gamma_p^{LF}$  are given as follows:
\begin{enumerate}
\item If the Maxwell system is in a non-critical configuration, $\gamma_m^{\operatorname{HF}}=m$ and $\gamma_p^{LF}=p+3/2$.
\item If the Maxwell system is in a critical configuration, $\gamma_m^{\operatorname{HF}}=m/2$ and $\gamma_p^{LF}=p+3/2$.
\end{enumerate}
\end{Thm}
\noindent  We point out that in Section \ref{sec-final-step}, we  show that the upperbounds \eqref{polynomial_decayncr} and \eqref{polynomial_decaycr} give the right lower bounds on the exponents $\gamma_m^{\operatorname{HF}}$ and  $\gamma_p^{LF}$. In Section \ref{sec-optim}, we show that these lower-bounds are  also upper-bounds and thus deduce the value of these exponents.
\section{Fourier reduction}\label{sec-reduction}
\subsection{The reduced operator $\bbA_{\bk}$}
The homogeneity of the propagation medium allows us to reduce the spectral  analysis of the operator $\bbA$ defined in (\ref{eq.defHamil}) to the spectral analysis  of a family  (non-self-adjoint) operators $(\bbA_{\bk})_{\bk\in \bbR^3}$  on a finite dimensional space.\\

\noindent We introduce the space $\bC^N$ with $\bC=\bbC^{3}$ endowed  with the inner product  $(\cdot ,\cdot)_{\bC^N}$ defined via  the expression \eqref{defPS}  by  replacing $\bU\in \mathcal{H}$
by    $\bbU=(\bbE,\bbH,\bbP,\dot{\bbP}, \bbM,\dot{\bbM})\in \bC^N$ with  
$$\bbP=(\bbP_j),  \dot{\bbP}=(\dot{\bbP}_j),\bbM=(\bbM_\ell),\dot{\bbM}=(\dot{\bbM}_\ell),$$  
(the same for $\bU'$) and the ${\bL^2}$ inner product $(\cdot,\cdot)_{\bL^2}$ by the usual one in  $\bC$: $(\bbE,\bbE')=\bbE\cdot \overline{\bbE}'$. \\ [10pt] The corresponding norms in $\bC$ or  $\bC^N$ are both denoted $|\cdot|$ for simplicity.\\[6pt]
\noindent In this functional framework,   $\bbA$ is unitarily equivalent via $\mathcal{F}$ to a direct  integral of operators $\bbA_{\bk}$ defined  on the finite dimension space $\bC^N$. 
Namely, one has:
\begin{equation}\label{eq.intdirecA}
\bbA=\mathcal{F}^{*}\, \Big( \int_{\bbR^3}^{\oplus} \bbA_{\bk} \, \rmd \bk \Big)  \mathcal{F} \   \mbox{ i.e. } \  \mathcal{F}  \big(\bbA\, \bU\big)(\bk)= \bbA_{\bk} \, \mathcal{F} \bU(\bk), \ \forall \, \bU\in D(\bbA)\mbox{ and a.e. } \bk \in \bbR^3,
\end{equation}
where the (bounded) linear operator $\bbA_{\bk}:  \bC^N \mapsto  \bC^N $ is  given by 
\begin{equation}\label{eq.defAk}
\forall \; \bbU =(\bbE, \bbH, \bbP, \dot{\bbP},\bbM , \dot{\bbM})\in \bC^N , \quad  \bbA_\bk\bbU=
\begin{pmatrix}
\displaystyle -\frac{\bk \times \bbH}{\varepsilon_0}- \rmi \sum \Omega_{e,j}^2 \,\dot{\bbP}_j \\[8pt]
\displaystyle \frac{\bk \times \bbE}{\mu_0}- \rmi \sum\Omega_{m,\ell}^2\, \dot{\bbM}_\ell\\[10pt]
\rmi \,  \dot {\bbP}\\[6pt]
- \rmi \, \alpha_{e,j} \dot{\bbP}_j - \rmi \omega_{e,j}^2 \bbP_j+\rmi \, \bbE 
\\[7pt]
\rmi \, \dot{ \bbM} \\[6pt]
- \rmi \, \alpha_{m,l} \dot{\bbM}_\ell - \rmi \omega_{m,l}^2 \bbM_\ell+\rmi\, \bbH\quad 
\end{pmatrix}.
\end{equation}
We point out that $\bbA_\bk$ is deduced  from   the definition \eqref{eq.defHamil} of $\bbA$ by replacing  the curl operator $\nabla \times$ operator by its spatial Fourier counterpart $\rmi\, \bk \times$. The  usual cross product on $\mathbb{R}^3\times \mathbb{R}^3$ has been here extended to  $\bC\times \bC$ via the formula $\ba\times \bb=(\ba_2 \bb_3-\ba_3 \bb_2,  \ba_3 \, \bb_1- \ba_1 \,\bb_3,\ba_1 \, \bb_2-\ba_2\bb_1 )$.

\begin{Rem}\label{rem.Disp}
A simple computation shows that for any $\bbU =(\bbE, \bbH, \bbP, \dot{\bbP},\bbM , \dot{\bbM})\in \bC^N$:
$$\operatorname{Im}(\bbA_{\bk}\bbU,\bbU)_{\bC^N}= -\sum_{j=1}^{N_e} \frac{ \varepsilon_0}{2} \, \alpha_{e,j} \, \Omega_{e,j}^2  \, |\dot{\bbP}_j|^2 - \sum_{\ell=1}^{N_m} \frac{\mu_0}{2} \, \alpha_{m,\ell} \, \Omega_{m,\ell}^2  \, |\dot{\bbM}_{\ell}|^2\leq 0.$$
Thus, for all $\bk \in \bbR^3$, the spectrum of  $\bbA_{\bk}$, $\sigma(\bbA_{\bk})$ is included in $ \overline{\bbC^-}=\{\omega \in \mathbb{C}\mid \operatorname{Im}(\omega)\leq0\}$. \\ [10pt] Hence, 
  $- \, \rmi \, \bbA(\bk)$ is maximal dissipative and the  resolvant  of $\bbA(\bk)$ , $R_{\bbA_{\bk}}(\omega)=(\bbA_{\bk}-\omega \mathrm{I}d)^{-1}$ is well-defined in $\bbC^+$  and  satisfies: 
$$\| R_{\bbA_{\bk}}(\omega) \| \leq \operatorname{Im}(\omega)^{-1}, \ \forall \, \omega \in \bbC^+ \mbox{ and } \forall \, \bk \in \bbR^3.$$ This estimate   justifies the existence of the direct integral   $\ds \int_{\bk\in \mathbb{R}^3}^{\oplus} \bbA_{\bk} \, \rmd \bk $ (see  \cite{Ng-20}  page 5).
\end{Rem}

\noindent From an operator point of view, one deduces  from \eqref{eq.intdirecA} (see e.g. Theorem  4.2 of  \cite{Ng-20}) that the contractions  $S(t)$ of the semigroup $\{\mathcal{S}(t)\}_{t\geq 0}$  are also unitarily equivalent via $\mathcal{F}$ to a direct  integral of contractions  $\rme^{-\rmi \bbA_{\bk}t}$  on  $\bC^N$. More precisely, one has for any  $t\geq 0$:
\begin{equation}\label{eq.intdirecS}
 \hspace{-0.37cm} S(t)=\mathcal{F}^{*}\, \Big( \int_{\bbR^3}^{\oplus}\rme^{-\rmi \bbA_{\bk}t} \, \rmd \bk \Big)  \mathcal{F},    \mbox{ i.e. }    \mathcal{F}  \big(S(t) \bU\big)(\bk)=\rme^{-\rmi \bbA_{\bk}t} \, \mathcal{F} \bU(\bk), \, \forall \, \bU\in \mathcal{H}  \mbox{  a.e.}\, \bk \in \bbR^3. \hspace{-0.23cm} 
\end{equation}
From a more practical point of view, it means that the Fourier transform $\bbU=\mathcal{F}\bU$ of the solution $\bU$ of  the  evolution problem  \eqref{eq.schro} with initial condition $\bU_0\in \mathcal{H}$ satisfies a family of Ordinary Differential Equations (ODE)  in $\bC^N$ parametrized by  the wave vector $\bk$. Namely, one has for a.e. $\bk\in \bbR^3$:
\begin{equation}\label{eq.cauchypb}
 \frac{ \rmd \bbU(\bk,t)}{\rmd  t} + \rmi \, \bbA_\bk \bbU(\bk,t)=0 \quad  \mbox{ for } t\geq 0    \ \mbox{ with }\
\bbU(\bk,0)=\bbU_0(\bk)=(\mathcal{F} \bU_0)(\bk).
\end{equation}
The solution of each these ODE's is given by
$
\bbU(\bk,t)=\rme^{-\rmi \bbA_{\bk}t} \bbU_0(\bk),  \ \, \forall \, t\geq 0.
$
We can now introduce the counterpart of the Hodge decomposition  \eqref{eq.decompositionHilb} in the  (spatial) frequency domain. In this perspective, one decomposes orthogonally the space $\bC^N$,  for all wave vector  $\bk \neq 0$, as
\begin{equation}\label{eq.space}
\bC^N= \bC_{ \bk, \parallel}^N \overset{\perp}{\oplus}  \bC_{ \bk, \perp}^N \mbox{ where } \ \bC_{ \bk, \parallel}:=\operatorname{span} \bk  \ \mbox{ and } \ \bC_{ \bk, \perp}:=\{ \bk \}^{\perp}.
\end{equation}
One shows easily that $ \bC_{ \bk, \parallel}^N$ and $\bC_{\perp, \bk}^N$ are stable by $\bbA_{\bk}$.  Thus, one can reduce $\bbA_{\bk}$ as a sum of two operators  $$\bbA_{\bk,\parallel}: \bC_{ \bk, \parallel}^N \to \bC_{ \bk, \parallel}^N \ \mbox{  and } \ \bbA_{\bk, \perp}: \bC_{ \bk, \perp}^N  \to \bC_{ \bk, \perp}^N$$
in such way that
\begin{equation}\label{eq.reducAk}
\bbA_\bk= \bbA_{\bk,\parallel} \oplus  \bbA_{\bk,\perp} \   \mbox{ with } \,  \bbA_{\bk,\parallel} \bU= \bbA_{\bk} \bU, \, \forall \; \bU\in \bC_{ \bk, \parallel}^N \mbox{ and } \, \bbA_{\bk,\perp} \bU= \bbA_{\bk} \bU, \,  \forall \; \bU\in \bC_{ \bk, \perp}^N.
\end{equation}
 The above relation is nothing but  the counterpart of \eqref{eq.reducA} for Fourier components. Indeed, following the decompositions  \eqref{eq.reducA} and \eqref{eq.reducAk}, one proves easily that the operators $\bbA_{\perp}$ and the contractions  $S_{\perp}(t)$ of its associated semigroup $\{\mathcal{S}_{\perp}(t)\}_{t\geq 0}$  are also unitarily equivalent via $\mathcal{F}$  to direct integral of operators  $ \bbA_{\bk,\perp}$ and  $\rme^{-\rmi \bbA_{\bk,\perp}t} $:
\begin{eqnarray*}
\bbA_{\perp}&=&\mathcal{F}^{*}\, \Big( \int_{\bbR^3}^{\oplus} \bbA_{\bk, \perp} \, \rmd \bk \Big)  \mathcal{F} \   \mbox{ i.e. } \  \mathcal{F}  \big(\bbA_{\perp} \bU\big)(\bk)= \bbA_{\bk,\perp} \, \mathcal{F}\, \bU(\bk), \, \forall \, \bU\in D(\bbA_{\perp}),\mbox{ a.e. } \bk \in \bbR^3, \\
S_{\perp}(t)&=&\mathcal{F}^{*}\, \Big( \int_{\bbR^3}^{\oplus}\rme^{-\rmi \bbA_{\bk,\perp}t} \, \rmd \bk \Big)  \mathcal{F}     \mbox{ i.e. }    \mathcal{F}  \big(S_{\perp}(t) \bU\big)(\bk)=\rme^{-\rmi \bbA_{\bk,\perp}t} \, \mathcal{F} \, \bU(\bk), \, \forall \, \bU\in \mathcal{H}_{\perp}, \,  \mbox{ a.e.}\, \bk \in \bbR^3. \nonumber
\end{eqnarray*}
Similarly, the operators $\bbA_{\parallel}$ and  $\mathcal{S}_{\parallel}(t)$ for $t\geq 0$ are also unitarily  equivalent via $\mathcal{F}$  to direct integral of operators  and the latter relations hold also if one replaces all $\perp$-symbols by $\parallel$-symbols. \\

\noindent The direct integral decomposition of the operator $S_{\perp}(t)$ implies  that the Fourier transform $\bbU=\mathcal{F}\bU$ of the solution $\bU$ of  the evolution problem  \eqref{eq.schro} with initial divergence free conditions $\bU_0\in\mathcal{H}_{\perp} $ (as e.g. in  \eqref{CI}) satisfies  for a.e. $\bk \in \mathbb{R}^3$:
\begin{equation}\label{eq.refsolutionfourier}
\bbU(\bk,t)=\rme^{-\rmi \bbA_{\bk,\perp }\, t}\  \bbU_0(\bk),  \ \forall\, t\geq 0  \ \mbox{ where } \  \bbU_0(\bk)=\big(\mathcal{F} \,\bU_0\big)(\bk)\in \bC_{\bk ,\perp}^N.
\end{equation}
\subsection{From $\bbA_{\bk}$ to $\bbA_{|\bk|}$}\label{section-istrop}
In the expression of the solution  \eqref{eq.refsolutionfourier},  the  space $\bC_{\perp,\bk}$ depends on $\bk$ which complicates slightly the analysis of the decay of the Fourier components $\bbU(\bk,t)$.
To remediate this point, one can use the isotropic character of the  medium. \\

\noindent We introduce $(\be_1,\be_2,\be_3)$ the canonical orthonormal basis of $\bC$ and for $\bk\neq 0$ the unit vector $\widehat{\bk}=\bk/|\bk|$. We now construct a rotation  $\mathrm{R}_{\bk}:\bC\to \bC$ that maps the vector  $\widehat{\bk}$ into $\be_3$ as follows
\begin{itemize}
\item [-]If $\widehat{\bk}\neq \pm \be_3$, we set  $\mathbf{w}_{\bk}:= (\widehat{\bk}\times \be_3) /|\widehat{\bk}\times \be_3|\in \bC_{\bk, \perp}$ so that $(\widehat{\bk}, \mathbf{w}_{\bk}, \widehat{\bk} \times \mathbf{w}_{\bk} )$ is an orthonormal basis of $\bC$.  We then define $\mathrm{R}_{\bk}$  in this basis by 
$$\mathrm{R}_{\bk}(\widehat{\bk})=\be_3, \quad \mathrm{R}_{\bk}(\mathbf{w}_{\bk})=\be_1, \quad \mathrm{R}_{\bk}( \widehat{\bk} \times \mathbf{w}_{\bk})=\be_2.$$
\item[-] If $\widehat{\bk} = \be_3$, we simply set $\mathrm{R}_{\bk}=\mathrm{Id}_{\bC}$.
\item[-] If $\widehat{\bk}=- \be_3$, we define   $\mathrm{R}_{\bk}$  by $\mathrm{R}_{\bk}(\be_1)=-\be_2, \, \mathrm{R}_{\bk}(\be_2)=-\be_1 \ \mbox{ and } \  \mathrm{R}_{\bk}(\be_3)=-\be_3$.\end{itemize} 
Finally, we define the operator $\mathcal{R}_{\bk}: \bC^{N}\to \bC^N$  defined for all $  \bbU =(\bbE, \bbH, \bbP, \dot{\bbP},\bbM , \dot{\bbM})\in \bC^N$ by
\begin{equation}\label{eq.defRkvect}
\mathcal{R}_{\bk} \bbU=\big(\mathrm{R}_{\bk} \bbE, \mathrm{R}_{\bk} \bbH, (\mathrm{R}_{\bk}\bbP_j), (\mathrm{R}_{\bk}\dot{\bbP}_j),(\mathrm{R}_{\bk}\bbM_\ell), (\mathrm{R}_{\bk}\dot{\bbM}_\ell) \big)  
.\end{equation}
Using the identity (we leave its proof to the reader):$$\mathrm{R}_{\bk}^*\big(\mathrm{R}_{\bk}(\bu)\times  \mathrm{R}_{\bk}(\bv)\big)=\bu\times \bv, \quad \forall\,  \bu, \, \bv \in \bC$$
applied to $\bu=\bk$ (and thus with $R_{\bk}(\bu)=|\bk| \; e_3$), one observes with  \eqref{eq.defAk} that
\begin{equation}\label{eq.unit}
\bbA_{\bk}=\mathcal{R}_{\bk}^* \bbA_{|\bk|} \mathcal{R}_{\bk}  \quad \mbox{ where } \quad \bbA_{|\bk|} := \bbA_{|\bk| \be_3}.
\end{equation}
Thus, $\bbA_{\bk}$ and $\bbA_{|\bk|}$ are unitarily equivalent.
Furthermore, as 
\begin{equation} \label{trucmuche}
\mathcal{R}_{\bk}( \bC_{ \bk, \parallel}^N)=\bC_{\parallel}^N  \mbox{ and } \mathcal{R}_{\bk}(\bC_{\perp,\bk}^N)=\bC_{\perp}^N \mbox{ where } \bC_{\parallel}:=\bC_{\parallel, e_3}\mbox{ and } \bC_{\perp}:=\bC_{\perp,e_3},  
\end{equation}
one has by \eqref{eq.reducAk}:
\begin{equation}\label{eq.unit2}
 \bbA_{\bk, \parallel}=\mathcal{R}_{\bk}^*\,  \bbA_{|\bk|,\parallel}  \, \mathcal{R}_{\bk}  \quad \mbox{ and } \quad
 \bbA_{\bk, \perp}=\mathcal{R}_{\bk}^*\,  \bbA_{|\bk|,\perp}  \, \mathcal{R}_{\bk} .
\end{equation}
In particular, it shows  with \eqref{eq.refsolutionfourier} that  the Fourier transform $\bbU=\mathcal{F}\bU$ of the solution $\bU$ of  \eqref{eq.schro} with initial divergence free conditions $\bU_0\in\mathcal{H}_{\perp} $ satisfies  for a.e. $\bk \in \mathbb{R}^3$:
\begin{equation}\label{eq.refsolutionfourier2}
\bbU(\bk,t)= \mathcal{R}_{\bk}^* \ \rme^{-\rmi  \bbA_{|\bk|,\perp} \, t} \ \mathcal{R}_{\bk}\, \bbU_0(\bk),  \ \, \forall \; t\geq 0  \ \mbox{ where } \  \bbU_0(\bk)=\big(\mathcal{F}\bU_0\big)(\bk)\in \bC_{\bk,\perp}^N.
\end{equation}
 Thus, to estimate the  norm of $|\bbU(\bk,t)|$,  one only needs to analyse  the spectral properties of  $ \bbA_{|\bk|,\perp}$ on the  space $\bC_{\perp}^N$ which  is  now independent of $\bk$.
\section{Modal analysis}\label{sec-modal-analysis}
\noindent The proofs of Theorems \ref{thm_Lorentz} and  Theorem \ref{thm-optm} are based on the spectral  properties of $\bbA_{|\bk|}$. These theorems  are proved in Section \ref{proof-general-cases} based on results established in Sections \ref{sec-modal-analysis} to \ref{sec_mid-frequencies}. 
\subsection{Spectrum and resolvent of the finite dimensional  operators $\bbA_{|\bk|}$}\label{sec-spec-prop}
Using  the expression \eqref{eq.permmitivity-permeabiity}, the rational functions  $\varepsilon$ and $\mu$ can be rewritten as 
\begin{equation}\label{eq.representationirreduc}
\varepsilon(\omega)= \varepsilon_0 \, \frac{P_e(\omega)}{Q_e(\omega)}  \ \mbox { and } \ \mu(\omega)= \mu_0 \, \frac{P_m(\omega)}{Q_m(\omega)},  
\end{equation}
where the unitary polynomials $P_e$ and $Q_e$ (resp.  $P_m$ and $Q_m$) are  of degree $2 N_e$ (resp. $2N_m$). More precisely, $Q_e$  and  $P_e$  are explicitly  given in term of the $q_{e,j}$  by 
\begin{equation}\label{eq.formulaQ_e}
  Q_e(\omega)=\prod_{j=1}^{N_e}q_{e,j}(\omega) \quad  \mbox{ and }  \quad P_e(\omega)=Q_e(\omega)-\sum_{j=1}^{N_e}  \Omega_{e,j}^2\prod_{k=1,k\neq j}^{N_e}  \, q_{e,k}(\omega) .
\end{equation}
The  reader will verify that, with \eqref{eq.formulaQ_e}, the assumption $(\mathrm{H}_1)$  implies that   $P_e$ and $Q_e$ do not share common zeroes. Therefore the representation   \eqref{eq.representationirreduc} of the rational function $\varepsilon$ is irreducible.\\ [10pt]
Thus, $\calP_e$, the set of poles of $\varepsilon$, is exactly  the set of roots of   $Q_e$ which is the union over $j$ of the roots of $q_{e,j}$ and is therefore  included in $\overline{\mathbb{C}^-}$ (see remark \ref{Rem-roots}). \\ [10pt]  Moreover, by  $(\mathrm{H}_1)$ and remark \ref{Rem-roots}, the multiplicity of
a root  $\omega_*$ of $Q_e$ is either $1$ or $2$. 
\\ [10pt]
Obviously, similar observations  hold true for the polynomials $P_m$ and $Q_m$.\\

\noindent 
We introduce the sets $\mathcal{P}:=\mathcal{P}_e\cup \mathcal{P}_m$ and $\mathcal{Z}:=\mathcal{Z}_e\cup \mathcal{Z}_m$,
where we recall that $\mathcal{P}_e$ (resp. $\mathcal{P}_m$)  is the set of poles of the function $\varepsilon$ (resp. $\mu$) and   $\mathcal{Z}_e$   (resp. $\mathcal{Z}_m$) the set of zeros of  $ \varepsilon$ (resp. $ \mu$) and the rational function $\mathcal{D}: \mathbb{C}\setminus \mathcal{P}  \to \bbC$ defined  by
\begin{equation}\label{eq.ireductibleF}
\mathcal{D}(\omega)=\omega^2 \, \varepsilon(\omega) \, \mu (\omega),  \quad  \forall \; \omega \in  \mathbb{C}\setminus \mathcal{P}.
\end{equation}
Using $(\mathrm{H}_2)$, it is easy to see  that $\mathcal{D}(\omega)$ is irreducible. From \eqref{eq.representationirreduc} and \eqref{eq.formulaQ_e}, one deduces immediately that its numerator is of degree $N=2+2N_e+2N_m$ and its denominator is of degree $2N_e+2N_m$. Moreover, its poles  $p$  and zeros  $z$ are respectively the elements $ \mathcal{P}$ and $\mathcal{Z}\cup \{0 \}$. 
\\ [10pt] 
\noindent We introduce for any $\bk\in \bbR^3$, the set 
$$
S(|\bk|):=\{\omega \in \mathbb{C}\setminus \mathcal{P} \mid  \mathcal{D}(\omega) =|\bk|^2\}.
$$
The following proposition establishes the link  between $S(|\bk|)$ and the spectrum of the operator $\bbA_{|\bk|, \perp}$ for $\bk\neq 0$. 

\begin{Pro}\label{Prop.spec}
Let assume  $(\mathrm{H}_1)$ and $(\mathrm{H}_2)$. Then for any $\bk\in \mathbb{R}^3\setminus \{0 \}$, one has
\begin{equation}\label{eq.spec}
  \sigma(\bbA_{|\bk|,\perp})= S(|\bk|).
\end{equation}
Moreover, each eigenvalue  $\omega \in \sigma(\bbA_{|\bk|,\perp})$ has a geometric multiplicity of $2$. 
\end{Pro}
\noindent Finally, we conclude this section  with the expression of the resolvent \begin{equation} \label{defres} R_{|\bk|}(\omega):=(\bbA_{|\bk|,\perp}- \omega \mathrm{I})^{-1}.\end{equation} To obtain a readable expression, it appears useful to introduce some intermediate operators. \\ [12pt]
We first define four linear operators in ${\cal L}(\bC_{\perp}^N, \bC_{\perp})$: given $\bbF := \big( \be, \bh, \bp_j, \dot\bp_j,  \bm_\ell, \dot\bm_\ell \big) \in \bC_{\perp}^N$
\begin{equation}\label{operatorsApAm}
	\begin{array}{lll}
\ds 	\bbA_{e,j}(\omega)\,  \bbF :=\frac{( -\rmi \alpha_{e,j}-\omega ) \, \bp_j-\rmi  \, \dot{\bp}_j}{q_{e,j}(\omega)}, \quad & \ds \bbA_{m,\ell}(\omega)  \, \bbF :=\frac{( -\rmi \alpha_{m,\ell}-\omega ) \bm_{\ell}-\rmi \, \dot{\bm}_{\ell}}{q_{m,\ell}(\omega)}, \\ [12pt]
	\ds 	\dot \bbA_{e,j}(\omega)\,  \bbF :=\frac{\rmi  \, \omega_{e,j}^2\,  \bp_{ j}  -\omega \, \dot{\bp}_{ j}}{q_{e,j}(\omega)}, \quad & \ds 
	\dot \bbA_{m,\ell}(\omega)  \, \bbF :=\frac{\rmi  \, \omega_{m,\ell}^2\,  \bm_{\ell}  -\omega \, \dot{\bm}_{\ell}}{q_{m,\ell}(\omega)},
		\end{array}
\end{equation}
from which we define two more  operators  in ${\cal L}(\bC_{\perp}^N, \bC_{\perp})$:
\begin{equation}\label{operatorsAeAh}
		\ds \bbA_e(\omega) \,  \bbF =-\varepsilon_0 \, \Big( \, \be+ \rmi\sum  \Omega_{e,j}^2 \, 	\dot \bbA_{e,j}(\omega)\,  \bbF\Big) , \quad  \bbA_m(\omega) \,  \bbF =-\mu_0\, \Big( \,\bh+ \rmi  \sum\  \Omega_{m,\ell}^2 \, 	\dot \bbA_{m,\ell}(\omega)  \, \bbF \Big).
\end{equation}
Finally, we shall also define the finite subset of $\bbC$ :
\begin{equation} \label{defST}
	{\cal S}_ {\cal T} := \mathcal{P} \cup \mathcal{Z}_m\cup \{ 0\}.
	\end{equation}
\begin{Pro}\label{Prop.res} Let $\bk\in \bbR^3$. For any $\omega \in  \bbC\setminus \big(S(|\bk|)\cup {\cal S}_ {\cal T}\big) $, the resolvent is given by 
\begin{equation}\label{eq.expressresolv}
R_{|\bk|}(\omega)=\mathcal{V}_{|\bk|}(\omega) \, \mathcal{S}_{|\bk|}(\omega)+\mathcal{T}(\omega)
\end{equation} 
with  $ \mathcal{S}_{|\bk|}(\omega) \in {\cal L}(\bC_{\perp}^N, \bC_{\perp})$ defined by:
given $\bbF := \big( \be, \bh, \bp_j, \dot\bp_j,  \bm_\ell, \dot\bm_\ell \big) \in  \bC_{\perp}^N$
 \begin{equation} \label{defS} 
 \mathcal{S}_{|\bk|}(\omega) \, \bbF:= \displaystyle \frac{\omega \mu(\omega ) \, \bbA_e(\omega) \,  \bbF  -|\bk| \; {\bf e_3} \times
 \bbA_m(\omega) \,  \bbF}{\mathcal{D}(\omega)-|\bk|^2}	,
\end{equation} 
$\mathcal{V}_{|\bk|}(\omega) \in {\cal L}(\bC_{\perp}, \bC_{\perp}^N)$ defined by: given $\bbX \in \bC^{\perp}$
 \begin{equation} \label{defV} 	
	\left| \; 	\begin{array}{lll}	\mathcal{V}_{|\bk|}(\omega) \, \bbX& := & \ds \bigg(\bbX,  0\,  ,  - \Big(\frac{\bbX}{q_{e,j}(\omega)}\Big),  \Big(\frac{ \rmi \, \omega  \,\bbX}{q_{e,j}(\omega)}\Big), 0 \, , 0 \, \bigg) \\ [18pt]
		&+ &  \ds \frac{|\bk| }{ \omega \mu(\omega)} \;  \bigg(0 \, , {\bf e_3}  \times \bbX ,  0 \, , 0 \,, - \Big(\frac{{\bf e_3}  \times \bbX}{ q_{m,\ell}(\omega)}\Big), \Big(\frac{\rmi \,\omega \,{\bf e_3}  \times \bbX}{ q_{m,\ell}(\omega)} \Big) \bigg) \quad 
	\end{array} \right.
\end{equation} 
and finally $\mathcal{T}(\omega) \in {\cal L}(\bC_{\perp}^N)$ defined by :  given $\bbF := \big( \be, \bh, \bp_j, \dot\bp_j,  \bm_\ell, \dot\bm_\ell \big) \in \bC_{\perp}^N$
 \begin{equation} \label{defT} 	
 \left| \; 	\begin{array}{lll}
 		\mathcal{T}(\omega) \bbF & :=  & \bigg(0 \, , \displaystyle \frac{\bbA_m(\omega) \,  \bbF}{\omega \mu(\omega)},0 \, ,0  \, , - \Big(\frac{\bbA_m(\omega) \,  \bbF}{\omega \mu(\omega) \, q_{m,\ell}(\omega)}\Big), \Big( \frac{\rmi \,  \bbA_m(\omega) \,  \bbF}{ \mu(\omega) q_{m,\ell}(\omega)}\Big)\bigg)\\ [18pt]
 		&+ &  \Big(0 \, , 0\,,\big(\bbA_{e,j}(\omega) \bbF \big), \big(\dot{\bbA}_{e,j}(\omega) \bbF \big)  , \big(\bbA_{m,\ell}(\omega)  \bbF \big) , \big( \dot{\bbA}_{\dot{m},\ell}(\omega) \bbF \big)\Big).
\end{array} \right.
\end{equation} 
\end{Pro}
\noindent The (purely computational) proofs of  Propositions \ref{Prop.spec} and \ref{Prop.res} are delayed to the Appendix.
	\begin{Rem}
The function $\omega \mapsto R_{|\bk|}(\omega)$ is  well-defined  and analytic for $\omega \in \bbC \setminus  \sigma(\bbA_{|\bk|,\perp})$. The set ${\cal S}_{\cal T}$ is the set of singularities for $\omega \mapsto {\cal T}(\omega)$ but they are removable singularities in the expression \eqref{eq.expressresolv}.
	\end{Rem}
\subsection{The dispersion relation}  \label{sec_dispersion}
\subsubsection{General properties of the dispersion relation} \label{sec_dispersion-1}
By Proposition \ref{Prop.spec}, for a fixed wave number $\bk\neq 0$, the eigenvalues of $\bbA_{|\bk|,\perp}$ are the solutions $\omega\in \mathbb{C}\setminus \mathcal{P}$  of the equation
\begin{equation}\label{eq.disp}
{\cal D}(\omega)=|\bk|^2,
\end{equation}
(with ${\cal D}(\omega)$ defined by \eqref{eq.ireductibleF}) referred in physics as the dispersion relation. In this section, we provide several  results and useful remarks  about  this equation.\\[6pt]
 \noindent As a consequence of  Remark \ref{rem.Disp}, one has $ \sigma(\bbA_{|\bk|})\subset \overline{\bbC^-}$, thus  $\sigma(\bbA_{|\bk|, \perp})\subset \sigma(\bbA_{|\bk|})\subset \overline{\bbC^-}$. We give here an elementary proof  that  for $\bk\neq 0$, the spectrum  $ \sigma(\bbA_{|\bk|,\perp})$ is included in the lower open complex half plane 
 \begin{equation} \label{loc_spectrum} 
 	 \sigma(\bbA_{|\bk|,\perp}) \subset \bbC^-,
 	\end{equation} 
 	based on the dispersion relation \eqref{eq.disp}.  Assume by contradiction that there exists $\omega \in \overline{\bbC^+}\cap \sigma(\bbA_{|\bk|, \perp})$. From Proposition  \ref{Prop.spec}, $\omega$ satisfies \eqref{eq.disp} and  therefore $\omega\neq 0$ and $\omega\notin \mathcal{P}$. 
 Taking the real part and the imaginary part of \eqref{eq.disp} leads to
\begin{eqnarray}
&& \operatorname{Re}\big(\omega\, \varepsilon(\omega)\big) \, \operatorname{Re}\big(\omega \mu(\omega)\big) =|\bk|^2+ \operatorname{Im}\big(\omega\,  \varepsilon(\omega)\big) \, \operatorname{Im}\big(\omega \mu(\omega)\big)> 0 \label{eq.samesign} \\&&  \operatorname{Re}(\omega \,\varepsilon(\omega)) \, \operatorname{Im}(\omega \mu(\omega))+ \operatorname{Re}(\omega\, \mu(\omega)) \, \operatorname{Im}(\omega \varepsilon(\omega))=0  \label{eq.oppositesign}
\end{eqnarray}
where the  rational Herglotz functions $\omega \mapsto  \omega \varepsilon(\omega)$ and  $\omega \mapsto \omega \mu(\omega)$ are analytic and have a  positive imaginary part  on $\bbC^+$. Furthermore, one has for  $\omega \in \R \setminus \calP$:
\begin{equation}\label{eq.positvity}
\operatorname{Im}(\omega \varepsilon(\omega))=\varepsilon_0 |\omega|^2\,  \sum_{j=1}^{N_e} \frac{\Omega^2_{e,j} \,\alpha_{e,j}} {|q_{e,j}(\omega)|^2}\geq 0 \ \mbox{ and } \ \operatorname{Im}(\omega \mu(\omega))=\mu_0 |\omega|^2\,  \sum_{\ell=1}^{N_m} \frac{\Omega^2_{m,\ell} \,\alpha_{m,\ell}} {|q_{m,\ell}(\omega)|^2}\geq 0.
 \end{equation}
Thus, by the weak dissipation assumption \eqref{WD},  at least one of the coefficients $\alpha_{e,j}$ or $\alpha_{m,\ell}$ is positive and it  follows from \eqref{eq.positvity} that 
 \begin{equation}\label{eq.consequenceWD}
 \operatorname{Im}(\omega \varepsilon(\omega))>0  \  \mbox{ or }  \ \operatorname{Im}(\omega \mu(\omega))>0 \  \mbox{ on }  \ \R^* \setminus \mathcal{P}.
 \end{equation}
 Hence, \eqref{eq.samesign} implies that $\operatorname{Re}\big(\omega \varepsilon(\omega)\big)$ and $\operatorname{Re}\big(\omega \mu(\omega)\big)$  have the same sign (and do not vanish), whereas  \eqref{eq.oppositesign} implies that they have opposite sign, which leads to a contradiction. Thus,  for $\bk \neq 0$, the solutions of the dispersion relation ${\cal D}(\omega)=|\bk|^2$ all lie in $\bbC^-$.
\\[6pt]
\noindent Using the fact that $\varepsilon(-\overline{\omega})=\overline{\varepsilon(\omega)}$ and $\mu(-\overline{\omega})=\overline{\mu(\omega)}$, one observes that if $\omega$ is a solution of the dispersion relation then $-\overline{\omega}$ is also a solution of this equation. Thus for a fixed $\bk\neq 0$, the  set $\sigma(\bbA_{|\bk|, \perp})$  is invariant by  the transformation $\omega \to -\overline{\omega}$.\\[6pt]
 \noindent From the irreducible form \eqref{eq.ireductibleF} of ${\cal D}$, one deduces that
  the dispersion relation is equivalent to a polynomial equation of degree $N$.  Namely,  for a fixed $\bk \in \R^{3}\setminus \{ 0\}$,  one has:
\begin{equation}\label{eq.disppolynom}
{\cal D}(\omega)=|\bk|^2  \Longleftrightarrow   D_{|\bk|}(\omega) =0, \quad D_{|\bk|}(\omega):= \varepsilon_0 \, \mu_0 \, \omega^2 \, P_e(\omega)  P_m(\omega) - |\bk|^2  Q_e(\omega) Q_m(\omega) \, \end{equation}
is a polynomial  of degree $N = 2 + 2N_e + 2 N_m$  with dominant coefficient $\varepsilon_0 \, \mu_0$. Thus, one has
\begin{equation}\label{equivD}
	 D_{|\bk|}(\omega)\sim \varepsilon_0 \, \mu_0 \, \omega^{2 + 2N_e + 2 N_m} , \quad (|\omega| \rightarrow + \infty). 
 \end{equation}
This  leads to the following corollary on the diagonalizability of  $\bbA_{|\bk|,\perp}$.
\begin{Cor} \label{eq.crit-diag}
For $\bk\in \R^3\setminus \{ 0\}$, $\bbA_{|\bk|,\perp}$ is diagonalizable on $\bC^N_{\perp}$ if and only if the roots of the polynomials $D_{|\bk|}$ defined in \eqref{eq.disppolynom} are simple.
\end{Cor}
\begin{proof}
By Proposition \ref{Prop.spec} and relation \eqref{eq.disppolynom},  the eigenvalues of $\bbA_{|\bk|,\perp}$ are the solutions the polynomial equation $D_{|\bk|}(\omega)=0$  of degree $N$. Furthermore,  each distinct solution of this equation is an eigenvalue of geometric multiplicity $2$.  As the space $\bC^N_{\perp}$  is of dimension $2N$, see \eqref{trucmuche} the result follows immediately from a simple argument  of dimension.
\end{proof}
\noindent We end this section with two paragraphs: one  on the poles $\mathcal{P}$ and one on the zeros $\mathcal{Z}\cup \{ 0\}$ of the rational function $\mathcal{D}$ associated to the dispersion relation. Indeed (as we will see in the Section \ref{sec-main-lines-analysis} and in the parts \ref{sec.HF} and \ref{sec.LF}) poles and zeros play a key roles for the asymptotics of eigenvalues of $\bbA _{|\bk|}$ respectively for $|\bk|\gg 1$ and $|\bk|\ll 1$.
\subsubsection{Properties of the poles of  the rational function $\mathcal{D}$}\label{sec_dispersion-poles}
As we saw in Section \ref{sec-spec-prop}, the elements of $\mathcal{P}$  are exactly the zeros of the denominator of the rational function $\mathcal{D}$, that is the polynomial of $Q_e Q_m$ of degree $2(N_e+N_m)$.
Thus, $\mathcal{P}$ contains $2(N_e+N_m)$ elements (counted with multiplicity) and this set is invariant by the transformation $\omega \mapsto -\overline{\omega}$  (see Remark \ref{Rem-roots}).\\ [12pt]
The multiplicity $\mathfrak{m}_p$  of a pole  $p\in \mathcal{P}$ can not take any arbitrary value. Indeed,  the elements of $\mathcal{P}$ are precisely  (see Section \ref{sec-spec-prop}) the  roots of the $2$-nd order polynomials $q_{e,j}$ or $q_{m,\ell}$ for $j\in \{ 1,\ldots, N_e\}$ and $\ell\in \{ 1,\ldots, N_m\}$. Thus, $\mathcal{P}\subset \overline{\bbC^-}$ (see Remark \ref{Rem-roots}) and  one has:
\begin{itemize}
\item  If  $p\notin  \rmi \, \R^{-,*}$ then $\mathfrak{m}_p=1$ if $p\notin \calP_e\cap \calP_m$ or  $\mathfrak{m}_p=2$  if $\omega\in \calP_e\cap \calP_m$,
\item Else  if $p \in \rmi \, \R^{-,*}$, then by $(\mathrm{H}_1)$,  $\mathfrak{m}_p\in \{1,2\}$ if $p\notin \calP_e\cap \calP_m$ or $\mathfrak{m}_p\in \{2,3,4\}$ if $p \in \calP_e\cap \calP_m$.
\end{itemize}
We will see  in Sections \ref{sec-main-lines-analysis} and  \ref{sec.estmHF} that the important poles for our  analysis for $|\bk|\gg1 $ are the ones that lie on  the real axis. These poles are associated to polynomials $q_{e,j}$ or $q_{m,\ell}$ for which  $\alpha_{e,j}=0$   or $\alpha_{m,\ell}=0$.  Thus, they are of the form $p=\pm \, \omega_{e,\ell}\in \calP_e$  or $p=\pm \, \omega_{m,\ell} \in \calP_m$ with $\omega_{e,j}, \, \omega_{m,\ell}>0$ and have  multiplicity $\mathfrak{m}_p=1$ if $p\notin \calP_e\cap \calP_m$ or $\mathfrak{m}_p=2$ if $p\in \calP_e\cap \calP_m$.

\subsubsection{Properties of the zeros of the rational function  $\mathcal{D}$}\label{sec_dispersion-zeros}
As we saw in Section \ref{sec-spec-prop}, the elements of $\mathcal{Z}\cup \{0 \} $  are the zeros   of the rational function $\mathcal{D}$, that is the zeros of the polynomial $\omega^2 P_e P_m$ of degree $N=2(N_e+N_m)+2$.
Hence $\mathcal{Z}\cup \{0 \} $ contains $N$ elements (counted with multiplicity).   $\mathcal{D}$ is defined as the product of the two (non-constant) rational  Herglotz functions $\omega \mapsto \omega \varepsilon(\omega)$ and $\omega \mapsto \omega \mu(\omega)$  which satisfy  (using \eqref{eq.permmitivity-permeabiity})  $\operatorname{Im}(\omega \varepsilon(\omega))>0$ and  $\operatorname{Im}(\omega \mu(\omega))>0$  on $\bbC^+$.  Thus, these functions  could not vanish in the upper-half plane $\bbC^+$ (indeed this property holds more generally for any non-constant Herglotz function  as a consequence of the open mapping theorem for analytic functions, see for e.g.  \cite{cas-mil-17})). Therefore, $\mathcal{Z}\cup \{0 \} $ is included in $\overline{\bbC^-}$. Furthermore, as $\varepsilon(-\overline{\omega})=\overline{\varepsilon(\omega)}$ and $\mu(-\overline{\omega})=\overline{\mu(\omega)}$, this set  is  also invariant by the transformation $\omega\mapsto -\overline{\omega}$.\\
Concerning the multiplicity  $\mathfrak{m}_{z}$ of a zero $z\in \mathcal{Z}\cup \{0 \} $, one observes that:
\begin{itemize}
\item  $0$ has a multiplicity $\mathfrak{m}_{0}=2$  since (by   \eqref{hypomegabis} and  \eqref{eq.permmitivity-permeabiity})  in the vicinity of $0$: 
\begin{equation}\label{eq.equivalentzero}
\mathcal{D}(\omega)\sim \omega^2  \varepsilon(0) \mu(0), \;  \varepsilon(0)=\varepsilon_0 \Big( 1+ \sum_{j=1}^{N_e} \frac{\Omega_{e,j}^2}{\omega_{e,j}^2}\Big) >0, \; \mu(0)=\mu_0 \Big( 1+ \sum_{\ell=1}^{N_m} \frac{\Omega_{m,\ell}^2}{\omega_{m,\ell}^2}\Big)> 0 .
\end{equation}
\item If $z\notin \rmi \R^{-}$, then $-\overline{z}$ is  a distinct zero with the same multiplicity $\mathfrak{m}_z$. 
\item 
From the property \eqref{eq.consequenceWD} of $\varepsilon(\omega)$ and $\mu(\omega)$, one immediately sees that  
\begin{enumerate}
\item If there exist two indices $j_0$ and $\ell_0$ such that $\alpha_{e,j_0}>0$ and $\alpha_{m,\ell_0}>0$ then $\mathcal{D}$ has no non-zero real zeros, i.e. $\R^*\cap \mathcal{Z}=\varnothing$.
\item  If not, either all $\alpha_{e,j}$ vanish in which case $\R^*\cap \mathcal{Z}=\mathcal{Z}_e$, either all $\alpha_{m,\ell}$ vanish in which case $\R^*\cap \mathcal{Z}=\mathcal{Z}_m$. Moreover,
these zeros  $z\in \mathcal{Z}_e$ (resp. $\mathcal{Z}_m$) have multiplicity $\mathfrak{m}_{z}=1$ (this is easily deduced from the graph  of the function $\omega \in \R\setminus \calP_e \mapsto \varepsilon(\omega)$ or $\omega \in \R\setminus \calP_m \mapsto \mu(\omega)$ given by \eqref{eq.permmitivity-permeabiity}).
\end{enumerate}
We point out that the second scenario occurs in particular in the critical configurations  described by  Definition  \ref{Critical_cases}.
\end{itemize}
We will see  in Sections \ref{sec-main-lines-analysis} and  \ref{sec.LF} that only the real zeros $\R\cap (\mathcal{Z}\cup \{ 0\})$ can contribute  at the main order to the  asymptotic  of  the Fourier-components  of $|\bU(k,t)|$ for $|\bk|\ll 1$.

\subsection{Main lines of the analysis}\label{sec-main-lines-analysis}
The strategy for proving Theorem \ref{thm_Lorentz} is quite clear and simple:
\begin{itemize}
	\item[(i)] For each $\bk \in \R^3\setminus \{ 0\}$, one estimates individually $\bbU(\bk,t)$ using formula \eqref{eq.refsolutionfourier}, or more precisely \eqref{eq.refsolutionfourier2}, via the estimation of the exponential $\rme^{-\rmi  \bbA_{|\bk|,\perp} \, t}$.
	\item [(ii)]  One gathers the above estimates to estimate the $L^2$-norms of the various fields ${\bf E}, \bH, ...$ via the norm $\|\bU(\cdot,t)\|_{\cal H}$ thanks to Plancherel's theorem.
\end{itemize}
For the first step, the key property is the fact that, as emphasized in Section \ref{sec_dispersion-1} (see \eqref{loc_spectrum}), the spectrum of $\bbA_{|\bk|,\perp}$ is included if the complex half-plane $\bbC^-$. As a consequence, each $|\bbU(\bk,t)|$  will decay exponentially to $0$ for large $t$. \\ [12pt]
	The reason why the exponential decay is lost for $\|\bU(\cdot,t)\|_{\cal H}$ and degenerates into a polynomial decay, is linked to the fact that the rate of decay of $|\bbU(\bk,t)|$   depends on $\bk$ and degenerates when $|\bk|$ tends to $0$ or $+\infty$. This decay rate is of course linked to the distance 
of $\sigma( \bbA_{|\bk|,\perp})$ to the real axis that can be deduced from the analysis of the dispersion relation ${\cal D}(\omega) = |\bk|^2, $ where ${\cal D}(\omega)=\omega^2 \, \varepsilon(\omega) \, \mu (\omega)$, that characterizes the eigenvalues of $\bbA_{|\bk|,\perp}$ (see Proposition \ref{Prop.spec}).
\begin{itemize}
	\item When $|\bk|$ is bounded from below and above, this distance is uniformly bounded from below by a  positive number which results into a uniform exponential decay of the corresponding $|\bbU(\bk,t)|$'s. 
	\item When $|\bk|$ tends to $0$ or $+\infty$, this distance tends to $0$ and obtaining sharp estimates of $|\bbU(\bk,t)|$ requires to analyse the asymptotic behaviour of the imaginary parts  $\operatorname{ Im }\, \omega(|\bk|)$ of those eigenvalues  $\omega(|\bk|)  \in  \sigma( \bbA_{|\bk|,\perp})$  whose distance to the real axis tends to 0. 
\end{itemize}
From this observation, it is natural to split the analysis into three steps depending on the values of the ``spatial frequency" $|\bk|$~: 
\begin{itemize}
	\item[(a)] For ``mid-range frequencies", typically  $k_- \leq |\bk| \leq k_+$ (with $k_->0$),  one will essentially abandon the spectral approach to the profit of standard techniques for ODE's combined with compactness arguments (the region $k_- \leq |\bk| \leq k_+ $ is compact). This will be detailed in Section \ref{sec_mid-frequencies}.
	\item[(b)] For ``high frequencies'', $|\bk|\geq k_+$, some eigenvalues can be arbitrarily close to the real axis when $|\bk|  \rightarrow + \infty$: considering the limit of the equation ${\cal D}(\omega) = |\bk|^2$ when $|\bk|  \rightarrow + \infty$ it is natural to look at where, for real $\omega$, the function ${\cal D}(\omega)$ tends to $+\infty$. This  only occurs in one of the following situations when 
	\begin{itemize} \item $\omega \rightarrow \pm \infty$ in which case ${\cal D}(\omega) \sim \omega^2/c^2$, so that one expects the existence of two branches of eigenvalues $\omega_\infty (|\bk|)$ and $- \, \overline{\omega}_\infty (|\bk|)$ where 
		$$
		\omega_\infty (|\bk|) = c \, |\bk| + o\big(|\bk|\big), \quad  |\bk|  \rightarrow + \infty,
		$$
		where $c:=\sqrt{\varepsilon_0\,\mu_0}^{-1}$ is the speed of light in the vacuum. 		
	\item $\omega \rightarrow \pm \, \omega_{\nu,q}$  with $\alpha_{\nu,q}= 0$, $\nu = e$ or $m$ where $\omega_{\nu,q}$ is a real pole of $\mathcal{D}$ of multiplicity 1 if $\omega_{\nu,q}\notin \calP_e\cap \calP_m$ or 2 if  $\omega_{\nu,q}\in \calP_e\cap \calP_m$ (see Section  \ref{sec_dispersion-poles}).\\ [10pt]  If  $\omega_{\nu,q}\notin \calP_e\cap \calP_m$, one  expects the existence of two branches of eigenvalues 
	$\omega_{\nu,q}(|\bk|)$ and $- \, \overline{\omega}_{\nu,q}(|\bk|)$ such that 
		$$
	\omega_{\nu,q}(|\bk|) = \omega_{\nu,q}+ o(1), \quad  |\bk|  \rightarrow + \infty.
	$$
	 Oppositely, if  $\omega_{e,q_1}=\omega_{m,q_2}\in \calP_e\cap \calP_m$, one expects  the existence of four branches of eigenvalues $\omega_{e,q_1}(|\bk|) $, $-\overline{\omega}_{e,q_1}(|\bk|) $ , $\omega_{m,q_2}(|\bk|) $, $-\overline{\omega}_{m,q_2}(|\bk|) $  such that   $\omega_{e,q_1}(|\bk|) $ and $\omega_{m,q_2}(|\bk|)$ tends to   $\omega_{e,q_1}$ as $|\bk| \rightarrow + \infty$.
\end{itemize}
In the above cases, a more precise analysis of the asymptotic behaviour $\operatorname{Im}\, \omega_\infty (|\bk|)$ and  $\operatorname{Im}\,\omega_{\nu,q} (|\bk|)$ when $|\bk|  \rightarrow + \infty$ plays a crucial role.
	\item [(c)]  For ``low frequencies", $0< |\bk| \leq k_- $, again some eigenvalues can be arbitrarily close to the real axis when $|\bk|  \rightarrow 0$ : considering the limit of the equation ${\cal D}(\omega) = |\bk|^2$ when $|\bk|  \rightarrow 0$ it is natural to look at where, for real $\omega$, the function ${\cal D}(\omega)$ vanishes. This  only occurs in one of the following situations when 
	\begin{itemize} \item $\omega \rightarrow 0$ in which case by \eqref {eq.equivalentzero}:  ${\cal D}(\omega) \sim \omega^2/c_0^2$ with $c_0:=(\varepsilon(0)\, \mu(0))^{-1/2}$. Thus, one expects the existence of two branches of eigenvalues $\omega_0 (|\bk|)$  and $- \, \overline{\omega}_0 (|\bk|)$ such that 
		$$
		\omega_0 (|\bk|) =c_0 \,|\bk|+ o(|\bk|), \quad  |\bk|  \rightarrow 0.
		$$
		\item  If we are in the second scenario  described in the Section  \ref{sec_dispersion-zeros},  non-zero real zeros  exist and are of multiplicity $\mathfrak{m}_z=1$. Thus, in this case, one has  $\omega \rightarrow  \pm z_{\nu}$, for  $\nu = e$ or $m$, where $z_{\nu}$ is a zero of ${\cal D}(\omega)$. One then expects the existence of two branches of eigenvalues 
		$\omega_{z_\nu}(|\bk|)$ and $-\overline{\omega}_{z_\nu}(|\bk|)$ such that 
		$$
		\omega_{z_\nu} (|\bk|) = z_{\nu} + o(1), \quad  |\bk|  \rightarrow 0.
		$$
		\end{itemize}
In the above cases, a more precise  analysis of the asymptotic behaviour of $\operatorname{Im} \, \omega_0(|\bk|)$ and  $\operatorname{Im}\,	\omega_{z_\nu} (|\bk|)$ when $|\bk|  \rightarrow 0$ plays a crucial role. \\ [12pt]
The analysis of point $(b)$ will explain why the polynomial stability is limited by the Sobolev regularity of the initial data with the first  term in the right hand sides of the estimates \eqref{polynomial_decayncr} and \eqref{polynomial_decaycr}, while the analysis of point $(c)$ will explain  why this polynomial stability in the  second term in the right hand sides of  \eqref{polynomial_decayncr} and \eqref{polynomial_decaycr} is related to the low frequency behavior of the Fourier components of the solution (and thus involved  naturally the spaces $\boldsymbol{\cal  L}_p (\R^3)^N$.)
\end{itemize}
Let us mention that, since the matrices $ \bbA_{|\bk|,\perp}$ are not normal, the estimate of $\rme^{-\rmi  \bbA_{|\bk|,\perp} \, t}$ in the regions $|\bk| \geq k_+$ and $|\bk| \leq k_-$ cannot be reduced to the study of their eigenvalues. That is why we shall complete the analysis by proving that
\begin{itemize} 
	\item for $k_-$ and $k_+$ well chosen, the matrices $ \bbA_{|\bk|,\perp}$ are diagonalizable in the above regions,
	\item the associated spectral projectors can be bounded uniformly in  $|\bk|$ in each region.
	\end{itemize}

\section{Asymptotic analysis for large spatial frequencies $|\bk|\gg 1$}\label{sec.HF}
As this section is the longest of the article, its seems useful to describe its structure. It is made of three main subsections:
\begin{itemize} 
	\item Section \ref{sec.HFdispcurv}: Asymptotics of dispersion curves for $|\bk| \gg 1$.
	\item  Section  \ref{sec-decomp-sol-inf}: Spectral decomposition of the solution for $|\bk| \gg 1$.
	\item  Section \ref{sec.estmHF}: Large time estimate of $\bbU(\bk,t )$ for $|\bk| \gg 1$.
\end{itemize}
Section \ref{sec.estmHF} puts together the results of Sections  \ref{sec.HFdispcurv} and \ref{sec-decomp-sol-inf}. It is decomposed into subsections 
dedicated to the estimate each of the components of the solution issued from Section \ref{sec-decomp-sol-inf}.\\ [12pt]
From the technical point of view, the proof of the asymptotic expansions will be based on a lemma proved in  the appendix \ref{sec-asymptotic}, namely the Lemma \ref{Lem-implicte-function}, that can be seen as a kind of implicit function theorem for functions in the complex plane.  This will also be the case in section \ref{sec.estmLF}.
\subsection{Asymptotics of dispersion curves  for $|\bk|\gg 1$}\label{sec.HFdispcurv}
In this section, we focus on long time estimates of the high (spatial) frequency components  $\bbU(\bk,t)$ (see \eqref{eq.refsolutionfourier}) of the solution. As explained in the above paragraph, the decay of $\bbU(\bk,t)$ is related to the analysis of the solutions of the dispersion relation \eqref{eq.disp} for $|\bk|\gg 1$. Roughly speaking, as $|\bk|^2\to +\infty$ when $|\bk|\to+\infty$, the solutions of   \eqref{eq.disp} must satisfy $|\mathcal{D}(\omega)|\to +\infty$ as $|\bk|\to+\infty$. Thus, we observe two scenari:
either they diverge to $\infty$ or they converge to a pole $p\in \mathcal{P}_e\cup \mathcal{P}_m$ around which, If $p\in \mathcal{P}$ is a pole of multiplicity $\mathfrak{m}_p$, $\mathcal{D}$ can be rewritten as 
\begin{equation}\label{eq.polevoisinage}
	\mathcal{D}(\omega)=(\omega-p)^{-\mathfrak{m}_p} f(\omega) \mbox{ with $f$ analytic in a vicinity of $p$ and }  f(p)=A_{p}\neq 0.
\end{equation}
This leads to the following proposition.
\begin{Pro}\label{prop.dispersioncurves}
There exists $k_+>0$  such that for $|\bk|\geq k_+$, the solutions of the rational dispersion relation  \eqref{eq.disp} (or of its equivalent polynomial form \eqref{eq.disppolynom})  are all simple. These solutions form $N$ distinct  branches which  are $C^{\infty}$-smooth functions (with respect to  $|\bk|$) characterized by their asymptotic  expansion for large $|\bk|$. More precisely: 
\begin{itemize}
	\item For any $p \in {\cal P}$ with multiplicity $\mathfrak{m}_p$, for large enough $|\bk|$, there exists $\mathfrak{m}_p$ distinct branches  of solutions $ \omega_{p,n},\, n=1,\ldots, \mathfrak{m}_p$  of  \eqref{eq.disp}  satisfying
	\begin{equation}\label{eq.pole}
		\omega_{p,n}(|\bk|)=p+  a_{p,n} \;  |\bk|^{-\frac{2}{\mathfrak{m}_p}} \, (1+o(1)), \quad a_{p,n}=|A_p|^{1/\mathfrak{m}_p} \; \rme^{\rmi \frac{\theta_p}{\mathfrak{m}_p}}  \, \rme^{\frac{2 \, \rmi n\pi}{\mathfrak{m}_p}} \quad (|\bk| \to +\infty )
	\end{equation}
	where $A_p$ is defined in  \eqref{eq.polevoisinage} and $\theta_p\in (-\pi, \pi]$ is the  principal argument of $A_p$. 
\item There are $2$  distinct branches of solutions $\omega_{\pm \infty}$ of  \eqref{eq.disp}  that diverge to $\infty$ as
\begin{equation}\label{eq.pminfty}
\omega_{\pm \infty}(|\bk|)=\pm \, c \, |\bk| \, (1+o(1)) \quad  \mbox{ with } c=(\varepsilon_0\,\mu_0)^{-\frac{1}{2}}  \quad (|\bk| \to +\infty). 
\end{equation}
\end{itemize}
Moreover, there are no other solutions of \eqref{eq.disp} (for $|\bk|\geq k_+$) than those described above. 
\end{Pro}
\begin{proof}
\noindent {\bf Step 1: construction of the branches of solutions $\omega_{p,n}(|\bk|)$.} \\[4pt]
Let $p\in\mathcal{P}$ a pole of multiplicity $\mathfrak{m}_p$. Then, the rational function $\mathcal{D}$ can be factorized as in 
\eqref{eq.polevoisinage}. Then $p$ is a zero of multiplicity $\mathfrak{m}_p$ of the  function $\omega\mapsto \mathcal{D}(\omega)^{-1}= (\omega-p)^{\mathfrak{m}_p} f(\omega)^{-1}$ which is analytic in an open neighbourhood of $p$ and satisfies $f(p)^{-1}=A_p^{-1}\neq 0$.  Moreover, for $ |\bk|\neq 0$, solving  the dispersion relation: 
$$
  \mathcal{D}(\omega)=|\bk|^2 \quad  \mbox{ is equivalent to solving the equation} \quad   \mathcal{D}(\omega)^{-1}=|\bk|^{-2}=\big(|\bk|^{-\frac{2}{\mathfrak{m}_p}}\big)^{\mathfrak{m}_p}.
$$ 
Then, by applying the Lemma \ref{Lem-implicte-function} of the appendix \ref{sec-asymptotic} with
$$
\mathcal{G}(\omega) =  \mathcal{D}(\omega)^{-1}, \quad z=p , \quad g(\omega)  =f(\omega)^{-1}, \quad \mathfrak{m}=\mathfrak{m}_p, \quad  A=A_p^{-1} \mbox{ and  } \zeta=|\bk|^{-2/\mathfrak{m}_p},$$
we deduce for $|\bk|$ large enough the existence of $\mathfrak{m}_p$  distinct  branches of solutions: $|\bk|\mapsto \omega_{p,n}(|\bk|)$  of the equation $\mathcal{D}(\omega)=|\bk|^2$ which are  $C^{\infty}$ functions with respect to $|\bk|$ and satisfy \eqref{eq.pole}.\\ [12pt]  
\noindent {\bf Step 2: construction of the $2$ distinct branches of solutions $\omega_{\pm \infty}(|\bk|)$.} \\[4pt]
To reduce ourselves to the application of Lemma \ref{Lem-implicte-function} as in step 1, the trick consists in 
saying that $\mathcal{D}(\omega)=\omega^2 \varepsilon(\omega)\, \mu(\omega)=|\bk|^2$ is equivalent to  $\mathcal{D}(\omega)^{-1}=\omega^{-2}\varepsilon(\omega)^{-1}\, \mu(\omega)^{-1}=|\bk|^{-2}$. Then we introduce the new unknown  $\xi = 1/\omega$,  so that $|\omega| \rightarrow + \infty \Leftrightarrow \xi \rightarrow 0$ and 
\begin{equation}\label{eq.dispers}
  \mathcal{D}(\omega)=\omega^2 \varepsilon(\omega)\, \mu(\omega)=|\bk|^2 \quad \Longleftrightarrow  \quad \xi^2 \, \varepsilon(1/\xi)^{-1} \, \mu(1/\xi)^{-1}=|\bk|^{-2},   \mbox{ modulo } \xi=1/\omega.
\end{equation}
Then we introduce the rational function $\mathcal{G}(\xi)$: 
$$\mathcal{G}(\xi) :=  \xi^2 \, g(\xi), \quad g(\xi) := \varepsilon(1/\xi)^{-1} \, \mu(1/\xi)^{-1} \quad \mbox{ well defined for } \xi \in  \bbC \setminus \big\{ 1/z, z \in {\cal Z} \big\}$$
after having remarked that  $g(\xi)$ could be extended analytically at $\xi = 0$ via
$$  g(0)=\displaystyle \lim_{\xi\to 0} \varepsilon(\xi^{-1})^{-1} \mu(\xi^{-1})^{-1}= (\varepsilon_0 \mu_0)^{-1} =c^2\neq 0.$$ 
Thus,  we can now apply the Lemma \ref{Lem-implicte-function}, replacing $\omega$ par $\xi$,  with 
$$
z=0 , \quad g(\xi) := \varepsilon(1/\xi)^{-1} \, \mu(1/\xi)^{-1}, \quad \mathfrak{m}=2, \quad  A=g(0) = c^2 \mbox{ and  } \zeta=|\bk|^{-1}.$$
From Lemma \ref{Lem-implicte-function}, we deduce the existence of  two distinct analytic functions $\zeta \mapsto \xi_{\pm}(\zeta)$, defined in the vicinity of $0$, such that ${\cal G}(\xi_{\pm}(\zeta))=\zeta^2$ and such that
\begin{equation}\label{eq.asymptinvbranchinfty}
\xi_{\pm}(\zeta)=\pm \, c^{-1} \zeta \, \big(1+o(1)\big), \ \mbox{ as }\zeta \to 0.
\end{equation}
Thus, setting $\omega_{\pm \infty}(|\bk|) = \xi_{\pm}\big(|\bk|^{-1}\big)^{-1}$, we construct two branches of solutions of  $\mathcal{D}(\omega)=|\bk|^2$ which admit  by \eqref{eq.asymptinvbranchinfty} the asymptotic expansion \eqref{eq.pminfty}.\\[12pt] 
\noindent 
\noindent {\bf Step 3: conclusion.} In step 1, since the sum of the $\mathfrak{m}_p$'s over $p \in {\cal P}$ is equal to $2N_e+2N_m$, we have constructed for  $|\bk|$ large enough, $2N_e+2N_m$ solutions, namely $\big\{ \omega_{p,n}(|\bk|), p \in {\cal P},   n \leq \mathfrak{m}_p\}$ which are all distinct due to the asymptotics \eqref{eq.pole}. \\ [12pt]
In step 2, we have constructed for $|\bk|$ large enough, 2 additional solutions $\omega_{\pm \infty}(|\bk|)$ that are distinct thanks to  \eqref{eq.pminfty}. \\  [12pt] 
From both asymptotics \eqref{eq.pole}  and \eqref{eq.pminfty}, none of this two solutions can coincide with any of those from step 1. Therefore, with $\big\{ \omega_{p,n}(|\bk|), p \in {\cal P}, n \leq \mathfrak{m}_p\} \cup \big\{ \omega_{\pm \infty}(|\bk|\big\}$, we have constructed $2N_e+2N_m +2$  distinct solutions of \eqref{eq.disp}. Since \eqref{eq.disp}  is equivalent to a  polynomial equation of degree $2N_e+2N_m +2$, cf. \eqref{eq.disppolynom},  there are no other solutions. 
\end{proof}
\subsection{Spectral decomposition of the solution for $|\bk| \gg 1$ }\label{sec-decomp-sol-inf}
 In the physics literature, the solutions  of  the dispersion relation \eqref{eq.disp}     $|\bk| \to \omega_{p,n}(|\bk|) $ and $|\bk| \to \omega_{\pm \infty}(|\bk|)$  (given here  for large $|\bk|$ by Proposition \ref{prop.dispersioncurves}) are referred to as the dispersion curves. The asymptotics  \eqref{eq.pole} and \eqref{eq.pminfty} of these curves show that  they do not  cross each other for $|\bk| \gg 1$. Thus, combining Proposition \ref{prop.dispersioncurves} and Corollary \ref{eq.crit-diag} immediately yields the following property of the operator $ \bbA_{|\bk|,\perp}$.
\begin{Cor}\label{eq.crit-diag2}
There exists $k_+>0$  such that for $|\bk|\geq k_+$, $\bbA_{|\bk|,\perp}$ is diagonalizable on $\bC^N_{\perp}$.
\end{Cor}
\noindent To express  the diagonal decomposition of the  solution $\bbU(\bk,t)$, we split (see Section \ref{sec_dispersion-poles}) the sets of poles  $\mathcal{P}\subset \overline{\bbC^-}$ in three disjoint subsets : $\mathcal{P}=\mathcal{P}_-\cup \mathcal{P}_s\cup \mathcal{P}_d$ with
$$
\mathcal{P}_-:=\mathcal{P} \cap \bbC^-, \quad   \mathcal{P}_s:=\{ p \in \mathcal{P} \cap \bbR \mid \mathfrak{m}_p=1\}, \ \mbox{ and } \   \mathcal{P}_d:=\{ p \in \mathcal{P} \cap \bbR \mid \mathfrak{m}_p=2\}.
$$
We point out that the weak dissipation condition  \eqref{WD} implies that at least one pole of $\mathcal{P}$ lies in $\bbC^-$, thus $\mathcal{P}_-\neq \varnothing$ whereas the strong dissipation condition \eqref{SD} implies that  
all poles lie in $\bbC^-$, that is $\mathcal{P}_-=\mathcal{P}$. From Section \ref{sec_dispersion-poles}, one  also  has that  $\mathcal{P}_d=\calP_e\cap \calP_m\cap \bbR.$ \\ [12pt]
Using  Propositions \ref{Prop.spec}  and \ref{prop.dispersioncurves}, we introduce the following partition of the spectrum of $\bbA_{|\bk|,\perp}$ for $|\bk|\geq k_+$:
\begin{equation} \label{eq.specdisp} 
\left| \begin{array}{lll}
\sigma(\bbA_{|\bk|, \perp})&=&\big\{ \omega_{\pm \infty}(|\bk|) \big\} \cup \big\{ \omega_{p}(|\bk|), p\in \mathcal{P}_s \big\} \cup \big\{ \omega_{p,r}(|\bk|), p\in \mathcal{P}_d, \, r\in \{1,2\}  \big\} \nonumber \\[12pt]
&\cup &  \big\{\omega_{p,n}(|\bk|),  \ p \in \mathcal{P}_-, \,  n =1,\ldots, \mathfrak{m}_p  \big \} ,
\end{array} \right. 
\end{equation} 
(where  we set $\omega_{p}(|\bk|):=\omega_{p,1}(|\bk|)$ for a real  pole $p\in \mathcal{P}_s$ of multiplicity $\mathfrak{m}_p=1$.) We refer to the figure  \ref{fig-disp-curv} for an illustration of the behavior of the dispersion curves for $|\bk|\gg 1$.
\begin{figure}[h!]
		\begin{center}
			\includegraphics[scale=0.32]{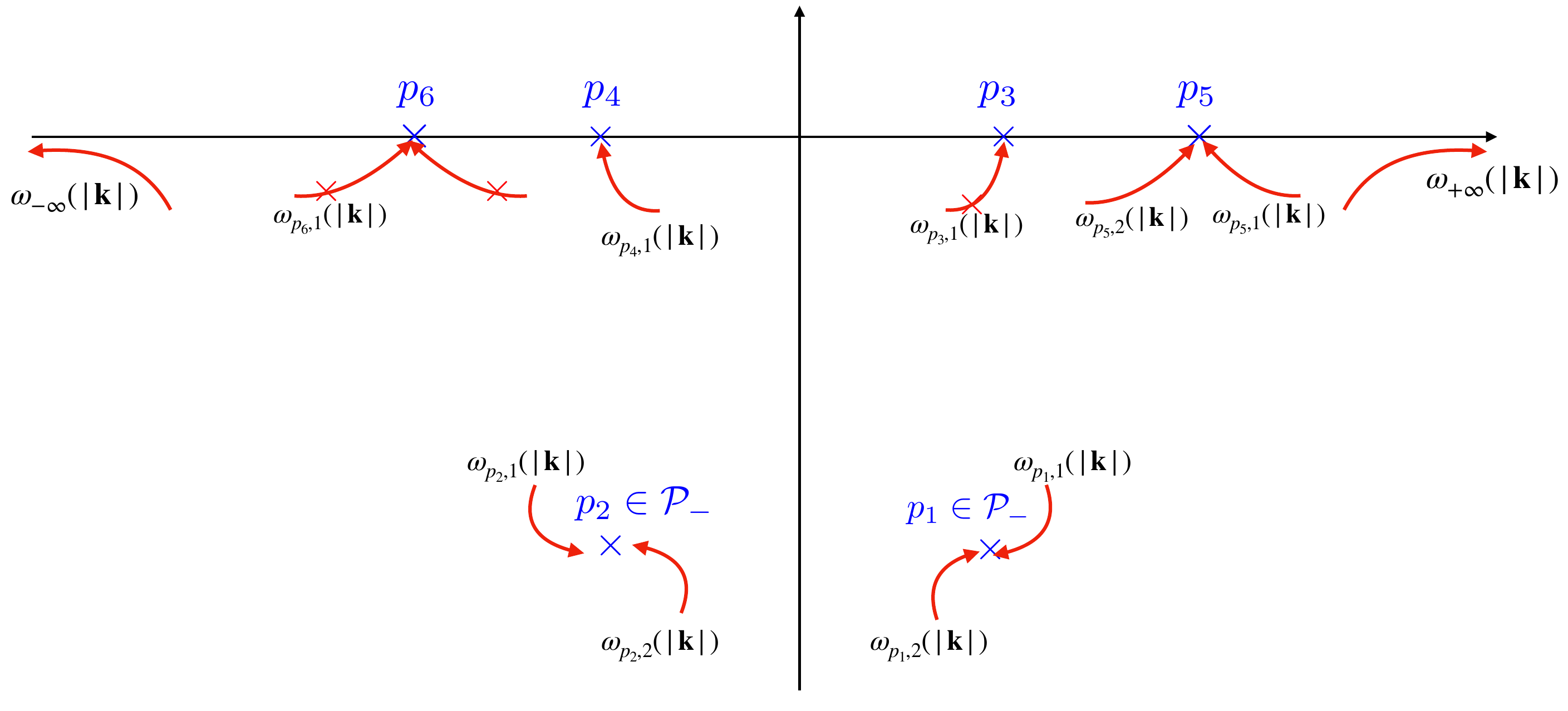}
		\end{center}
		\caption{Sketch of a configuration of  the dispersion curves   for large $|\bk|$ large in the case where $\mathcal{P}=\mathcal{P}_-\cup \mathcal{P}_s\cup \mathcal{P}_d$ with $\mathcal{P}_-=\{ p_1, p_2\}$, $\mathcal{P}_s=\{ p_3, p_4\}$ and $\mathcal{P}_d=\{ p_5, p_6\}$.  }
		\label{fig-disp-curv}	
\end{figure}
\noindent For $|\bk|\geq k_+$,  as $\bbA_{|\bk|, \perp}$ is diagonalizable by Corollary \ref{eq.crit-diag},  $\bC_{\perp}^N$ can be decomposed as
\begin{equation}\label{deq.decompCNperp}
 \bC_{\perp}^N=\bigoplus_{\pm} V_{|\bk|,\pm \infty}\oplus  \bigoplus_{p\in \mathcal{P}_{s}}V_{|\bk|,p} \oplus \bigoplus_{p\in \mathcal{P}_{d}} \mathop{\oplus}_{\substack {r=1}}^{2}V_{|\bk|,p,r}  \oplus \bigoplus_{p\in \mathcal{P}_{-}} \mathop{\oplus}_{\substack {n=1}}^{\mathfrak{m}_p}V_{|\bk|,p,n}
\end{equation}
\begin{equation}
\left\{ \begin{array}{llll}
V_{|\bk|,\pm \infty}=\operatorname{ker}\big( \bbA_{|\bk|, \perp} - \omega_{\pm \infty}(|\bk|)\, \mathrm{I}d\big), \quad & V_{|\bk|, p,n}=\operatorname{ker}\big( \bbA_{|\bk|, \perp} - \omega_{p,n}(|\bk|) \, \mathrm{I}d \big), \; n \leq \mathfrak{m}_p, \, \\[10pt]
 V_{|\bk|,p}= \operatorname{ker}\big( \bbA_{|\bk|, \perp} - \omega_{p}(|\bk|)\, \mathrm{I}d\big), &     V_{|\bk|, p,r}= \operatorname{ker}\big( \bbA_{|\bk|, \perp} - \omega_{p,r}(|\bk|)\, \mathrm{I}d \big),  \;  r=1, 2,
\end{array} \right.
\end{equation}
where  the above direct sums are (in general) non-orthogonal. Following these direct sums, one decomposes uniquely any vector $x\in  \bC_{\perp}^N$ as
 $$
\left| \;  \begin{array} {ll}
 x= \ds \sum_{\pm \infty } x_{ |\bk|, \pm \infty}+\sum_{p\in \mathcal{P}_{s}} x_{ |\bk|, p}+\sum_{p\in \mathcal{P}_{d}}\sum_{r=1}^2 x_{ |\bk|, p,r}  + \sum_{p\in  \mathcal{P}_{-}}  \sum_{n=1}^{\mathfrak{m}_p}x_{|\bk|, p, n}
\\ [18pt]
 x_{ |\bk|, \pm \infty} \in V_{|\bk|,\pm \infty}, \ x_{ |\bk|, p}\in V_{|\bk|,p}, \ x_{ |\bk|, p, r}\in V_{|\bk|,p,r}, \mbox{ and }\ x_{ |\bk|, p, n}\in V_{|\bk|,p,n}.
\end{array} \right. $$
Then, we  define  the spectral projectors   $\Pi_{p,n}(|\bk|)$, $p\in\mathcal{P}_-$ and $n\in \{ 1, \ldots, \mathfrak{m}_p\}$,  $\Pi_{\pm \infty}(|\bk|)$,   $\Pi_p(|\bk|)$ for $p\in \mathcal{P}_s$ and $\Pi_{p,r}(|\bk|)$ for $p\in \mathcal{P}_d$ and $r\in \{ 1,2\}$   associated respectively to the eigenvalues  $\omega_{p,n}(|\bk|),$ $\omega_{\pm \infty}(|\bk|)$, $\omega_p(|\bk|)$ and $\omega_{p,r}(|\bk|)$ by:
\begin{equation}\label{eq.projdef}
\left\{  \begin{array}{lll}\Pi_{\pm \infty}(|\bk|)(x)= x_{ |\bk|, \pm \infty}, \quad \Pi_{p,n}(|\bk|)(x)=x_{|\bk|,p,n},
\\ [10pt] \Pi_{p}(|\bk|)(x)=x_{|\bk|, p}, \quad \Pi_{p,r}(|\bk|)(x)=x_{|\bk|, p,r}. 
 \end{array}  \right. 
 \end{equation}
 From Proposition \ref{Prop.spec},  the geometric  multiplicity of each  eigenvalues in $\sigma(\bbA_{|\bk|,\perp})$ is two. Thus,  all the $\Pi_{p,n}(|\bk|)$, $\Pi_{\pm \infty}(|\bk|)$,   $\Pi_p(|\bk|)$ and $\Pi_{p,r}(|\bk|)$ are rank two projectors.  We emphasize  that the  dissipative operator $\bbA_{|\bk|,\perp}$ is not normal, thus its spectral projectors are not  orthogonal.
 
\noindent For $|\bk|\geq k_+$,  as $\bbA_{|\bk|, \perp}$ is diagonalizable (by Corollary \ref{eq.crit-diag}), one has
\begin{eqnarray*}
\bbA_{|\bk|,\perp}&=&\sum_{\pm }\omega_{\pm \infty}(|\bk|)  \Pi_{\pm \infty}(|\bk|) + \sum_{p\in \mathcal{P}_s}   \omega_{p}(|\bk|) \Pi_{p}(|\bk|)+\sum_{p\in \mathcal{P}_d} \sum_{r=1}^{2} \omega_{p,r}(|\bk|) \Pi_{p,r}(|\bk|)  \\
&& + \sum_{p\in \mathcal{P}_-}\sum_{n=1}^{\mathfrak{m}_p} \omega_{p,n}(|\bk|) \Pi_{p,n}(|\bk|).
\end{eqnarray*}
Thus,  for $|\bk|> k_+$, the solution $\bbU(\bk,t)$  given by \eqref{eq.refsolutionfourier2} can be expressed for all $ t \geq 0$    as

\begin{equation}\label{eq.decompositionfourterm}
\bbU(\bk,t)= \bbU_{\infty}(\bk,t)+\bbU_{s}(\bk,t)+\bbU_{d}(\bk,t) + \bbU_-(\bk,t)
\end{equation}
where
\begin{equation} \label{decompU}
\left\{ \begin{array}{llll}  
\bbU_{\infty}(\bk,t)&=& \ds \sum_{\pm } \rme^{-\rmi \, \omega_{\pm \infty}(|\bk|) \,t}  \mathcal{R}_{\bk}^* \, \Pi_{\pm \infty}(|\bk|)\, \mathcal{R}_{\bk} \,\bbU_0(\bk), & \quad(i)\\[19pt]
\bbU_{s}(\bk,t)&=& \ds \sum_{p\in \mathcal{P}_s}  \rme^{-\rmi \, \omega_{p}(|\bk|) \, t} \, \mathcal{R}_{\bk}^* \, \Pi_{p}(|\bk|) \mathcal{R}_{\bk}\,\bbU_0(\bk), &\quad(ii)\\[15pt]
\bbU_{d}(\bk,t)&=& \ds \sum_{p\in \mathcal{P}_d} \sum_{r=1}^{2}\rme^{-\rmi \omega_{p,r}(|\bk|) \, t}  \mathcal{R}_{\bk}^*\, \Pi_{p,r}(|\bk|) \,\mathcal{R}_{\bk} \,\bbU_0(\bk), &\quad(iii)\\[18pt]
\bbU_-(\bk,t)&= & \ds \sum_{p\in \mathcal{P}_-} \sum_{n=1}^{\mathfrak{m}_p} \rme^{-\rmi \, \omega_{p,n}(|\bk|) \,t} \, \mathcal{R}_{\bk}^* \, \Pi_{p,n}(|\bk|)\, \mathcal{R}_{\bk} \bbU_0(\bk) . &\quad (iv)
\end{array} \right. 
\end{equation}
\subsection{Estimates of $\bbU(\bk,t)$ for $|\bk| \gg 1$}\label{sec.estmHF}
In each subsection of this section, we shall provide a lower bound $k_+> 0$ for giving a sense to $|\bk| \gg 1$ via $|\bk| \geq k_+$. A priori, the value of $k_+$ will change from one section to the other but we can always choose a value for $k_+$ that is larger than its previous values. This convention will be adopted systematically without being explicitly mentioned.
\subsubsection{Orientation} \label{orientation} In what follows we are going to estimate successively, in Lemmas \ref{LemEstiinfty}, \ref{LemEstis} , \ref{LemEstid} and \ref{Estiminus}, each of the terms appearing in the decomposition \eqref{eq.decompositionfourterm}. \\ [12pt]
For  $\bbU_{\infty}(\bk,t), \bbU_{s}(\bk,t)$ and $\bbU_{d}(\bk,t)$, which will  be treated in Sections \ref{estiUinfty} to \ref{estiUd},
 we shall bound separately each of the terms of the sums  in \eqref{decompU}(i) to  \eqref{decompU}(iii). For each  eigenvalue  $\omega(|\bk|) \in \big\{  \omega_{\pm \infty}(|\bk|), \omega_{p}(|\bk|), \omega_{p,r}(|\bk|) \}$, we shall first 
 estimate in Lemmas \ref{LemProinfty}, \ref{LemPros}  and \ref{LemProd} the corresponding spectral projector $\Pi (|\bk|)$. Many approaches are possible. It appeared useful to use here  the expression of $\Pi (|\bk|)$ provided by the Riesz-Dunford functional calculus  (see  e.g.  \cite{Dun-88}, sections VII.1 and VII.3), in terms of a contour integral in the complex plane whose integrand involves the resolvent  $R_{|\bk|}(\omega)$ studied in Section \ref{sec-spec-prop}. This contour will be taken as a (positively oriented) circle ${\cal C}_{|\bk|}$ centered at  $\omega(|\bk|)$ 
 \begin{equation} \label{defcont}
 	{\cal C}_{|\bk|} := \big \{ \omega \in \bbC \; / \; \big| \, \omega - \omega(|\bk|)\,  \big| = \rho_{|\bk|} \big \},
 \end{equation} 
 	whose radius $\rho_{|\bk|}$ must be chosen in such a way that
 \begin{equation} \label{hypcont}
 \left\{ 	\begin{array}{l}
 		\mbox{(i) ${\cal C}_{|\bk|}$ does not enclose or intersect any other eigenvalue,}  \\ [8pt]
 	\mbox{(ii) ${\cal C}_{|\bk|}$ does not enclose or intersect any point of the set ${\cal S}_{\cal T}$}.
 		\end{array} \right. 
 	\end{equation} 
 This will be automatically achieved if we take $\rho_{|\bk|}$ equal to half of the distance of $\omega(|\bk|)$ to all the points that one wants to avoid, namely 
 \begin{equation} \label{choicerhok}
 	\rho_{|\bk|} = 1/2 \;\mbox{dist}\big(\omega(|\bk|), \Omega(|\bk|)\big), \quad  \Omega(|\bk|) =  \Big(\sigma(\bbA_{|\bk|,\perp}) \cup {\cal S}_{\cal T}\Big)  \setminus \{\omega(|\bk|)\}.
 	\end{equation}
We then have the formula
 \begin{equation}\label{eq.RieszDunford}
 	\Pi(|\bk|)=- \frac{1}{2\rmi \pi}\int_{\mathcal{C}_{|\bk|}} R_{|\bk|}(\omega) \, \rmd \omega = - \frac{1}{2\rmi \pi}\int_{\mathcal{C}_{|\bk|}} \mathcal{V}_{|\bk|}(\omega) \mathcal{S}_{|\bk|}(\omega) \, \rmd \omega, 
 \end{equation}
where the first equality is justified by \eqref{hypcont}(i) and the second by  the expression  of the resolvent: $R_{|\bk|}(\omega)=\mathcal{V}_{|\bk|}(\omega) \mathcal{S}_{|\bk|}(\omega)+\mathcal{T}(\omega)$  (given in Proposition \ref{Prop.res}) and by  \eqref{hypcont}(ii) since as  ${\cal S}_{\cal T}$ is defined as the set of singularities of $\omega \mapsto {\cal T}(\omega)$, this function is analytic in $\bbC\setminus {\cal S}_{\cal T}$. \\ [12pt]
From \eqref{eq.RieszDunford}, we deduce the inequality that we shall use systematically in the following, namely 
\begin{equation}\label{estiproj}
	\|\Pi(|\bk|)\| \leq \rho_{|\bk|} \, \sup_{\omega \in \mathcal{C}_{ |\bk| }}  \big(\|  \mathcal{V}_{|\bk|}(\omega) \|  \, \| \mathcal{S}_{|\bk|}(\omega)  \|\big).
\end{equation} 
where $\| \cdot\|$ denotes the operator norm of $\mathcal{L}(\bC_{\perp}^N)$.\\[12pt]
\noindent In a second step, we shall concentrate on the exponentials appearing in each factor whose estimate for large $|\bk|$ will rely on the asymptotic expansion of the eigenvalues (and more particularly their imaginary parts), see Lemmas \ref{LemEigeninfty}, \ref{LemEigend} and \ref{LemEigens}. 
These results will be transformed into sharp exponential decay estimates for  $|\bbU_{\infty}(\bk,t)|, |\bbU_{s}(\bk,t)|$ and $|\bbU_{d}(\bk,t)|$, in which the rate of decay degenerates for $|\bk| \rightarrow + \infty$, see Lemmas \ref{LemEstiinfty}, \ref{LemEstis} and \ref{LemEstid}. \\ [12pt]
Finally, the last term of  \eqref{decompU},  $\bbU_-(\bk,t)$, will not be treated by using \eqref{decompU}(iv) but an alternative expression directly issued from the Riesz-Dunford functional calculus  (Section \ref{estiUminus}). In fact, the exponential decay of $|\bbU_-(\bk,t)|$ will be, contrarily to the previous terms, uniform with respect to $|\bk|$, so that it will not contribute at the end to the large time equivalent of $\bU(\cdot,t)$. This is the reason why we can be satisfied with rough estimates.

\subsubsection{Estimates of $\bbU_{\pm \infty}(\bk,t)$ for $|\bk| \gg 1$} \label{estiUinfty}
In the following Lemma, we estimate the spectral projectors $\Pi_{\pm \infty} (|\bk|)$.
\begin{Lem} \label{LemProinfty}
The spectral projectors  $\Pi_{\pm \infty} (|\bk|)$ are uniformly bounded for large $|\bk|$.  
\end{Lem}
\begin{figure}[h!]
		\begin{center}
			\includegraphics[scale=0.33]{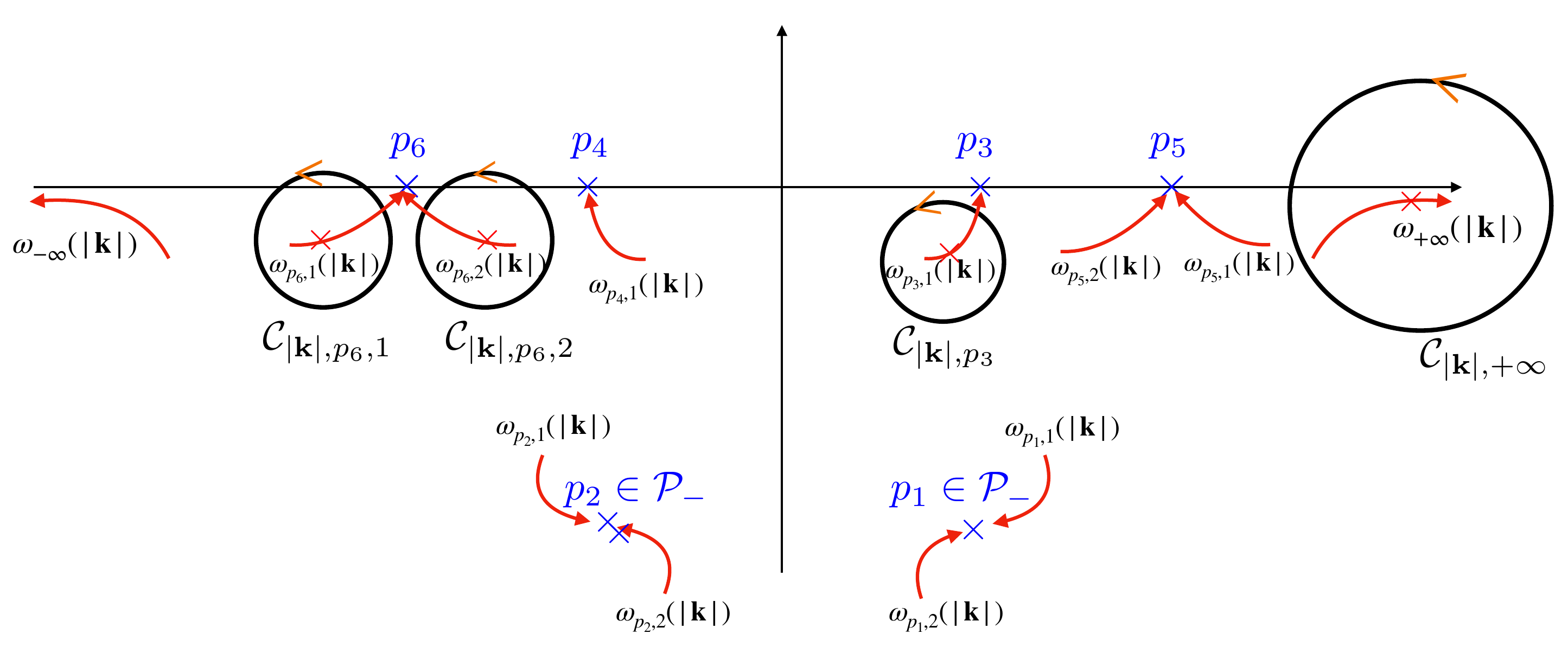}
		\end{center}
		\caption{Contours of integration $\mathcal{C}_{ |\bk|, +\infty,}$, $\mathcal{C}_{|\bk|, p_3}$, $\mathcal{C}_{|\bk|, p_6,1}$  and $\mathcal{C}_{|\bk|, p_6,2}$  used  for the estimate of   $\Pi_{+ \infty} (|\bk|)$, $\Pi_{p_3} (|\bk|)$ for $p_3\in \mathcal{P}_s$, $\Pi_{p_6,1} (|\bk|)$ and  $\Pi_{p_6,2} (|\bk|)$ for $p_6\in \mathcal{P}_d$ (corresponding to the figure \ref{fig-disp-curv}).}
\label{fig-contour-proj}	
\end{figure}
\begin{proof} We follow the approach described in Section \ref{orientation} for $\omega(|\bk|) = \omega_{\pm\infty}(|\bk|)$ and denote  $\mathcal{C}_{\pm\infty, |\bk|}$ the corresponding contour, see \eqref{defcont} and figure \ref{fig-contour-proj}. \\ [12pt]
		 {\bf Step 1: Estimate of $ \rho_{|\bk|}$}. Defining $ \rho_{|\bk|}$ by \eqref{choicerhok} for $\omega(|\bk|) = \omega_{\pm\infty}(|\bk|)$ and using the asymptotic behaviour  \eqref{eq.pminfty}, it is  clear that 
	\begin{equation} \label{estirhokinfty}
	\rho_{|\bk|}  \sim \mbox{$\frac{1}{2}$} \;  c |\bk|, \quad (|\bk| \rightarrow + \infty), \quad \mbox{thus} \quad \rho_{|\bk|}  \lesssim  |\bk| \quad \mbox{ for $|\bk| $ large enough}. 
\end{equation}  
\noindent {\bf Step 2: Estimate of $(\mathcal{D}(\omega)-|\bk|^2)^{-1}$.} This term is of interest because it appears in the expression \eqref{defS} of $\mathcal{S}_{|\bk|}(\omega)$. Using \eqref{eq.representationirreduc}, \eqref{eq.ireductibleF}, we compute that
\begin{equation}\label{eq.expressioncalDinv}
	(\mathcal{D}(\omega)-|\bk|^2)^{-1}=\frac{Q_e(\omega)\, Q_m(\omega)}{D_{|\bk|}(\omega)}, \quad \mbox{with $D_{|\bk|}(\omega)$ given in \eqref{eq.disppolynom}} .
\end{equation}
According to  Proposition  \ref{prop.dispersioncurves}, we know that $D_{|\bk|}(\omega)$ is a polynomial of degree $N$ with simple roots so that, according to \eqref{equivD}, it can be factorized as (note that we distinguish below $\omega_{\pm \infty}(|\bk|)$ that go to $\infty$ with $|\bk|$, from the other roots which remain bounded) 
\begin{equation}\label{eq.Dfactor}
	D_{|\bk|}(\omega)=  \varepsilon_0 \, \mu_0\,  D_{\mathrm{b},|\bk|}(\omega) \, \big(\omega-\omega_{+ \infty}(|\bk|)\big)\big(\omega-\omega_{-\infty}(|\bk|)\big) 
\end{equation}
with $D_{\mathrm{b},|\bk|}(\omega)$ given by
\begin{equation}\label{eq.Dbound} D_{\mathrm{b},|\bk|}(\omega)= \prod_{p\in \calP_-} \prod_{n=1}^{\mathfrak{m}_p}\big(\omega-\omega_{p,n}(|\bk|)\big) \, \prod_{p\in \calP_s} \big(\omega-\omega_{p}(|\bk|)\big)\, \prod_{p\in \calP_d} \prod_{r=1}^{2}\big(\omega-\omega_{p,r}(|\bk|)\big).
\end{equation}
\\ [12pt]
In the same way, thanks to  \eqref{eq.pole}, for any $ \omega(|\bk|) \in \big\{ \omega_{p,n}(|\bk|),  \omega_{p}(|\bk|), \omega_{p,r}(|\bk|) \big\}$, 
\begin{equation}\label{eq.boundomegamp2}
|\omega- \omega(|\bk|)| \geq   |\omega_{\pm \infty}(|\bk|) - \omega(|\bk|)| - |\omega-\omega_{\pm \infty}(|\bk|)| \sim  c \, |\bk|/2, \quad (|\bk| \rightarrow + \infty) .
 \end{equation}	
Thus, by definition \eqref{eq.Dbound} of $D_{\mathrm{b},|\bk|}(\omega)$, we deduce from \eqref{eq.boundomegamp} and \eqref{eq.boundomegamp2}, that 
 \begin{equation}\label{eq.boundDbdbis}
 |D_{\mathrm{b},|\bk|}(\omega)|  \geq C \;  |\bk|^{2(N_e+N_m)},  \quad  \forall \; \omega\in \mathcal{C}_{\pm \infty, |\bk|}, \quad \mbox{for some } C > 0.
\end{equation}
Next, if $ \omega \in \mathcal{C}_{\pm \infty, |\bk|}$, by the reverse triangular inequality
\begin{equation}\label{eq.boundomegamp}
	|\omega-\omega_{\mp \infty}(|\bk|)| \geq |\omega_{\pm \infty}(|\bk|) -\omega_{\mp \infty}(|\bk|)| - |\omega-\omega_{\pm \infty}(|\bk|)|  \sim 3 \, c \, |\bk|/2, \quad (|\bk| \rightarrow + \infty) 
\end{equation}	
since, by \eqref{eq.pminfty}, $|\omega_{\pm \infty}(|\bk|) -\omega_{\mp \infty}(|\bk|)| \sim 2 \, c \, |\bk|$ and $|\omega-\omega_{\pm \infty}(|\bk|)| = \rho_{|\bk|} \sim c \, |\bk| / 2$. 
Thus, using  \eqref{eq.boundomegamp} and \eqref{eq.boundDbdbis} in \eqref{eq.Dfactor} yields, as $|\omega-\omega_{\pm \infty}(|\bk|)| = \rho_{|\bk|}$ on $\mathcal{C}_{\pm \infty, |\bk|}$,
\begin{equation}\label{eq.Dboundinvbis}
|D_{|\bk|}(\omega)|^{-1}\lesssim  \;  \rho_{|\bk|}^{-1}  \;   |\bk|^{-(2(N_e+N_m)+1)},  \quad \forall \; \omega\in \mathcal{C}_{\pm \infty, |\bk|}.
\end{equation}
For bounding $Q_e(\omega) Q_m(\omega)$, that appears in \eqref{eq.expressioncalDinv}, we use an upper bound for $\omega \in \mathcal{C}_{\pm \infty, |\bk|}$
\begin{equation}\label{eq.upperboundomega}
 |\omega|\leq|\omega_{\pm\infty}(|\bk|)| + \rho_{|\bk|} \sim  3 \,  {c|\bk|}/{2} , \quad |\bk| \rightarrow + \infty,  \quad \forall \; \omega\in \mathcal{C}_{\pm \infty, |\bk|},
\end{equation}
to deduce that, as $Q_e Q_m$ is a polynomial of degree $2(N_e+N_m)$, 
\begin{equation}\label{boundQeQm}
|Q_e(\omega) \,Q_m(\omega)| \lesssim |\omega|^{2(N_e+N_m)} \lesssim |\bk|^{2(N_e+N_m)},  \quad \forall \, \omega\in \mathcal{C}_{\pm \infty, |\bk|}.
\end{equation}
Thus, using \eqref{boundQeQm} and  \eqref{eq.Dboundinvbis} in  \eqref{eq.expressioncalDinv}, we get, for $|\bk|$ large enough,
\begin{equation}\label{boundDcalpm}
|(\mathcal{D}(\omega)-|\bk|^2)^{-1}| \lesssim   \rho_{|\bk|}^{-1}  \; |\bk|^{-1},  \quad \forall \; \omega\in \mathcal{C}_{\pm \infty, |\bk|}. 
\end{equation}
\noindent {\bf Step 3: Estimates of $\mathcal{S}_{|\bk|}(\omega)$ and $\mathcal{V}_{|\bk|}(\omega)$.} 
Owing to \eqref{defS}, to estimate $\mathcal{S}_{|\bk|}(\omega)$, it remains to estimate $\omega \mu(\omega ) \bbA_e(\omega) -|\bk| \, {\bf e_3} \times \bbA_m(\omega)$. We claim that,
 \begin{equation}\label{eq.estimSoppartpminf2}
	\mbox{ for $|\bk|$ large enough, } \quad \big\|   \omega \mu(\omega ) \bbA_e(\omega) -|\bk| \, {\bf e_3} \times \bbA_m(\omega) \big\| \lesssim |\bk| , \quad  \forall\,  \omega \in     \mathcal{C}_{\pm \infty, |\bk|},
\end{equation}
which, combined with  \eqref{boundDcalpm} provides, via \eqref{defS}, 
\begin{equation}\label{eq.estimSpm}
	\|\mathcal{S}_{|\bk|}(\omega)  \|\lesssim    \rho_{|\bk|}^{-1},  \quad  \forall \; \omega\in \mathcal{C}_{\pm \infty, |\bk|} .
\end{equation}
To prove \eqref{eq.estimSoppartpminf2}, notice first that along $\mathcal{C}_{\pm \infty, |\bk|}$, 
\begin{equation}\label{eq.omegainflim}
|\omega| \geq |\omega_{\pm \infty}(|\bk|)| - \rho_{|\bk|} \sim c \, |\bk|/2 \ \mbox{ when } \ |\bk| \rightarrow + \infty.
\end{equation}
Thus, as $\mu(\omega)=\mu_0+o(1)$ as $|\omega|\to \infty$, with  the upper 
bound \eqref{eq.upperboundomega}, one gets
$$
\big\| \, \omega \mu(\omega ) \bbA_e(\omega) -|\bk| \, {\bf e_3} \times \bbA_m(\omega) \big\| \lesssim |\bk| \; \big( \, \big\| \bbA_e(\omega) \big\| + \big \|  \bbA_m(\omega) \big\|  \, \big), \quad  \forall\,  \omega \in     \mathcal{C}_{\pm \infty, |\bk|}.
$$
Finally, we prove $  \bbA_e(\omega) $ and $\bbA_m(\omega) $ are uniformly bounded (in $\bk$) on $\mathcal{C}_{\pm \infty, |\bk|}$. Towards this goal, we first estimate  the operators $(\bbA_{e,j}(\omega), \dot \bbA_{e,j}(\omega),  \bbA_{m,\ell}(\omega), \dot  \bbA_{m,\ell}(\omega))$. From \eqref{operatorsApAm}, 
\begin{equation} \label{boundoperators} 
\|\bbA_{e,j}(\omega)\|, \|\dot \bbA_{e,j}(\omega)\| \lesssim \frac{1 + |\omega|}{|q_{e,j}(\omega)|}, \quad \|\bbA_{m,\ell}(\omega)\|, \|\dot \bbA_{m,\ell}(\omega)\| \lesssim \frac{1 + |\omega|}{|q_{m,\ell}(\omega|)}, 
\end{equation} 
 Thus, using  \eqref{eq.omegainflim}, as $q_{e,j}(\omega), q_{m,\ell}(\omega) \sim \omega^2$  when $|\omega| \rightarrow + \infty$ (see \eqref{eq.polynom}), one has
\begin{equation}\label{estimq}
	|q_{e,j} (\omega)|^{-1}\lesssim |\bk|^{-2}, \quad |q_{m,\ell} (\omega)|^{-1}\lesssim |\bk|^{-2} \quad \mbox{ for } \omega \in \mathcal{C}_{\pm \infty, |\bk|}.\end{equation}
Thus, using the above inequalities and the upper bound \eqref{eq.upperboundomega} in \eqref{boundoperators}, we get
$$
\|\bbA_{e,j}(\omega)\|, \|\dot \bbA_{e,j}(\omega)\| \lesssim  |\bk|^{-1}, \quad \|\bbA_{m,\ell}(\omega)\|, \|\dot \bbA_{m,\ell}(\omega)\| \lesssim  |\bk|^{-1}.
$$
The formulas  \eqref{operatorsAeAh} show that $ \bbA_e(\omega) $ and $\bbA_m(\omega)$ are the sums of a fixed (independent of $\omega$) operator with a (fixed) linear combination of 
$({\bf e}, {\bf h}, \bbA_{e,j}(\omega), \dot \bbA_{e,j}(\omega), \bbA_{m,\ell}(\omega), \dot \bbA_{m,\ell}(\omega))$. 
They are thus uniformly bounded (in $\bk$) on $\mathcal{C}_{\pm \infty, |\bk|}$, proving \eqref{eq.estimSoppartpminf2}.
 \\[12pt]
 Finally, using  \eqref{eq.upperboundomega},  \eqref{estimq} and  that, on $ \mathcal{C}_{\pm \infty, |\bk|}$, $|\omega \mu(\omega)|^{-1} \lesssim  \mu_0^{-1} \,| \bk|^{-1}$, when $ |\bk| \rightarrow + \infty$, one sees on the  definition  \eqref{defV} of the operator $\mathcal{V}_{|\bk|}(\omega)$ that, for $|\bk|$ large enough, 
 \begin{equation}\label{eq.estimVpm}
 \|\mathcal{V}_{|\bk|}(\omega)  \|\lesssim    1, \quad  \forall \; \omega\in \mathcal{C}_{\pm \infty, |\bk|}.
 \end{equation}
 \noindent {\bf Conclusion.} It suffices to use \eqref{eq.estimSpm}  and  \eqref{eq.estimVpm} in \eqref{estiproj} for $\Pi(|\bk|) = \Pi_{\pm \infty} (|\bk|)$, to conclude that  $\|\Pi_{\pm \infty} (|\bk|)\|\lesssim 1$ for $|\bk|$ large enough. \end{proof}
\noindent The next lemma is about the asymptotic expansion (in powers of $|\bk|$)) of  $\omega_{\pm \infty}(|\bk|)$ when $|\bk|\to +\infty$. It is important  to push the expansion up to the first apparition of a negative imaginary part, since it will govern the decay  of $\bbU_{\infty}(|\bk|,t)$ for large $|\bk|$. 
\begin{Lem}  \label{LemEigeninfty}
	The eigenvalues $\omega_{\pm \infty}(|\bk|)$ satisfy the following asymptotic expansion
	\begin{equation}\label{eq.asymptinf}
	\omega_{\pm \infty}(|\bk|)=\pm \, c \,  |\bk|\pm \frac{A_{1,\infty}}{2\,c}\,  |\bk|^{-1}-\rmi \, \frac{A_{2,\infty}}{2\, c^2}\, |\bk|^{-2}+o(|\bk|^{-2}),\ \mbox{ as } |\bk|\to +\infty,
	\end{equation} 
where  $A_{1,\infty}$ and $A_{2,\infty}$ are two real constants  given by
\begin{equation}\label{eq.constantinf}
A_{1,\infty}= \sum \; \Omega_{e,j}^2+ \sum \; \Omega_{m,\ell}^2 \quad \mbox{and}\quad  A_{2,\infty}= \sum \; \alpha_{e,j}\, \Omega_{e,j}^2+ \sum \; \alpha_{m,\ell}\,\Omega_{m,\ell}^2,
\end{equation}
with  $A_{1,\infty} > 0$ and $ A_{2,\infty} > 0$ thanks to the weak dissipation condition \eqref{WD}. 
In particular 
\begin{equation} \label{eq.asymptinfIm}
	\operatorname{Im}(\omega_{\pm \infty}(|\bk|))=-\frac { A_{2,\infty}} {2\,c^{2}}\, |\bk|^{-2}  +o(|\bk|^{-2}), \ \mbox{ as } |\bk | \to +\infty .
\end{equation} 
\end{Lem}
\begin{proof}
Since $\omega_{-\infty}(|\bk|)=-\overline{\omega_{+\infty}(|\bk|)}$ (see Section \ref{sec_dispersion-1}),  we only prove \eqref{eq.asymptinf} for $$\omega_{+\infty}(|\bk|)= \xi_{+}\big(|\bk|^{-1}\big)^{-1},  \quad \mbox{(cf. the proof of Proposition \ref{prop.dispersioncurves}, step 2). }$$
\noindent {\bf Step 1: Proof of $\omega_{+ \infty}(|\bk|)= c \,  |\bk|+ O(|\bk|^{-1})$}\\[4pt]
The proof of Proposition \ref{prop.dispersioncurves} defines the functions $\zeta \mapsto \xi_{\pm}(\zeta)$  as the branches of solutions  of the equation $G(\xi)=\zeta^2$ near  $\xi=0$, where $G(\xi)=\xi^2 g(\xi)$ and $ g(\xi)=  \varepsilon(1/\xi)^{-1} \mu(1/\xi)^{-1}$  is analytic near $0$. In particular $g'(0) = 0$ and, by Lemma   \ref{Lem-implicte-function} (with  $z=0, g(0)=c^2$ and $\mathfrak{m}=2$), and more precidely \eqref{eq.asympexpansion}, we deduce that
\begin{equation} \label{expxi0}
\xi_{+}(\zeta)=c^{-1} \, \zeta  \big(1 + O(\zeta^2)\big)  \mbox{ near 0},  \quad \mbox{thus} \; \; \omega_{+ \infty}(|\bk|)= c \,  |\bk|+ O(|\bk|^{-1}) \mbox{ at $\infty$. } 
\end{equation} 
\noindent {\bf Step 2: Computations of the terms of order $|\bk|^{-1}$ and $|\bk|^{-2}$}\\[4pt]
By analyticity of $\xi_+(\zeta)$, and using the step 1, we know that, for some coefficients $(A_1, A_2)$, 
\begin{equation} \label{expxi}
\xi_{+}(\zeta)=c^{-1} \zeta \,  \big(1 +  A_1	 \, \zeta^2 + A_2 \, \zeta^3 \, + O(\zeta^4)\big).
\end{equation}
In order to compute $A_1$ and $A_2$, we are going to proceed by identification in the equation
\begin{equation} \label{identif}
\xi_{+}(\zeta)^2 \, g \big( \xi_{+}(\zeta) \big) = \zeta^2.
\end{equation}
Let us compute the Taylor expansion of the left hand side of \eqref{identif}. By 
(\ref{eq.permmitivity-permeabiity},\ref{eq.polynom}), one has
\begin{equation*}
\begin{array}{llll}
 \varepsilon(1/\xi)&=&  \displaystyle \varepsilon_0 \; \Big[ \, 1 -\Big(\sum\Omega_{e,j}^2\Big)\, \xi^2 \; \, +\rmi  \,\Big(\sum \alpha_{e,j} \, \Omega_{e,j}^2 \Big)\, \xi^3 \; \, +O(\xi^4) \, \Big],& \quad \mbox{ as } \xi \to 0,\\[6pt]
 \mu(1/\xi) &=& \displaystyle  \mu_0 \; \Big[ \, 1 -\Big(\sum \Omega_{m,\ell}^2\Big)\, \xi^2+\rmi  \,\Big(\sum \alpha_{m,\ell}\, \Omega_{m,\ell}^2\Big)  \, \xi^3+O(\xi^4) \, \Big], &\quad \mbox{ as } \xi \to 0.
\end{array}
\end{equation*}
It follows that, with $(A_{1,\infty}, A_{2,\infty})$ the coefficients defined by \eqref{eq.constantinf}, 
\begin{equation}\label{eq.taylorg2}
g(\xi)=c^{2} \, \big( 1+A_{1,\infty} \, \xi^2-\rmi\, \,A_{2,\infty}\,  \xi^3+O(\xi^4) \big), \quad \mbox{ as } \xi \to 0,
\end{equation}
that is to say, as $\xi(\zeta) = O\big(\zeta)$,
\begin{equation}\label{eq.taylorg}
	\displaystyle \xi_{+}(\zeta)^2 \, g \big( \xi_{+}(\zeta) \big) =c^{2} \, \big(  \xi_{+}(\zeta)^2 +A_{1,\infty}\, \xi_{+}(\zeta)^4-\rmi\, \, A_{2,\infty}\, \xi_{+}(\zeta)^5+O(\zeta^6) \big), \quad \mbox{ as } \zeta \to 0.
\end{equation}
On the other hand, using \eqref{expxi}, we have, truncating the expansions at the order 5 in $\zeta$, 
$$
\begin{array}{lll}
\xi_{+}(\zeta)^2  = c^{-2} \, \zeta^2 \,  \big(1 +  2 \,  A_1 \, \zeta^2  + 2 \,  A_2 \, \zeta^3+ O(\zeta^4)\big), \\ [8pt]
\xi_{+}(\zeta)^4  = c^{-4} \, \zeta^4 \,  \big(1 + O(\zeta^2)\big),  \quad
\xi_{+}(\zeta)^5  = c^{-5}\,  \zeta^5 \,  \big(1 +  O(\zeta^2)\big), 
\end{array}
$$
which we can substitute into \eqref{eq.taylorg} to obtain
\begin{equation}\label{eq.taylor3}
	\displaystyle \xi_{+}(\zeta)^2 \, g \big( \xi_{+}(\zeta) \big)  = \zeta^2 \, \big( 1 + 2 \,  A_1 \, \zeta^2  + 2 \,  A_2 \, \zeta^3 \big) 
	+ A_{1,\infty}\, c^{-2} \, \zeta^4 -  \rmi\, \,  A_{2,\infty} \, c^{-3} \, \zeta^5 + O (\zeta^6).
\end{equation}
Using the above in  \eqref{identif} leads to the equations
$$
2 \, A_1 + A_{1,\infty} \, c^{-2} = 0, \quad 2 \, A_2 -  \rmi\, \, A_{2,\infty}\, c^{-3} = 0.
$$
from which we conclude easily. The remaining details are left to the reader.
\end{proof}
\noindent We conclude this section by estimating the term $ \bbU_{\infty}(\bk,t)$ in  the decomposition \eqref{eq.decompositionfourterm} .
\begin{Lem}\label{LemEstiinfty}
For any constant  $C\in(0, A_{2,\infty}/(2 c^2))$  with $ A_{2,\infty}$ given by  \eqref{eq.constantinf}, there exists $k_+>0$  such that for $|\bk|\geq k_+$, the functions $\bbU_{ \infty}(\bk,t)$ defined by \eqref{eq.decompositionfourterm} and \eqref{decompU}  satisfy
	\begin{equation}\label{eq.estimateUm0}
		| \bbU_{\infty}(\bk,t)| \lesssim \rme^{-\frac{C}{|\bk|^2}\,  t} \, | \bbU_0(\bk)|, \quad \, \forall \; t\geq 0.
	\end{equation}
\end{Lem}
\begin{proof}
The functions $\bbU_{ \infty}(\bk,t)$ are defined  by \eqref{eq.decompositionfourterm} and \eqref{decompU} for $|\bk|$ large enough. As the operator $\mathcal{R}_{\bk}$ and its inverse $\mathcal{R}_{\bk}^*$ are unitary, one has 
$$
|\bbU_{ \infty}(\bk,t)| \leq   \sum_{\pm} \rme^{\operatorname{Im}(\omega_{\pm \infty}(|\bk|))\,  t} \ \|\Pi_{\pm \infty} (|\bk|)\| \ | \bbU_0(\bk)|, \quad  \forall \; t\geq 0.
$$
Thus, the upper bound \eqref{eq.estimateUm0}  follows  immediately from \eqref{eq.asymptinfIm} (Lemma \ref{LemEigeninfty}) and the inequality $\|\Pi_{\pm \infty} (|\bk|)\|\lesssim 1 $ (Lemma \ref{LemProinfty}).
\end{proof}
\subsubsection{Estimates of $\bbU_{s}(\bk,t)$ for $|\bk| \gg 1$} \label{estiUs}
\noindent We estimate now the term $\bbU_{s}(\bk,t)$ in   \eqref{eq.decompositionfourterm}  which involves in particular the projectors $\Pi_{p} (|\bk|)$.
\begin{Lem} \label{LemPros}
	The spectral projectors  $\Pi_{p} (|\bk|)$, $p \in {\cal P}_s$ are uniformly bounded for large $|\bk|$.  
\end{Lem}
 \begin{proof}   We follow the approach described in Section \ref{orientation} for $\omega(|\bk|) = \omega_{p}(|\bk|)$, $p\in \mathcal{P}_s$ and denote  $\mathcal{C}_{p, |\bk|}$ the corresponding contour, see \eqref{defcont} and figure \ref{fig-contour-proj}.\\ [12pt]
 	{\bf Step 1: estimate of $\rho_{|\bk|}$}. Obviously, for $|\bk|$ large enough, since $\omega_{p}(|\bk|)$ tends to $p$ when $|\bk|\to+\infty$, the infimum defining $\rho_{|\bk|}$ in \eqref{choicerhok} is attained at the point $p$, i.e.:
 		\begin{equation} \label{estirhokp}
 		\rho_{|\bk|} =   \mbox{$\frac{1}{2}$} \; | \, \omega_{p}(|\bk|) - p \, | \sim  \mbox{$\frac{1}{2}$} \;  |A_p| \; |\bk|^{-2} \quad (|\bk| \rightarrow + \infty), \quad \mbox{with} \quad  |A_p|> 0,
 	\end{equation}  
 according to \eqref{eq.pole} for $\mathfrak{m}_p = 1$ since $p\in \mathcal{P}_s$. 
\\[12pt]
For the rest of the proof, we essentially follow the proof of Lemma \ref{LemProinfty}. \\ [12pt] 
\noindent {\bf Step 2: Estimate of \big($\mathcal{D}(\omega)-|\bk|^2\big)^{-1}$.} 
\\[6pt]
First, due to \eqref{eq.pole} and \eqref{estirhokp}, it is clear that there exists $C>0$ such that , for $|\bk|$ large enough,
$$
\forall \; \tilde{p}\in \mathcal{P}\setminus \{ p\}, \quad \forall \; \omega \in \mathcal{C}_{p, |\bk|}, \quad \big|\,\omega- \omega_{\tilde{p},n}(|\bk|) \big|, \big|\,\omega- \omega_{\tilde{p}}(|\bk|)|, \big|\, \omega- \omega_{\tilde{p},r}(|\bk|)\big| \geq C > 0.
$$
Thus, for $|\bk|$ large enough, the function  $D_{b,|\bk|}(\omega)$ given by the product \eqref{eq.Dbound} satisfies 
\begin{equation}\label{boundDb} \forall \; \omega \in \mathcal{C}_{p, |\bk|}, \quad |D_{b,|\bk|}(\omega)| \geq C \; \rho_{|\bk|} \quad \mbox{thus} \quad |D_{b,|\bk|}(\omega)|^{-1}\lesssim \rho_{|\bk|}^{-1}. \end{equation}
On the other hand, due to \eqref{eq.pole} and \eqref{estirhokp}, 
$$
|\omega-\omega_{\pm \infty}(|\bk|)|=c\, |\bk|+o(|\bk|) \quad \mbox{ uniformly for } \omega \in \mathcal{C}_{p, |\bk|}.
$$
It is then immediate to infer from the definition \eqref{eq.Dfactor} of $D_{|\bk|}(\omega)$ that,  for $|\bk|$ large enough,
\begin{equation}\label{eq.Dboundinvp}
	|D_{|\bk|}(\omega)|^{-1}\lesssim    \rho_{|\bk|}^{-1}   \,  |\bk|^{-2},  \ \forall \; \omega\in \mathcal{C}_{p, |\bk|}. 
\end{equation}
Moreover,  $p\in \mathcal{P}_s$ is a simple zero of the product $Q_e\, Q_m$. Hence, for any $\omega\in \mathcal{C}_{p, |\bk|}$
\begin{equation}\label{eq.boundQps}
|Q_e(\omega)\, Q_m(\omega)| \lesssim  |\omega-p|\lesssim \rho_{ |\bk|}+|\, \omega_p(|\bk|)-p \,| \sim \mbox{$\frac{3}{2}$} \;  |A_p| \; |\bk|^{-2} 
\end{equation}
As a consequence, from the expression \eqref{eq.expressioncalDinv} of $(\mathcal{D}(\omega)-|\bk|^2)^{-1}$, it is clear that 
\begin{equation}\label{boundDcalpm2}
	|(\mathcal{D}(\omega)-|\bk|^2)^{-1}| \lesssim   \rho_{|\bk|}^{-1}  \; |\bk|^{-4},  \quad \forall \; \omega\in \mathcal{C}_{p, |\bk|} . 
\end{equation}
\noindent {\bf Step 3: Estimates of $\mathcal{S}_{|\bk|}(\omega)$ and $\mathcal{V}_{|\bk|}(\omega)$.} \\ [12pt] 
To estimate $\mathcal{S}_{|\bk|}(\omega)$, it remains to estimate $\omega \mu(\omega ) \bbA_e(\omega) -|\bk| \, {\bf e_3} \times \bbA_m(\omega)$. Due to this, we shall be led to distinguish two cases,  $p\in \mathcal{P}_m$ and  $p\in \mathcal{P}_e$,  and show that, for
$|\bk|$ large enough,  
\begin{equation}\label{eq.estimSps}
\left\{	\begin{array}{lllll}
	\mbox{If } p \in {\cal P}_m, & \ \forall \;  \omega\in \mathcal{C}_{p, |\bk|}, & \quad \|\mathcal{S}_{|\bk|}(\omega)  \|\lesssim  \rho^{-1}_{|\bk|} \  |\bk|^{-1}, & \quad  \|\mathcal{V}_{|\bk|}(\omega)  \|\lesssim    |\bk|, & \quad (i) \\ [12pt]
	\mbox{If } p \in {\cal P}_e, & \ \forall \;  \omega\in \mathcal{C}_{p,|\bk|}, & \quad \|\mathcal{S}_{|\bk|}(\omega)  \|\lesssim  \rho^{-1}_{|\bk|} \  |\bk|^{-2}, & \quad  \|\mathcal{V}_{|\bk|}(\omega)  \|\lesssim    |\bk|^2. & \quad (ii) 
	\end{array}  \right. 
\end{equation}
{\it Proof of \eqref{eq.estimSps}(i)}. As $p \in {\cal P}_m \cap {\cal P}_s$, $p \notin {\cal P}_e$ and by $\mathrm{H}_1$, there is a unique index $\ell_0$ such that $p = \pm \omega_{m,\ell_0}$. Thanks to  \eqref{boundoperators}, all the operators $(\bbA_{e,j}(\omega), \dot \bbA_{e,j}(\omega),  \bbA_{m,\ell}(\omega), \dot  \bbA_{m,\ell}(\omega))$ are uniformly bounded along $\mathcal{C}_{p, |\bk|}$ except $(\bbA_{m,\ell_0}(\omega), \dot  \bbA_{m,\ell_0}(\omega)\big)$, as $p$ is a simple pole of $\omega \mapsto \bbA_{m,\ell_0}(\omega)$ and $\omega \mapsto \dot  \bbA_{m,\ell_0}(\omega)$. More precisely, as  $|q_{m,\ell_0}(\omega)|^{-1}\sim|\omega-p|^{-1} \mbox{ when } \omega \rightarrow p$, for $|\bk|$ large enough, 
\begin{equation}\label{eq.qeqmps}
	\forall \; \omega \in \mathcal{C}_{p, |\bk|}, \quad |q_{m,\ell_0}(\omega)|^{-1}\lesssim|\omega-p|^{-1} =  \rho^{-1}_{|\bk|} \lesssim |\bk|^{2}, \quad \mbox{according to \eqref{estirhokp}} 
\end{equation}
Thus, using the bounds \eqref{boundoperators} and the fact that $\mathcal{C}_{p, |\bk|}$ lies in a bounded set, 
\begin{equation} \label{boundA0} 
\forall \; \omega \in \mathcal{C}_{p, |\bk|}, \quad	\|\bbA_{m,\ell_0}(\omega)\|, \|\dot  \bbA_{m,\ell_0}(\omega)\| \lesssim |\bk|^{2}.
	\end{equation} 
Thanks to formulas \eqref{operatorsAeAh} for $(\bbA_{e}(\omega),\bbA_{m}(\omega) )$, we decuce that 
\begin{equation}\label{estiAeAm}
\forall \; \omega \in \mathcal{C}_{p, |\bk|}, \quad \|\bbA_{e}(\omega)\| \lesssim 1, \quad \|  \bbA_{m}(\omega)\| \lesssim |\bk|^{2}.
\end{equation}
Moreover, as $p \in \mathcal{P}_m$ is a simple pole of $\mu(\omega)$ and $\mathcal{C}_{p, |\bk|}$ is uniformly bounded: 
\begin{equation} \label{estimmu} 
|\mu(\omega)| \sim \mu_p  \, |\omega-p|^{-1}, \quad  \mbox{with $\mu _p > 0$}, \quad \mbox{thus} \quad 
|\mu(\omega)| \lesssim |\omega-p|^{-1}\lesssim |\bk|^{2} \mbox{ along } \mathcal{C}_{p, |\bk|},
\end{equation}
which, joined to \eqref{estiAeAm}, gives (the dominant term is the second one) 
\begin{equation} \label{estimAeAmbis} 
\forall \; \omega \in \mathcal{C}_{p, |\bk|}, \quad\| \omega \mu(\omega ) \bbA_e(\omega) -|\bk| \, {\bf e_3} \times \bbA_m(\omega)\| \lesssim |\bk|^{3}.
\end{equation}
The estimate for $\mathcal{S}_{|\bk|}(\omega)$ in \eqref{eq.estimSps}(i) then follows from \eqref{boundDcalpm2},  \eqref{estimAeAmbis}  and \eqref{defS}.\\ [12pt]
For $\mathcal{V}_{|\bk|}(\omega)$, we first observe that $|\mu(\omega)| \sim \mu_p  \, |\omega-p|^{-1}$ when $\omega \rightarrow  p$ also implies that
\begin{equation} \label{boundmu2}
|\mu(\omega)|^{-1} \lesssim |\bk|^{-2} \mbox{ along } \mathcal{C}_{p, |\bk|}.
\end{equation}
On the other hand one sees on formula \eqref{defV} that the dominant term in $\mathcal{V}_{|\bk|}(\omega)$ on $\mathcal{C}_{p, |\bk|}$ is asymptotically proportional to the function
$$
|\bk| \; \mu(\omega)^{-1} \; q_{m,\ell_0}(\omega)^{-1}.
$$
This explains  the estimate for $\mathcal{V}_{|\bk|}(\omega)$ in \eqref{eq.estimSps}(i) since, according to the inequalities \eqref{eq.qeqmps} and \eqref{boundmu2}, $|\mu(\omega)|^{-1} \;| q_{m,\ell_0}(\omega)^{-1} |\lesssim 1$ on  $\mathcal{C}_{p, |\bk|}$.\\ [12pt]
{\it Proof of \eqref{eq.estimSps}(ii)}. The proof is very similar to that of $(i)$ and we shall simply point out the difference with $(i)$. As $p \in {\cal P}_e \cap {\cal P}_s$, $p \notin {\cal P}_m$ and there is a unique index $j_0$ such that $p = \pm \omega_{e,j_0}$. \\ [12pt]
This time, all the operators $(\bbA_{e,j}(\omega), \dot \bbA_{e,j}(\omega),  \bbA_{m,\ell}(\omega), \dot  \bbA_{m,\ell}(\omega))$ are uniformly bounded along $\mathcal{C}_{p, |\bk|}$ except $(\bbA_{e,j_0}(\omega), \dot  \bbA_{e,j_0}(\omega)\big)$ and, similarly to \eqref{estiAeAm}, one shows that 
\begin{equation}\label{estiAeAmbis}
	\|\bbA_{e}(\omega)\| \lesssim |\bk|^{2}, \quad \|  \bbA_{m}(\omega)\| \lesssim 1.
\end{equation}
The difference with $(i)$ is that, this time, $\mu(\omega)$ remains bounded on $\mathcal{C}_{p, |\bk|}$, which is the reason why we have
\begin{equation} \label{estimu2} 
	\| \omega \mu(\omega ) \bbA_e(\omega) -|\bk| \, {\bf e_3} \times \bbA_m(\omega)\| \lesssim |\bk|^{2},
\end{equation}
which leads to the estimate for $\mathcal{S}_{|\bk|}(\omega)$ in \eqref{eq.estimSps}(ii).  \\ [12pt] 
For $\mathcal{V}_{|\bk|}(\omega)$, one sees on formula \eqref{defV} that, this time, the dominant terms in $\mathcal{V}_{|\bk|}(\omega)$ on $\mathcal{C}_{p, |\bk|}$ is now proportional to  
$$
|q_{e,j_0}(\omega)^{-1}| \lesssim |\omega - p|^{-1}\lesssim  | \rho_\bk|^{-1}\lesssim  |\bk|^{2}.
$$
This explains  the estimate for $\mathcal{V}_{|\bk|}(\omega)$ in \eqref{eq.estimSps}(ii).
~\\ [12pt]
\noindent {\bf Conclusion.} The important observation is that, from \eqref{eq.estimSps}, in all cases
$$
\|\mathcal{S}_{|\bk|}(\omega)\| \,  \|\mathcal{V}_{|\bk|}(\omega)  \| \lesssim \rho^{-1}_{|\bk|}.
$$
One concludes by applying \eqref{estiproj} to $\Pi(|\bk|) = \Pi_{p} (|\bk|)$.
\end{proof}  
\noindent We now give the asymptotic expansion of the  eigenvalues $\omega_{p}(|\bk|)$ for large $|\bk|$. As the study of their imaginary part requires to distinguish two cases, for the sake of readability, we delay it to the corollary \ref{CoroIm}.
\begin{Lem}  \label{LemEigens}
	Let $p\in \calP_s$. The eigenvalue $\omega_{p}(|\bk|)$ satisfies the following asymptotic expansion
	\begin{equation}\label{eq.asymptps}
	\omega_{p}(|\bk|)=p+A_{p,2}\,  |\bk|^{-2}+A_{p,4}\,  |\bk|^{-4}+o(|\bk|^{-4}),\ \mbox{ as } |\bk|\to +\infty,
	\end{equation} 
where the complex numbers $A_{p,2}$ and $A_{p,4}$	 are given by
\begin{equation}\label{eq.constCps}
\begin{array}{llll}
\mbox{ if } p=\pm \,  \omega_{e,j_p} \in \calP_e,& \, \displaystyle  A_{p,2}=- \mbox{$\frac{1}{2}$} \, \varepsilon_0   \,p \, \mu(p)  \,  \Omega_{e,j_p}^2 \neq 0&  \mbox{ and } & \displaystyle A_{p,4}=\mbox{$\frac{1}{2}$}\Big(\big(\omega^2   \mu \,h_{e,p}\big)^2\Big)'(p).\\ [12pt]
  \mbox{ if } p=\pm \, \omega_{m,\ell_p} \in \calP_m, & \,  \displaystyle  A_{p,2}=- \mbox{$\frac{1}{2}$} \, \mu_0 \, p\, \varepsilon(p)  \, \Omega_{m,\ell_p}^2 \neq 0 & \mbox{ and } & \displaystyle A_{p,4}=\mbox{$\frac{1}{2}$}\Big(\big(\omega^2   \varepsilon \,h_{m,p}\big)^2\Big)'(p).
 \end{array}
\end{equation}
where for any $\omega \in (\bbC\setminus {\cal P}) \cup  \{p \}$,
\begin{equation}\label{eq.functionh}
		\mbox{ if } p=\pm \,  \omega_{e,j_p} , \quad h_{e,p}(\omega) = (\omega - p) \;  \varepsilon(\omega),  \quad \mbox{ if } p=\pm \,  \omega_{e,j_p} , \quad h_{m,p}(\omega) = (\omega - p) \;  \mu(\omega).
\end{equation}
\end{Lem}
\begin{proof}
Let $p\in \calP_s$.  Thus, either $p \in \calP_e$ or $p \in \calP_m$.
We assume without a loss of generality that $p \in \calP_e$. The proof is  done  by symmetric arguments if  $p \in \calP_m$. 
\\ [12pt]
\noindent  
We are in the situation  covered by the step 1 of the proof  of Proposition \ref{prop.dispersioncurves},  in the case $\mathfrak{m}_p = 1$, see also  \eqref{eq.polevoisinage} .  As $p$ is a simple pole of $\omega \, \varepsilon(\omega)$, there exists (by $\mathrm{H}_1$) a unique index $1 \leq j_p \leq N_e $ such that $p=\pm \, \omega_{e,j_p}$.
In that case the set  $\{ \omega_{p,n}(|\bk|), 1 \leq n \leq \mathfrak{m}_p\}$ is reduced to one single function denoted  $|\bk|\mapsto \omega_p(|\bk|)$. It is constructed  as the unique branch of solutions of the equation  $\mathcal{D}(\omega)^{-1}=\zeta $ with $\zeta=|\bk|^{-2}$ that converges to $p$ when $\zeta \to 0$. \\ [12pt] 
Next, we make more precise the factorization of $\mathcal{D}(\omega)^{-1} \equiv (\omega^ 2 \, \varepsilon \, \mu)^{-1}$ corresponding to \eqref{eq.polevoisinage} for $\mathfrak{m}_p = 1$.  By definition \eqref{eq.functionh} of $h_{e,p}$, we can write
$$ \mathcal{D}(\omega)^{-1}=(\omega -p) \, g_p(\omega) \quad \mbox{with } g_p(\omega) : =(\omega^2   \mu \,h_{e,p})(\omega)^{-1}.$$
From the definition of $\varepsilon(\omega)$, we also have 
\begin{equation}\label{prophep}  h_{e,p}(\omega)=  \ds \varepsilon_0 \, (\omega -p) \,  \Big(1- \sum_{ j\neq j_p} \frac{\Omega_{e,j}^2}{q_{e,j}(\omega)} \Big)- \varepsilon_0 \,  \frac{\Omega_{e,j_p}^2}{\omega+p}, \quad \mbox{\Big(thus } h_{e,p}(p)= - \varepsilon_0 \,  \frac{\Omega_{e,j_p}^2}{2p}\Big) \end{equation} 
which shows that $g_p$ is analytic in the neighbourhood of $p$ (by  $(\mathrm{H}_2)$, $\mu(p)\neq 0$) and that 
$$g_p(p)=-  2 \, \big(\varepsilon_0   \,p \, \mu(p) \, \Omega_{e,j_p}^2  \big)^{-1}\neq 0.$$
By formula \eqref{eq.asympexpansion}, with  $\zeta=|\bk|^{-2}$,  $z=p$,  $g = g_p$, $\mathfrak{m}=1$ , $ a_{n}^{-1}=g(p)^{-1}$, we get 
\begin{equation}
\omega_{p}(|\bk|)= p-   \varepsilon_0   \,p\,  \mu(p) \,  \frac{ \Omega_{e,j_p}^2 }{2} |\bk|^{-2} -\frac{  g_p'(p) |\bk|^{-4}}{ \, g_p(p)^3}  +o(  |\bk|^{-4} ),  \mbox{ as } |\bk| \to +\infty.
\end{equation}
This is nothing but \eqref{eq.asymptps} since $ \ds - g'_p/ \, g_p^2 = \big(1/g_p\big)' = (\omega^2   \mu \,h_{e,p})'$ by definition of $g_p$. Thus, one gets 
$$-\frac{  g_p'(p)}{ \, g_p(p)^3} = (\omega^2   \mu \,h_{e,p})'(p) (\omega^2   \mu \,h_{e,p})(p)=\frac{1}{2}\Big(\big(\omega^2   \mu \,h_{e,p}\big)^2\Big)'(p).$$ 
Finally, note that, since ${\cal P}_e \cap {\cal Z}_m = \varnothing$,  (see ($\mathrm{H}_2$)), $p \notin {\cal Z}_m$  so that $\mu(p) \neq 0$  thus $A_{p,2} \neq 0$.
\end{proof}
\noindent The behaviour of the imaginary part of $A_{2,p} $ leads us to make the distinction between the critical and non critical configurations. For the ease of the reader, we recall below that the critical case corresponds to one of the following two situations (see definition \ref{Critical_cases}):
\begin{enumerate}
	\item   $\forall \,  \ell \in \{ 1, \ldots, N_m \} , \; \alpha_{m,\ell}=0$ and  $\exists \,j \in \{ 1, \ldots, N_e\} \mid \alpha_{e,j}=0\,   \mbox{  and } \omega_{e,j} \notin \big\{ \omega_{m,\ell} \big\}.$
	\item $\forall \,  j \in \{ 1, \ldots, N_e \} , \; \alpha_{e,j}=0$ and  $\exists \, \ell \in \{ 1, \ldots, N_m\} \mid \alpha_{m,\ell}=0\,   \mbox{  and } \omega_{m,\ell}  \notin \big\{ \omega_{e,j} \big\}.$
\end{enumerate}
\begin{Cor} \label{CoroIm} Following the notations introduced in Lemma \ref{LemEigens}.  In the non critical case,  
		\begin{equation} \label{NonCrit}
		\forall \; p \in  {\cal P}_s,  \quad \operatorname{Im}A_{2,p} < 0.
	\end{equation}	As a consequence, there exists $C>0$ such that, for $|\bk|$ large enough, 
		\begin{equation} \label{estimImNC}
		\forall \; p \in {\cal P}_s, \quad	\operatorname{Im}\omega_{p}(|\bk|)<- \, C \,  |\bk|^{-2}.
	\end{equation}	
In the critical case, in situation 1, ${\cal P}_e \cap {\cal P}_s \neq \varnothing$ and 
	\begin{equation} \label{Crit1}
 \forall \; p \in {\cal P}_e \cap {\cal P}_s, \quad \operatorname{Im}A_{2,p} = 0 \mbox{ and } \operatorname{Im}A_{4,p} < 0, \quad 	\forall \;  p \in {\cal P}_m \cap {\cal P}_s, \quad \operatorname{Im}A_{2,p} < 0,
	\end{equation}	
while, in situation 2, ${\cal P}_m \cap {\cal P}_s \neq \varnothing$ and 
\begin{equation} \label{Crit2}
	\forall \; p \in {\cal P}_m \cap {\cal P}_s, \quad \operatorname{Im}A_{2,p} = 0 \mbox{ and } \operatorname{Im}A_{4,p} < 0, \quad 	\forall \;  p \in {\cal P}_e \cap {\cal P}_s, \quad \operatorname{Im}A_{2,p} < 0.
\end{equation}	
As a consequence, there exists $C>0$ such that, for $|\bk|$ large enough 
	\begin{equation} \label{estimImC}
\forall \; p \in {\cal P}_s, \quad	\operatorname{Im}\omega_{p}(|\bk|)<- \, C \,  |\bk|^{-4}.
\end{equation}	
	\end{Cor} 
\begin{proof}{\bf Non critical case.} We prove \eqref{NonCrit} for $p \in {\cal P}_s \cap {\cal P}_e$, the case  $p \in {\cal P}_s \cap {\cal P}_m$ is similar.\\ [12pt]  According to formula \eqref{eq.constCps}, we simply have to check that $\operatorname{Im} (p \, \mu(p)) > 0$. Thanks to  \eqref{eq.positvity} , this will be true as soon as,  at least one index $\ell$, $\alpha_{m,\ell}> 0$. However, $p \in {\cal P}_s \cap {\cal P}_e$ means that  $p = \pm \, \omega_{e,j}$ for some $j$ with $\alpha_{e,j}= 0$ and also that $\omega_{e,j} \notin \big\{ \omega_{m,\ell} \big\}$ (since $p$ is a simple pole). Thus, if all $\alpha_{m,\ell}$ vanished, we would be in the situation 1 of the critical case, which is excluded. 
	\\ [12pt]
	Of course, \eqref{estimImNC} is a direct consequence of \eqref{NonCrit} . \\ [12pt]
{\bf Critical case.}  
We prove \eqref{Crit1} in situation 1, the proof of \eqref{Crit2} in situation 2 being similar. \\ [12pt] 
First  ${\cal P}_e \cap {\cal P}_s$ is non empty since any $\omega_{e,j}$ such that $\alpha_{e,j} = 0 \mbox{  and } \omega_{e,j} \notin \big\{ \omega_{m,\ell} \big\}$ (a set which is itself non empty by definition of the situation 1) belongs to ${\cal P}_e \cap {\cal P}_s$. \\ [12pt] Let $p \in {\cal P}_e \cap {\cal P}_s$. In situation 1,  all $\alpha_{m,\ell}$ vanish thus, by \eqref{eq.positvity}, $\mu(\omega)$ and $\mu'(\omega)$ are real-valued  for $\omega \in \R\setminus  \calP_m$. In particular, $ \operatorname{Im}  \, \big(p \, \mu(p)\big) = 0$ which  implies  by \eqref{eq.constCps} that $\operatorname{Im}A_{2,p} = 0$. 
Next, as $h_{e,p}(p)\in \R$ by \eqref{prophep}, using \eqref{eq.constCps} again yields
\begin{equation}\label{ImAp4}
A_{p,4}= (\omega^2   \mu \,h_{e,p})'(p) (\omega^2   \mu \,h_{e,p})(p) \quad \Rightarrow \quad\operatorname{Im}  \, A_{p,4} =    p^4  \, \mu^2(p) \, h_{e,p}(p) \  \operatorname{Im}  h'_{e,p}(p) . 
\end{equation}
Moreover, $p \in {\cal P}_e \cap {\cal P}_s$ implies $p = \pm \, \omega_{e,k}$ for some $k$ such that  $\alpha_{e,k} = 0$. However, from  the formulas for $\mu(\omega)$ \eqref{eq.permmitivity-permeabiity},  the polynomials  $q_{e,j}$ (\ref{eq.polynom}) and  $h_{e,p}$ (\ref{prophep}), one computes that
\begin{equation}\label{eq.Imhprimee}
	\ds h_{e,p}(p)= - \,  \varepsilon_0 \,  \frac{\Omega_{e,k}^2}{2p}  \quad \mbox{and} \quad 	\operatorname{Im} h_{e,p}'(p)= \varepsilon_0 \, p \sum_{j\neq k} \frac{\alpha_{e,j} \, \Omega_{e,j}^2}{ |q_{e,j}(p)|^2} . 
\end{equation}
The weak dissipation condition \eqref{WD} implies that at least one $\alpha_{e,j}$ for $j \neq k$  is  positive so that $ h_{e,p}(p )\,\operatorname{Im} h_{e,p}'(p) < 0$, implying by \eqref{ImAp4} that  $\operatorname{Im} A_{p,4} < 0$.\\ [12pt]
Finally, for $p \in {\cal P}_s \cap {\cal P}_m$, the proof of $\operatorname{Im}  \, A_{p,2} < 0$ uses  \eqref{eq.constCps}, \eqref{eq.positvity} and the fact that at least one $\alpha_{e,j}$ for $j \neq k$  is positive. The proof of  \eqref{Crit1} is thus complete. \\ [12pt]
Finally \eqref{estimImC} is easily deduced from \eqref{eq.asymptps},  \eqref{Crit1}  and  \eqref{Crit2} after  remarking that, for  $|\bk|$ large enough, $|\bk|^{-4} \lesssim |\bk|^{-2}$ which implies $- |\bk|^{-2} \lesssim - |\bk|^{-4}$ (details are again left to the reader).
\end{proof} 
Thus, proceeding as in the proof of Lemma \ref{LemEstiinfty}, one shows using Lemma \ref{LemPros}  and Corollary  \ref{CoroIm} the following result.
\begin{Lem}\label{LemEstis}
There exists $k_+>0$ and $ C>0$  such that the function  $\bbU_s(\bk,t)$   defined by   \eqref{eq.decompositionfourterm} and \eqref{eq.constantinf}  satisfies the following estimates: 
\begin{enumerate}
\item If the system  \eqref{planteamiento Lorentz} is in a non-critical configuration, then
	\begin{equation} \label{polynomial_decayncr-US}
	|\bbU_s(\bk,t)| \lesssim  \rme^{-\frac{C\, t}{|\bk|^2}} |\bbU_0(\bk)|,  \quad \forall t \geq 0 \  \mbox{  and }  \ \forall \; |\bk|\geq k_+ .
	\end{equation}
	\item If the  system   \eqref{planteamiento Lorentz} is in a critical configuration, then  
	\begin{equation} \label{polynomial_decaycr-US}
	|\bbU_s(\bk,t)| \lesssim  \rme^{-\frac{C\, t}{|\bk|^4}} \ |\bbU_0(\bk)|,  \quad \forall t \geq 0  \ \mbox{  and }  \ \forall \; |\bk|\geq k_+.
	\end{equation}
\end{enumerate}
\end{Lem}
\subsubsection{Estimates of $\bbU_{d}(\bk,t)$ for $|\bk| \gg 1$}  \label{estiUd}
\noindent This time we estimate $\bbU_{d}(\bk,t)$ in  \eqref{eq.decompositionfourterm}  which involves in particular the projectors $\Pi_{p,r} (|\bk|)$.
\begin{Lem} \label{LemProd}
	The projectors  $\Pi_{p,r} (|\bk|)$, $p \in {\cal P}_d, r = 1,2 $ are uniformly bounded for large $|\bk|$.  
\end{Lem}
\begin{proof}
We follow again the approach of Section \ref{orientation} for $\omega|\bk|) = \omega_{p,r}(|\bk|)$, $p\in \mathcal{P}_d, r =1,2$ and denote  $\mathcal{C}_{p,r,|\bk|}$ the corresponding contour (see \eqref{defcont} and figure \ref{fig-contour-proj}). Without any loss of generality, we can restrict ourselves  to $r=1$.  \\ [12pt]
{\bf Step 1: estimate of $\rho_{|\bk|}$}. Thanks to the asymptotic \eqref{eq.pole}, it is clear that the distance from $\omega_{p,1}(|\bk|)$ to any other eigenvalue that is different from $\omega_{p,2}(|\bk|)$ remains bounded from below, for large $|\bk|$, by a positive constant. The same observation holds true for distance from $\omega_{p,1}(|\bk|)$ to any other point of the set ${\cal S}_{\cal T}$ than $p$. \\ [12pt]
Oppositely, the distances $|\omega_{p,1}(|\bk|)-p|$ and $|\omega_{p,1}(|\bk|)-\omega_{p,2}(|\bk)|$  tend to 0 when $|\bk| \rightarrow + \infty$. More precisely, from \eqref{eq.pole} applied with $\mathfrak{m}_p = 2$, since $p \in {\cal P}_d$, one deduces that 
	\begin{equation} \label{distances}
|\, \omega_{p,1}(|\bk|)-p \, | \sim  |A_p|^{\frac{1}{2}} \,  |\bk|^{-1}, \quad | \, \omega_{p,1}(|\bk|)-\omega_{p,2}(|\bk)\, | \sim  2 \, |A_p|^\frac{1}{2} \,  |\bk|^{-1}, \quad (|\bk| \to +\infty )
\end{equation}  
with $|A_p|> 0$. As a consequence, for $|\bk|$ large enough, we have 
 		\begin{equation} \label{estirhokd}
	\rho_{|\bk|} =   \mbox{$\frac{1}{2}$} \; | \, \omega_{p}(|\bk|) - p \, | \sim  \mbox{$\frac{1}{2}$} \;  |A_p|^{\frac{1}{2}} \; |\bk|^{-1} \quad (|\bk| \rightarrow + \infty), \quad \mbox{with} \quad  |A_p|> 0.
\end{equation}  
\noindent 
\hspace*{-0.1cm} {\bf Step 2: Estimate of \big($\mathcal{D}(\omega)-|\bk|^2\big)^{-1}$.} It is very similar to that of in Lemma \ref{LemPros}. One observes thanks to \eqref{eq.pole}  that, when $|\bk| \rightarrow +\infty$,  in the product \eqref{eq.Dfactor} defining $D_{b,|\bk|}(\omega)$,  only two terms, namely $\omega - \omega_{p,1}(|\bk|)$  and $\omega - \omega_{p,2}(|\bk|)$, are not bounded from below when $\omega$ describes $\mathcal{C}_{p,1,|\bk|}$ and $|\bk| \rightarrow + \infty$. Thus for some $C>0$, $\omega \in \mathcal{C}_{p,1,|\bk|}$ and $|\bk|$ large enough 
\begin{equation} \label{boundinf} 
D_{b,|\bk|}(\omega)| \geq C \; \big| \,  \omega - \omega_{p,1}(|\bk|) \, \big| \  \big| \,  \omega - \omega_{p,2}(|\bk|) \, \big| = C \; \rho_{|\bk|}  \,  \big| \,  \omega - \omega_{p,2}(|\bk|) \, \big| .
\end{equation}
By the reverse triangular inequality, along $\mathcal{C}_{p,1,|\bk|}$,
$$ \big| \,  \omega - \omega_{p,2}(|\bk|) \, \big| \geq \big| \,  \omega_{p,1}(|\bk|)  - \omega_{p,2}(|\bk|) \, \big| - \big| \,  \omega - \omega_{p,1}(|\bk|) \, \big| =  \big| \,  \omega_{p,1}(|\bk|)  - \omega_{p,2}(|\bk|) \, \big| - \rho_{|\bk|},
$$ 
thus, using the equivalent \eqref{estirhokd} for $\rho_{|\bk|}$ and $|\omega_{p,1}(|\bk|)-\omega_{p,2}(|\bk)| \sim  2 \, |A_p| \,  |\bk|^{-1}$, 
$$ 
\big| \,  \omega - \omega_{p,2}(|\bk|) \, \big| \geq  \big| \,  \omega_{p,1}(|\bk|)  - \omega_{p,2}(|\bk|) \, \big| - \rho_{|\bk|} \sim  \mbox{$\frac{3}{2}$} \;  |A_p| \; |\bk|^{-1}.
$$
Therefore, from \eqref{boundinf}, we deduce that
$$|D_{b,|\bk|}(\omega)|^{-1}\lesssim \rho_{|\bk|}^{-1} \, |\bk|,  \quad  \forall \; \omega\in \mathcal{C}_{p,1, |\bk|}. $$
Thus, proceeding as in Lemma \ref{LemPros} (observe that by passing from $|D_{b,|\bk|}(\omega)|^{-1}$ to $|D_{|\bk|}(\omega)|^{-1}$, see \eqref{boundDb} and \eqref{eq.Dboundinvp}, one looses   two powers of $|\bk|$),
\begin{equation}\label{eq.Dboundinvpd}
	|D_{|\bk|}(\omega)|^{-1}\lesssim    \rho_{|\bk|}^{-1}   \,  |\bk|^{-1},  \quad  \forall \; \omega\in \mathcal{C}_{p,1, |\bk|} \, . 
\end{equation}
Moreover,  $p\in \mathcal{P}_d$ is a double zero of the product $Q_e\, Q_m$. Hence, for any $\omega\in \mathcal{C}_{p,1, |\bk|}$ 
\begin{equation}\label{eq.boundQpsd}
	|Q_e(\omega)\, Q_m(\omega)| \lesssim  |\omega-p|^2 \leq \big( \rho_{ |\bk|}+|\, \omega_p(|\bk|)-p|\big)^2 \,| \sim \mbox{$\frac{9}{4}$} \;  |A_p| \; |\bk|^{-2} \quad (\mbox{by (\ref{distances}, \ref{estirhokd})}).
\end{equation}
As a consequence, from \eqref{eq.expressioncalDinv}, we finally get by product of \eqref{eq.Dboundinvpd} and \eqref{eq.boundQpsd}, 
\begin{equation}\label{boundDcalpm2d}
	|(\mathcal{D}(\omega)-|\bk|^2)^{-1}| \lesssim   \rho_{|\bk|}^{-1}  \; |\bk|^{-3},  \quad \forall \; \omega\in \mathcal{C}_{p, 1, |\bk|} . 
\end{equation}
\noindent {\bf Step 3: Estimates of $\mathcal{S}_{|\bk|}(\omega)$ and $\mathcal{V}_{|\bk|}(\omega)$}
~\\ [12pt]
As $p \in {\cal P}_m \cap {\cal P}_e$,  there is a unique pair  of indices $\{\ell_0, j_o\}$ such that $p =\pm \omega_{e,j_0} =\pm \omega_{m,\ell_0}$. The situation is a kind of mix between the two situations $(i)$ and $(ii)$ met in the step 3 of the proof of Lemma \ref{LemPros}. \\ [12pt]  Using  \eqref{boundoperators} one sees that all the operators $(\bbA_{e,j}(\omega), \dot \bbA_{e,j}(\omega),  \bbA_{m,\ell}(\omega), \dot  \bbA_{m,\ell}(\omega))$ are uniformly bounded along $\mathcal{C}_{p,1, |\bk|}$ except $(\bbA_{e,j_0}(\omega), \dot  \bbA_{e,j_0}(\omega), \bbA_{m,\ell_0}(\omega), \dot  \bbA_{m,\ell_0}(\omega) \big)$. \\ [12pt]
Proceeding as for \eqref{estiAeAm} in the proof of Lemma \ref{LemPros}, one easily gets (we omit the details)
\begin{equation}\label{estiAeAmd}
	\|\bbA_{e}(\omega)\| \lesssim |\bk|, \quad \|  \bbA_{m}(\omega)\| \lesssim |\bk|.
\end{equation}
As $p$ is a simple pole of $\mu(\omega)$, we obtain, similarly to   \eqref{estimmu}  in the proof of Lemma \ref{LemPros},
\begin{equation} \label{estimmud2} 
\forall \; \omega \in  \mathcal{C}_{p,1, |\bk|}, \quad 
	|\mu(\omega)| \lesssim |\omega-p|^{-1}\lesssim |\bk|,
\end{equation}
which, joined to \eqref{estiAeAmd}, gives 
\begin{equation} \label{estimAeAmbisd} 
	\| \omega \mu(\omega ) \bbA_e(\omega) -|\bk| \, {\bf e_3} \times \bbA_m(\omega)\| \leq |\bk|^{2}.
\end{equation}
Finally using  \eqref{boundDcalpm2d} ans  \eqref{estimAeAmbisd}  in the definition \eqref{defS} of $\mathcal{S}_{|\bk|}(\omega)$, we get 
\begin{equation}\label{eq.estimSpmd}
	\|\mathcal{S}_{|\bk|}(\omega)  \|\lesssim    \rho_{|\bk|}^{-1} \; |\bk|^{-1},  \quad  \forall \; \omega\in \mathcal{C}_{p,1, |\bk|} .
\end{equation}
For $\mathcal{V}_{|\bk|}(\omega)$, we first observe that $|\mu(\omega)| \sim \mu_p  \, |\omega-p|^{-1}$ ($\omega \rightarrow  p$) also implies that  along $\mathcal{C}_{p,1,|\bk|}$,
\begin{equation*} \label{boundmu2bis}
	|\mu(\omega)|^{-1} \lesssim \mu_p^{-1} \, \rho_{|\bk|}^{-1} \lesssim   |\bk|^{-1}  \ \mbox{ as } \ |\bk| \to +\infty.
\end{equation*}
Thus, the function ${|\bk|}\, {\mu(\omega)^{-1}}$is bounded for $|\bk|$ large enough. As a consequence, one sees on  formula \eqref{defV} that $\mathcal{V}_{|\bk|}(\omega)$  blows  with the same rate than $q_{e,j_0}(\omega)^{-1}$ and $q_{m,\ell_0}(\omega)^{-1}$, that is to say proportionally to $\rho_{|\bk|}^{-1}$, in other words  proportionally to $|\bk|$.  Therefore
\begin{equation}\label{eq.estimVpmd}
	\|\mathcal{V}_{|\bk|}(\omega)  \|\lesssim    |\bk|,  \quad  \forall \; \omega\in \mathcal{C}_{p,1, |\bk|} .
\end{equation}
\noindent {\bf Conclusion.} Again, by \eqref{eq.estimSpmd} and \eqref{eq.estimVpmd}, 
$
\|\mathcal{S}_{|\bk|}(\omega)\| \,  \|\mathcal{V}_{|\bk|}(\omega)  \| \leq \rho^{-1}_{|\bk|}
$
and one concludes with \eqref{estiproj} for $\Pi(|\bk|) = \Pi_{p} (|\bk|)$.
\end{proof}

\noindent We now give the asymptotic expansion of the  eigenvalues $\omega_{p,r}(|\bk|)$ for large $|\bk|$.
\begin{Lem} \label{LemEigend}
For $p\in \calP_d$,  then  $p=\pm \, \omega_{e,j_p}=\pm \, \omega_{m,\ell_p}$ for some $(j_p,\ell_p)$ and  
\begin{equation}\label{eq.asymptd}
\omega_{p,r}(|\bk|)=p+ (-1)^r \,  \frac{\Omega_{e,j_p} \, \Omega_{m,\ell_p} }{2\, c} \, |\bk|^{-1} + A_{p,2}  \, |\bk|^{-2}+o(|\bk|^{-2}), \quad \mbox{ as } |\bk| \to +\infty,
\end{equation}
where  the complex number $A_{p,2}$ is given by 
\begin{equation}\label{eq.coeffomegad}
A_{p,2}= \frac{1}{2} \,(\omega^2\, h_{e,j_p} h_{m,\ell_p})'(p),
\end{equation}
where the functions $h_{e,p}$ and $h_{m,p}$ are defined in \eqref{eq.functionh}. Moreover, one has 
\begin{equation}\label{eq.coeffomegadbis}
	\operatorname{Im} \, A_{p,2}< 0.
\end{equation}
\end{Lem}
\begin{proof}
Let $p\in \calP_d$. Then (by $\mathrm{H}_1$), there exist two unique indices  $j_p\in \{ 1,\ldots , N_e\}$ and $\ell_p\in \{ 1,\ldots , N_m\}$   such that $p=\pm \omega_{e,j_p}=\pm \omega_{m,\ell_p}$. Then the two functions $|\bk|\mapsto \omega_{p,r}(|\bk|)$  are   defined   in the proof of Proposition \ref{prop.dispersioncurves} for $|\bk|$ large enough  via the Lemma  \ref{Lem-implicte-function} as the two  branches of solutions of the equation  $\mathcal{D}(\omega)^{-1}=\zeta^2 $ with $\zeta=|\bk|^{-1}$ in a vicinity of $p$. Here,  $p\in \calP_e\cap \calP_m$ is a common simple pole of the rational functions   $\omega \varepsilon$ and  $\omega \mu$. More precisely, from the definition \eqref{eq.constCps} of $( h_{e,p}, h_{m,p})$ (see Lemma \ref{LemEigens}) we have $$\mathcal{D}(\omega)= \omega^2 h_{e,p}(\omega)\, h_{m,p}(\omega) \,(\omega-p)^{-2}.$$
The latter expression emphasizes  that $p$ is a double pole of the rational function $\mathcal{D}(\omega)$, thus a double zero of  $\mathcal{D}^{-1}$ :  $\mathcal{D}(\omega)^{-1}=(\omega -p)^2 g(\omega)$  where $g(\omega):=( \omega^2 h_{e,p}(\omega)\, h_{m,p}(\omega) )^{-1}$ is  analytic in the vicinity of $p$ and satisfies (using the expression \eqref{prophep} for $h_{e,p}$ and its equivalent for $h_{m,p}$):  $$g(p)=4 \, c^2\Omega_{e,j_p}^{-2} \, \Omega_{m,\ell_p}^{-2 }>0.$$
Thus, using the asymptotic  formula \eqref{eq.asympexpansion} of Lemma \ref{Lem-implicte-function}  (with  $\zeta=|\bk|^{-1}$,  $z=p$,  $\mathfrak{m}=2$ and $ a_{1}=-\sqrt{g(p)}=-2c /(\Omega_{e,j_p}\,  \Omega_{m,\ell_p})$, $a_{2}=\sqrt{g(p)}$ the two roots of $X^2=g(p)$)    yields
\begin{equation}\label{eq.asympomegad}
\omega_{p,r}(|\bk|)= p+  a_{r}^{-1} |\bk|^{-1}  -\frac{a_r^{-2} g'(p)}{2\,  g(p)} \, |\bk|^{-2} +o(  |\bk|^{-2} ),  \mbox{ as } |\bk| \to +\infty.
\end{equation}
To conclude, it remains to remark  that $a_r^{-2}=g(p)^{-1}$ and $g^{-1}= \omega^2 h_{e,p}\, h_{m,p}$ which gives
\begin{equation}\label{eq.derivepoledouble}
-\frac{a_r^{-2} g'(p)}{  g(p)} =-\frac{ g'(p)}{ g^2(p)}=(g^{-1})'(p)=( \omega^2 h_{e,p}\, h_{m,p})'(p).
\end{equation}
Finally, it remains to show that  $\operatorname{ Im}A_{p,2}<0$.
From  \eqref{eq.Imhprimee} and the expressions \eqref{eq.functionh} of $h_{e,p}$ and $h_{m,p}$, it follows (as $p^2$,  $h_{e,p}(p)$, $h_{m,p}(p )$ are real)  that 
$$
\operatorname{ Im}A_{p,2}= \frac{p^2}{2} \, \big[ h_{m,p}(p) \operatorname{Im} h'_{e,p}(p) +h_{e,p}(p) \operatorname{Im}h'_{m,p}(p)\big].
$$
Finally using  that $ h_{e,p}(p)= - \varepsilon_0\, \Omega_{e,j_p}^2/(2\, p)$ and $ h_{m,p}(p)= - \mu_0\, \Omega_{m,\ell_p}^2/(2\, p)$, the expression \eqref{eq.Imhprimee} for  $ \operatorname{Im}(h'_{e,p}(p))$ and its equivalent form for $\operatorname{Im}(h'_{m,p}(p))$ gives that
$$
\operatorname{ Im}A_{p,2}=-c^{-2}\, p^2\,  \bigg( \frac{ \Omega_{m,\ell_p}^2}{4}  \sum_{j=1, j\neq j_p}^{N_e} \frac{\Omega_{e,j}^2 \alpha_{e,j} }{ |q_{e,j}(p)|^2}  + \frac{ \Omega_{e,j_p}^2}{4}   \sum_{\ell=1, \ell\neq \ell_p}^{N_m} \frac{\Omega_{m,\ell}^2 \alpha_{m,\ell} }{ |q_{m,\ell}(p)|^2} \bigg)
$$
(where $c^{-2}=\varepsilon_0 \, \mu_0$). Thus, this term is negative by the weak dissipation condition \eqref{WD} since at least one coefficient $ \alpha_{e,j}$ or $ \alpha_{m,\ell}$ is positive. Finally, combining \eqref{eq.asympomegad} and \eqref{eq.derivepoledouble} gives  \eqref{eq.asymptd}. 
\end{proof}
\noindent We now estimate the term $\bbU_{d}(\bk,t)$ for large $|\bk|$. Proceeding as in the proof of Lemma \ref{LemEstiinfty}, one shows using Lemmas \ref{LemProd}  and   \ref{LemEigend} the following result.
\begin{Lem}\label{LemEstid}
If $\calP_d\neq\varnothing$, then there exists $k_+>0$ and $C>0$  such that   the function  $\bbU_d(\bk,t)$   defined by   \eqref{eq.decompositionfourterm} and \eqref{eq.constantinf}  satisfies the following estimate
	\begin{equation}\label{eq.estimateUd}
		| \bbU_{d}(\bk,t)| \lesssim    \rme^{-\frac{C\, t}{|\bk|^2}} |\bbU_0(\bk)|,   \quad \forall t \geq 0 \ \mbox{ and }  \ \forall  |\bk|\geq k_+.
	\end{equation}
\end{Lem}

\subsubsection{Estimates of $\bbU_-(\bk,t)$ for $|\bk| \gg 1$} \label{estiUminus}
As announced in Section \ref{orientation}, since we simply want to obtain a ``rough" exponential decay estimate for  $\bbU_-(\bk,t)$, we do not need to separate the analysis in three lemmas as in Sections \ref{estiUinfty}, \ref{estiUs} and \ref{estiUd} but give a direct proof using Riesz-Dunford functional calculus.
\begin{Lem}\label{Estiminus}
There exists  $\delta>0$ and $k_+>0$ such that $\bbU_{-}(\bk,t)$, defined by  \eqref{eq.decompositionfourterm} and \eqref{decompU},  satisfies
	\begin{equation}\label{eq.estimateUm}
		| \bbU_-(\bk,t)| \lesssim \rme^{-\delta\, t} \ | \bbU_0(\bk)|, \quad  \forall \, t\geq 0, \quad  \forall \;  |\bk| \geq k_+.
	\end{equation}
\end{Lem}
\begin{proof}
	\noindent 
	We introduce a (positively oriented) simple closed contour $\Gamma$, included in $\bbC^-$ such that  all the poles of $\calP_-$  lies inside $\Gamma$  (see figure \ref{fig-disp-curv-pm}). We denote $\delta$ the distance from $\Gamma$ to the real axis:
	$$\delta=\min \{-\operatorname{Im}(\omega), \, \omega \in \Gamma \}>0. $$
For $|\bk|$ large enough, by  \eqref{eq.pole} and \eqref{eq.pminfty},   $\Gamma$ encloses all eigenvalues $\omega_{p,n}(|\bk|)$ for $p\in \calP_-$ and $n\in \{ 1, \ldots, \mathfrak{m}_p\}$ but no other elements of the spectrum of $\bbA_{|\bk|}$. Then, by the Riesz-Dunford functional calculus, there exists $k_+>0$ such that we have the formula: 
	\begin{figure}[h!]
		\begin{center}
			\includegraphics[scale=0.31]{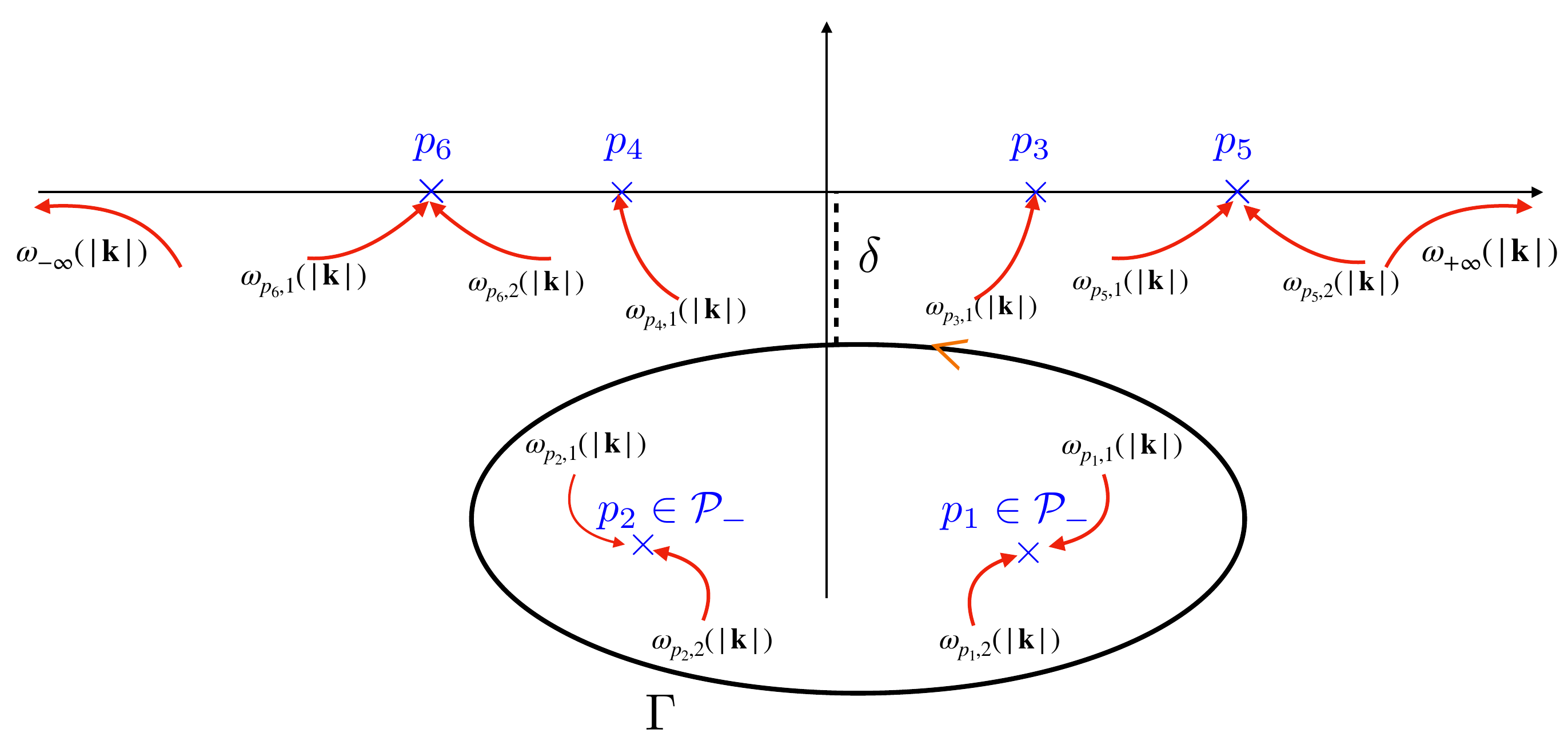}
		\end{center}
		\caption{Contour integration for the estimate of  $\bU_-(\bk,t)$  in the case where $\mathcal{P}_-=\{ p_1, p_2\}$ (corresponding to the figure \ref{fig-disp-curv}).  }
		\label{fig-disp-curv-pm}	
	\end{figure}\\
	\begin{equation}\label{eq.RieszDunford-1}
		\bbU_-(\bk,t)=- \frac{\, \mathcal{R}_{\bk}^*} {2\rmi \pi}\int_{\Gamma}  e^{-\rmi \omega t} \,  R_{|\bk|}(\omega) \, \, \mathcal{R}_{\bk} \bbU_0(\bk) \, \rmd \omega , \ \mbox{ for } |\bk|\geq k_+.
	\end{equation}
	As $\mathcal{R}_{\bk}$ is unitary, it leads to
	\begin{equation}\label{eq.Um1}
		| \bbU_-(\bk,t)| \lesssim \rme^{-\delta\, t} \, \Big(\max_{\omega \in \Gamma}\|  R_{|\bk|}(\omega) \| \Big) \ | \bbU_0(\bk)|.
	\end{equation}
 It remains now to estimate  $R_{|\bk|}(\omega)$ from its expression 
	$R_{|\bk|}(\omega)= \mathcal{V}_{|\bk|}(\omega) \mathcal{S}_{|\bk|}(\omega)+\mathcal{T}(\omega)$ in Proposition \ref{Prop.res} (choosing $\Gamma$ such that $\Gamma\cap \mathcal{Z}_m=\varnothing$). The more involved step concerns the esimate of $\mathcal{S}_{|\bk|}(\omega)$.
	\\[6pt]
	\noindent {\bf Step 1: Estimate of $\mathcal{S}_{|\bk|}(\omega)$.} 
	The expression of $\mathcal{S}_{|\bk|}(\omega)$, given in \eqref{defS}, involved the term $(\mathcal{D}(\omega)-|\bk|^2)^{-1}$. We first bound this term  by using \eqref{eq.expressioncalDinv}, \eqref{eq.Dfactor} and \eqref{eq.Dbound}.      
From the asymptotic behaviour \eqref{eq.pole}, it follows that $D_{\mathrm{b},|\bk|}(\omega)$ given by \eqref{eq.Dbound} saitisfies:
	\begin{equation}\label{eq.Dbd}
		|D_{\mathrm{b},|\bk|}(\omega)^{-1}|\lesssim 1, \quad \forall \;  \omega \in \Gamma.
	\end{equation}
	From the asymptotic expansion \eqref{eq.pminfty}, $|\omega-\omega_{\pm \infty}(|\bk|)|^{-1}=c^{-1}|\bk|^{-1}+o(|\bk|^{-1})$.  Hence, it yields with  \eqref{eq.Dfactor} and  \eqref{eq.Dbd}:
	\begin{equation}\label{eq.estimD}
		|D_{|\bk|}(\omega)^{-1}|\lesssim |\bk|^{-2}, \quad \mbox{ and thus with \eqref{eq.expressioncalDinv}}\quad |(\mathcal{D}(\omega)-|\bk|^2)^{-1}|\lesssim |\bk|^{-2} \quad \forall \, \omega \in \Gamma.
	\end{equation}
	On the other hand, from the definitions \eqref{operatorsApAm} and \eqref{operatorsAeAh}, it is immediate that 
	\begin{equation}\label{eq.estimSoppart}
		\big\| - \omega \mu(\omega ) \bbA_e(\omega) +|\bk| \, {\bf e_3} \times \bbA_h(\omega) \big\| \lesssim |\bk| , \quad  \forall\,  \omega \in \Gamma.
	\end{equation}
	Combining \eqref{eq.estimD} and  \eqref{eq.estimSoppart} with \eqref{defS} yields
	\begin{equation}\label{eq.opSnorm}
		\| \mathcal{S}_{|\bk|}(\omega)\| \lesssim |\bk|^{-1},  \ \forall \omega \in \Gamma  \mbox{ and } |\bk| \geq k_+.
	\end{equation}
	\noindent {\bf Step 2: Final estimate.}
	From the expression \eqref{defV} and \eqref{defT} of $\mathcal{V}_{|\bk|}(\omega)$ and $\mathcal{T}(\omega)$ we deduce that
	\begin{equation}\label{eq.opVnorm}
		\| \mathcal{V}_{|\bk|}(\omega)\| \lesssim |\bk| , \quad  \  \|\mathcal{T}(\omega)\|\lesssim 1, \quad    \forall \; \omega \in \Gamma ,
	\end{equation}
		thus by \eqref{eq.expressresolv} and \eqref{eq.opSnorm}, $
	\|\mathcal{R}_{|\bk|}(\omega)\|\lesssim 1, \,   \forall \, \omega \in \Gamma
	$ which we substitute  into \eqref{eq.Um1} to get \eqref{eq.estimateUm}.
\end{proof}
\subsubsection{The global estimates} \label{Estiglobale} 
\begin{Thm}\label{HF-estm}
There exists $k_+>0$  such that for $|\bk|\geq k_+$, the spatial Fourier components $|\bbU(\bk,t)|$ of the solution   of \eqref{eq.schro} with initial condition $\bU_0\in \mathcal{H}_{\perp}$ satisfy the following estimates:
\begin{enumerate}
\item If the Maxwell system is in a non-critical configuration, then  there  $\exists \; C,\,\widetilde{C}>0$  such that
	\begin{equation} \label{polynomial_decayncr-HF}
	|\bbU(\bk,t)| \leq   \widetilde{C}  \, \rme^{-\frac{C\, t}{|\bk|^2}}  \, |\bbU_0(\bk)|,  \quad \forall  \, t \geq 0    .
	\end{equation}
	\item If the Maxwell system is in a critical configuration, then  there $\exists  \; C,\,\widetilde{C}>0$ such that
	\begin{equation} \label{polynomial_decaycr-HF}
	|\bbU(\bk,t)| \leq \widetilde{C} \,  \rme^{-\frac{C\, t}{|\bk|^4}} \ |\bbU_0(\bk)|,  \quad \forall  \, t \geq 0      .
	\end{equation}
\end{enumerate} 
\end{Thm}
\begin{proof} The inequalities  \eqref{polynomial_decayncr-HF} and \eqref{polynomial_decaycr-HF} follow immediately  from the expression \eqref{eq.decompositionfourterm} of  $|\bbU(\bk,t)|$ for  $|\bk|$ large enough and the estimates of the four terms  $|\bbU_{\infty}(\bk,t)|$, $|\bbU_s(\bk,t)|$, $|\bbU_d(\bk,t)| $ and $|\bbU_-(\bk,t)|$ given respectively by  Lemmas \ref{LemEstiinfty}, \ref{LemEstis}, \ref{LemEstid} and \ref{Estiminus} (since for $|\bk|$ large enough, $t\geq 0$ and $C_1,C_2>0$: $e^{- \delta\, t}\leq   e^{-C_1\, |\bk|^{-2} t}\leq    e^{- C_2\, |\bk|^{-4} t}$). 
\end{proof}

\noindent We prove in the following result that the estimates of Theorem \ref{HF-estm} are optimal for an infinite family of well chosen  initial conditions $\bU_0\in \mathcal{H}^{m}_{\perp,\operatorname{HF}}$ for any fixed  $m>0$.
\begin{Thm}\label{HF-estm-opt}
Let $m, \, k_+>0$ and  $\phi: \bbR^+ \mapsto \bbR$  be any a   measurable and bounded function satisfying:
\begin{equation}\label{eq.phik}
\operatorname{supp} \phi \subset  [k_+,+\infty[ \mbox{ and  }  0< |\phi(|\bk|)|\lesssim \, (1+|\bk|^2)^{-s/2}  \mbox{ for } |\bk| >k_+ +1  \mbox{ and } s>\frac{3}{2}+m . 
\end{equation}
\begin{enumerate}
\item If the Maxwell system is in a non-critical configuration,  one defines (for $k_+$ sufficiently large) an initial  condition $\bU_0$  of \eqref{eq.schro} via its Fourier transform:
\begin{equation}\label{eq.defUOkopt}
\bbU_0(\bk)= \mathcal{F}(\bU_0)(\bk)=  \phi(|\bk|) \; \mathcal{R}_{\bk}^* \, \frac{\mathcal{V}_{|\bk|}\big(\omega_{+\infty}(|\bk|)\big)  \mathbf{e}_{1}}{|\mathcal{V}_{|\bk|}\big(\omega_{+\infty}(|\bk|)\big)  \bf{e}_1|}, \quad \forall \, \bk\in \R^{3,*}.
\end{equation}
Then, $\bU_0\in  \mathcal{H}^{m}_{\perp,\operatorname{HF}}$ and     $\exists \; C,\,\widetilde{C}>0$  such that  the associated  solution $\bU$ of \eqref{eq.schro}  satisfies
	\begin{equation} \label{polynomial_decayncr-HF2}
	  \widetilde{C}  \, \rme^{-\frac{C\, t}{|\bk|^2}} \, |\bbU_0(\bk)|\leq |\bbU(\bk,t)|,  \quad \forall \, t \geq 0   \   \mbox{ and }\ \forall \, \bk\in \R^{3,*}.
	\end{equation}
	\item If the Maxwell system is in a critical configuration, by Corollary  \ref{CoroIm}, there exists (at least one) $p\in \calP_s$ such that
	\begin{equation}\label{eq.branchcirtique}
	\operatorname{Im}\omega_p(|\bk|)=\operatorname{Im}A_{p,4} \,  |\bk|^{-4}+ o(|\bk|^{-4}), \mbox{ as } \ |\bk|\mapsto+ \infty, \mbox{ with } \operatorname{Im}A_{p,4} < 0.
	\end{equation}
	Defining (for $k_+$ sufficiently large) as an initial  condition $\bU_0$  for \eqref{eq.schro} by
\begin{equation}\label{eq.defUOkoptcr}
\bbU_0(\bk)= \mathcal{F}(\bU_0)(\bk)= \phi(|\bk|) \; \mathcal{R}_{\bk}^* \, \frac{\mathcal{V}_{|\bk|}\big(\omega_{p}(|\bk|)\big)  \mathbf{e}_{1}}{|\mathcal{V}_{|\bk|}\big(\omega_{p}(|\bk|)\big)  \bf{e}_1|}, \quad \forall \, \bk\in \R^{3,*},
\end{equation}	
then, $\bU_0\in  \mathcal{H}^{m}_{\perp,\operatorname{HF}}$ and   $\exists \; C,\,\widetilde{C}>0$  such that the  associated     solution $\bU$ of \eqref{eq.schro} satisfy
	\begin{equation} \label{polynomial_decaycr-HF2}
	 \widetilde{C} \,  \rme^{-\frac{C\, t}{|\bk|^4}} \ |\bbU_0(\bk)| \leq |\bbU(\bk,t)|,  \quad \forall \, t \geq 0      \mbox{ and }  \bk\in \R^{3,*}.
	\end{equation}
\end{enumerate}
In other words, the  estimates \eqref{polynomial_decayncr-HF} and \eqref{polynomial_decaycr-HF} are optimal for an infinite family of solutions.
\end{Thm}
\begin{proof} We separate it in two steps. \\ [10pt]
\noindent {\bf  Step 1.  Proof of the lower  bound \eqref{polynomial_decayncr-HF2}:}\\[4pt]
From Proposition \ref{prop.dispersioncurves}, we know for $|\bk|>k_+$,  $|\bk|\mapsto \omega_{{+\infty}}(|\bk|)$ is well defined and satisfy (by Lemma \ref{LemEigeninfty})  the asymptotic expansion \eqref{eq.estimVpm}. 
\\ [12pt]
Furthermore,  the proof of  Proposition \ref{Prop.spec} in appendix \ref{sec-app-spec2},  shows that any  $\bbU$ in the two-dimensional eigenspace $\operatorname{ker}\big(\bbA_{|\bk|}-\omega_{+\infty}(|\bk|) \, \mathrm{Id}\big)$  is of the form $\bbU=\big(\bbE,\bbH, \bbP, \dot \bbP, \bbM, \dot \bbM\big)$ for some $\bbE\in \bC_{\perp}$ and $\big(\bbH, \bbP, \dot \bbP, \bbM, \dot \bbM\big)$  deduced from $\bbE$ by formula $(\ref{EV1}, \ref{EV2})$, with $\omega=\omega_{+\infty}(|\bk|)$. \\ [12pt]
This can be expressed with the help of  the  operator $\mathcal{V}_{|\bk|}(\omega) \in {\cal L}(\bC_{\perp}, \bC_{\perp}^N)$, see \eqref{defV},  as follows 
\begin{equation}\label{eq.noyauAmodk}
 \operatorname{ker}(\bbA_{|\bk|}-\omega_{+\infty}(|\bk|)\, \mathrm{Id})=\mathcal{V}_{|\bk|}\big(\omega_{+\infty}(|\bk|)\big)\big(\bC_{\perp}).
\end{equation}
Let $m \in \N^*$. Let us  define $\bU_0$  via its Fourier transform $\bbU_0(\bk)$ and formula \eqref{eq.defUOkopt} where $\mathcal{R}_{\bk}$, see \eqref{eq.defRkvect}, is  unitary.
For $s >3/2 + m$, $\bk \mapsto (1+|\bk|^2)^{m/2} \, \bbU_0(\bk)\in \bL^2(\bbR^3)^N$, thus $\bU_0 \in \bH^m(\bbR^3)^N$ and therefore belongs by construction to $\mathcal{H}^{m}_{\perp,\operatorname{HF}}$. \\ [12pt]
Next, thanks to formula \eqref{eq.refsolutionfourier2} for $\bbU(\bk,t)$, \eqref{eq.defUOkopt} and \eqref{eq.noyauAmodk}, one has $$\bbU(\bk,t)=\rme^{-\rmi \bbA_{\bk} t} \, \bbU_0(\bk) = \rme^{- \rmi  \omega_{{+\infty}}(|\bk|) \, t} \, \bbU_0(\bk)  \quad \mbox{ and thus }\quad
|\bbU(\bk,t)| =\rme^{-\operatorname{Im}\omega_{{+\infty}}(|\bk|)\, t} \, | \bbU_0(\bk) |.
$$
Moreover, using Lemma \ref{LemEigeninfty} and more precisely the expansion \eqref{eq.asymptinf} for $\omega_{+\infty}(|\bk|)$, one knows that for large enough $|\bk|$, with $A_{2,\infty} > 0$ defined by  \eqref{eq.constantinf}, 
$$-\frac{A_{2,\infty}}{ c^2\, |\bk|^2} \leq \operatorname{Im}\omega_{+\infty}(|\bk|)<0 \quad (\mbox{simply because } 1/3 < 1/2).
$$
We thus get
$   \ds \rme^{-\frac{A_{2,\infty}\, t}{ c^2  |\bk|^2}} |\bbU_0(\bk)| \lesssim  |\bU(\bk,t)|,
$
which  achieves the proof of \eqref{polynomial_decayncr-HF2}.  \\[10pt]
\noindent {\bf  Step 2    Proof of the lower  bound \ \eqref{polynomial_decaycr-HF2}:}\\[4pt]
If the Maxwell system is in the critical configuration, to show that  the estimate \eqref{polynomial_decaycr-HF}, one only needs to define $\bU_0$ by \eqref{eq.defUOkoptcr} (instead of  \eqref{eq.defUOkopt})
replacing $\omega_{\infty}(|\bk|)$ by $\omega_{p}(|\bk|)$ satisfying  \eqref{eq.branchcirtique} (instead of  \eqref{eq.asymptinfIm}).
\end{proof}

\section{Asymptotic analysis for small spatial frequencies $|\bk|\ll 1$}\label{sec.LF}
As the reader can expect, the structure of this section is similar to  that of Section  \ref{sec.HF} with three subsections:
\begin{itemize} 
	\item Section \ref{sec.LFdispcurv}: Asymptotics of dispersion curves for $0<|\bk| \ll 1$.
	\item  Section \ref{spec-decomp-zero}: Spectral decomposition of the solution for $0<|\bk| \ll 1$.
	\item  Section \ref{sec.estmLF}: Large time estimate of $\bbU(\bk,t)$ for $0<|\bk| \ll 1$.
\end{itemize}
This section is however shorter than Section  \ref{sec.HF} because much less particular cases appear.

\subsection{Asymptotics of the  dispersion curves  for $|\bk|\ll 1$}\label{sec.LFdispcurv}
This section is the counterpart of Section \ref{sec.HFdispcurv} for low spatial frequencies.
In other words,  we focus  here on long time estimates of the low (spatial) frequency components  $\bbU(\bk,t)$ (see \eqref{eq.refsolutionfourier}) of the solution. As explained in  Section \ref{sec-main-lines-analysis}, the decay of $\bbU(\bk,t)$ for $|\bk|\ll 1$  is related to the analysis of the solutions of the dispersion relation \eqref{eq.disp} at low frequencies. Roughly speaking, as $|\bk|^2\to 0$ when $|\bk|\to 0$, the solutions of   \eqref{eq.disp} must satisfy $|\mathcal{D}(\omega)|\to 0$ as $|\bk|\to0$. Thus,
 they converge to  a zero $z\in \mathcal{Z}  \cup\{ 0\}$ of the  rational function $\mathcal{D}$.
Thus, as one can rewrite $\mathcal{D}$ in the vicinity of  a zero $z$ of multiplicity $\mathfrak{m}_z$ as 
\begin{equation}\label{eq.zerovoisinage}
	\mathcal{D}(\omega)=(\omega-z)^{\mathfrak{m}_z} g(\omega) \mbox{ with $g$ analytic in a vicinity of $z$ and }  g(z)=A_{z}\neq 0.
\end{equation}
This leads to the following proposition (the equivalent of Proposition \ref{prop.dispersioncurves} in Section \ref{sec.HFdispcurv}).
\begin{Pro}\label{prop.dispersioncurvesLF}
There exists $k_->0$  such that for $0<|\bk|\leq k_-$, the solutions of the rational dispersion relation  \eqref{eq.disp} (or of its equivalent polynomial form \eqref{eq.disppolynom})  are all simple. These solutions form $N$ distinct  branches which  are $C^{\infty}$-smooth functions of  $|\bk| \in (0, k_-]$. These branches are characterized by their asymptotic  expansion for small $|\bk|$. More precisely, if  $z\in \mathcal{Z}\cup \{0\}$ is a zero of multiplicity $\mathfrak{m}_z$ of $\mathcal{D}$, there exists $\mathfrak{m}_z$ distinct branches  of solutions $ \omega_{z,n},\, n=1,\ldots, \mathfrak{m}_z$  of  \eqref{eq.disp}  satisfying
\begin{equation}\label{eq.zero}
	\omega_{z,n}(|\bk|)=z+  a^{-1}_{z,n} \;  |\bk|^{\frac{2}{\mathfrak{m}_z}} \, (1+o(1)), \quad a_{z,n}=|A_z|^{1/\mathfrak{m}_z} \; \rme^{\rmi \frac{\theta_z}{\mathfrak{m}_z}}  \, \rme^{\frac{2 \, \rmi n\pi}{\mathfrak{m}_z}} \quad (|\bk| \to 0 ),
\end{equation}
where $A_z$ is defined in  \eqref{eq.zerovoisinage} and $\theta_z\in (-\pi, \pi]$ is the  principal argument of $A_z$. 
\end{Pro}
\begin{proof}
\noindent {\bf Step 1: Construction of the branches of solutions $\omega_{z,n}(|\bk|)$.} \\[4pt]
Let $z\in\mathcal{Z}\cup \{0\}$ a zero of multiplicity $\mathfrak{m}_z$. Then, the rational function $\mathcal{D}$ can be factorized as in \eqref{eq.zerovoisinage}. Moreover, the dispersion relation can be rewritten as 
$$
 \mathcal{D}(\omega)=|\bk|^{2}=\big(|\bk|^{\frac{2}{\mathfrak{m}_z}}\big)^{\mathfrak{m}_z}.
$$
Thus, by applying the Lemma \ref{Lem-implicte-function} of the appendix \ref{sec-asymptotic} with
$$
\mathcal{G}(\omega) =  \mathcal{D}(\omega), \quad \mathfrak{m}=\mathfrak{m}_z, \quad  A=A_z \mbox{ and  } \zeta=|\bk|^{2/\mathfrak{m}_z},$$
we deduce the existence of $k_->0$ and $\mathfrak{m}_z$  distinct $C^\infty$  branches of solutions  $|\bk| \in (0, k_-]\mapsto \omega_{z,n}(|\bk|)$ of the dispersion $\mathcal{D}(\omega)=|\bk|^2$. 
Furthermore, for the solutions  $\omega_{z,n}(|\bk|)$, we can use the  asymptotic expansion  \eqref{eq.asympexpansion} of Lemma \ref{Lem-implicte-function} with   $\zeta=|\bk|^{2/\mathfrak{m}_z}$ which gives  \eqref{eq.zero}. \\[12pt]
\noindent {\bf Step 2: Conclusion.}
From the asymptotic  \eqref{eq.zero}, $\omega_{z,n}(|\bk|)$ for  $z\in {\cal Z}\cup \{0 \}$ and $n\in \{1,\ldots, \mathfrak{m}_z\}$  are all distinct  for $|\bk|$ small enough and positive. Since the sum of the $\mathfrak{m}_z$'s, for $z \in {\cal Z}\cup \{ 0\}$ is equal to $2N_e+2N_m+2$, we have constructed $2N_e+2N_m +2$  distinct solutions of \eqref{eq.disp}. As \eqref{eq.disp} is equivalent to  a polynomial equation  of degree $2N_e+2N_m +2$, cf. \eqref{eq.disppolynom}, these solutions are simple roots  and there are no other solutions for small positive $|\bk|$.
\end{proof}

\subsection{Spectral decomposition of the solution for $|\bk| \ll 1$}\label{spec-decomp-zero}
Our goal here is to obtain for small $|\bk|$ a decomposition of $\bbU(\bk,t)$ similar to the one obtained for large $|\bk|$ (see (\ref{eq.decompositionfourterm},\ref{decompU})).\\ [12pt]
For this, we first notice that, combining Proposition \ref{prop.dispersioncurvesLF} and Corollary \ref{eq.crit-diag} immediately yields the following property of  the operator $ \bbA_{|\bk|,\perp}$.
\begin{Cor}\label{eq.crit-diag3}
There exists $k_->0$  such that for $0<|\bk|\leq k_-$, $\bbA_{|\bk|,\perp}$ is diagonalizable on $\bC^N_{\perp}$.
\end{Cor}
\noindent To obtain  the decomposition of $\bbU(\bk,t)$, we split (see Section \ref{sec_dispersion-zeros}) the sets of zeros $\{ 0\}\cup\mathcal{Z}\subset \overline{\bbC^-}$  of $\mathcal{D}$  in three disjoint subsets : $  \{ 0\}\cup \mathcal{Z}=\{0 \}\cup \mathcal{Z}_s\cup \mathcal{Z}_-$ where  the multiplicity of $z=0$ is $\mathfrak{m}_0=2$ and
\begin{equation} \label{Zs}
\mathcal{Z}_-:=\mathcal{Z} \cap \bbC^- \mbox{ and }  \mathcal{Z}_s:=\{ z \in \mathcal{Z} \cap \bbR \mid \mathfrak{m}_z=1\}.
\end{equation}
\begin{Rem} \label{remZmoins} [On the set $\mathcal{Z}_-$]
We point out that the weak dissipation condition  \eqref{WD} implies (see \eqref{eq.consequenceWD} and \eqref{eq.positvity}) that at least one element  of $\mathcal{Z}$ lies in $\bbC^-$, thus $\mathcal{Z}_-\neq \varnothing$. \\ [6pt] The strong dissipation condition \eqref{SD} implies that  
all  elements of $\mathcal{Z}$ lie in $\bbC^-$, that is $\mathcal{Z}_-=\mathcal{Z}$.
\end{Rem}
\begin{Rem} \label{remZs} [On the set $\mathcal{Z}_s$]
	According to the analysis of the zeros of  made in Section \ref{sec_dispersion-zeros}, the structure of the set $\mathcal{Z}_s$ strongly relies on the condition \begin{equation} \label{condition} 
\exists \; (j_0,\ell_0) \mbox{ such that }\alpha_{e,j_0}>0 \mbox{ and } \alpha_{m,\ell_0}>0.
\end{equation} 
More precisely,  if \eqref{condition} holds, $\mathcal{Z}_s=\varnothing$. On the contrary, under the weak dissipation condition  \eqref{WD},  if \eqref{condition} does not hold, either all $\alpha_{e,j}$ vanish in which case  $\mathcal{Z}_s=\mathcal{Z}_e$ either all $\alpha_{\ell,m}$ vanish and $\mathcal{Z}_s=\mathcal{Z}_m$.
\end{Rem}
\noindent According to Propositions \ref{Prop.spec}  and \ref{prop.dispersioncurvesLF}, we have the corresponding partition of the spectrum of $\bbA_{|\bk|,\perp}$ for $0<|\bk|\leq k_-$:
\begin{equation} \label{eq.specdispLF} 
\sigma(\bbA_{|\bk|, \perp})=\big\{ \omega_{0,r}(|\bk|) ,\, r=1,2 \big\} \cup \big\{ \omega_{z}(|\bk|), z\in \mathcal{Z}_s \big\} \cup   \big\{\omega_{z,n}(|\bk|),  \ z \in \mathcal{Z}_-, \,  n =1,\ldots, \mathfrak{m}_z  \big \},
\end{equation} 
where, for simplicity,  we write $\omega_{z}(|\bk|)$ instead of $\omega_{z,1}(|\bk|)$ for the simple zeros  $z\in \mathcal{Z}_s$. We point out that ``by symmetry'' of the dispersion relation $\omega_{0,2}(|\bk|)=-\overline{\omega}_{0,1}(|\bk|)$. \\ [12pt] 
For a sketch of the small $|\bk|$ behaviour of the dispersion curves, see figure \ref{fig-disp-curv-LF}.\\ [12pt] 
\begin{figure}[h!]
		\begin{center}
			\includegraphics[scale=0.31]{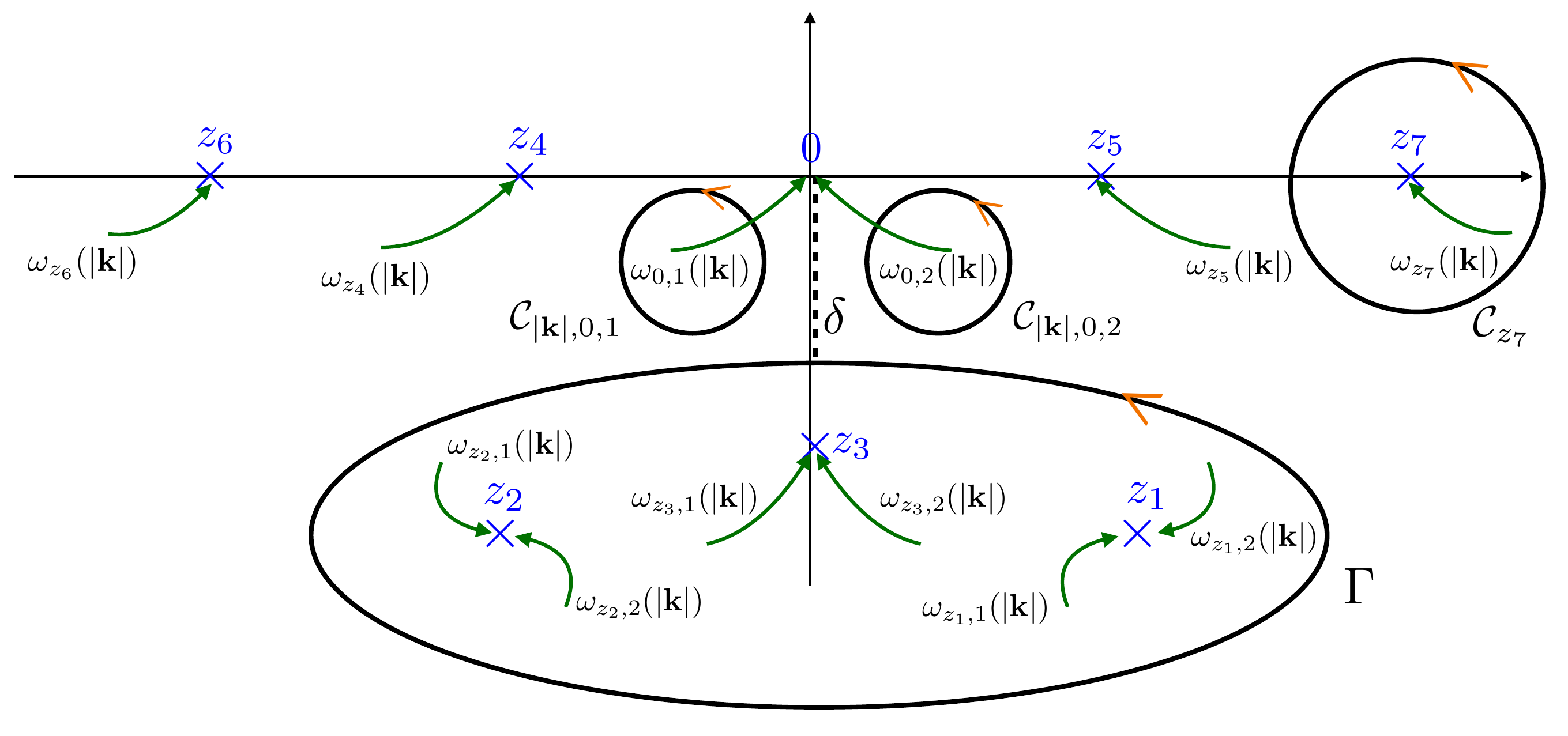}
		\end{center}
		\caption{Sketch of   a configuration of  dispersion curves   for small  values of $|\bk|$ in the case where $\{0\} \cup \mathcal{Z}=\{0\} \cup \mathcal{Z}_-\cup \mathcal{Z}_s$ with $\mathcal{Z}_-=\{ z_1, \,z_2,\, z_3\}$ and  $\mathcal{Z}_s=\{ z_4, z_5, z_6, z_7\}$.  }
\label{fig-disp-curv-LF}	
\end{figure}
\noindent For $0<|\bk|\leq k_-$,  as $\bbA_{|\bk|, \perp}$ is diagonalizable by Corollary \ref{eq.crit-diag3},  $\bC_{\perp}^N$ can be decomposed as
\begin{equation}\label{deq.decompCNperpbis}
 \bC_{\perp}^N= \bigoplus_{\substack {r=1}}^{2}V_{|\bk|,0,r} \bigoplus_{z\in \mathcal{Z}_{s}}V_{|\bk|,z}  \oplus \bigoplus_{z\in \mathcal{Z}_{-}} \mathop{\oplus}_{\substack {n=1}}^{\mathfrak{m}_z}V_{|\bk|,z,n}
\end{equation}
where we have defined the 2D (see below) subspaces 
\begin{equation}
 \left\{ \begin{array}{llll}
      V_{|\bk|,0,r}= \operatorname{ker}\big( \bbA_{|\bk|, \perp} - \omega_{0,r}(|\bk|)\, \mathrm{I}d \big),  \;  r=1, 2  \quad  V_{|\bk|,z}= \operatorname{ker}\big( \bbA_{|\bk|, \perp} - \omega_{z}(|\bk|)\, \mathrm{I}d\big), \\[10pt]
 V_{|\bk|,z,n}=\operatorname{ker}\big( \bbA_{|\bk|, \perp} - \omega_{z,n}(|\bk|) \, \mathrm{I}d \big), \;1\leq n \leq \mathfrak{m}_z,
\end{array} \right.
\end{equation}
where  the above direct sums are (in general) non-orthogonal. Thus, one decomposes uniquely any vector $x\in  \bC_{\perp}^N$ as
 $$
\left| \;  \begin{array} {ll}
 x= \ds \sum_{r=1 }^2 x_{ |\bk|,0,r}+\sum_{z\in \mathcal{Z}_{s}} x_{ |\bk|, z}+\sum_{z\in \mathcal{Z}_-} \sum_{n=1}^{\mathfrak{m}_z}x_{|\bk|, z, n}
\\ [18pt]
 x_{ |\bk|, 0, r} \in V_{|\bk|,0,r}, \ x_{ |\bk|, z}\in V_{|\bk|,z}\  \mbox{ and } x_{ |\bk|, z, n}\in V_{|\bk|,z,n}.
\end{array} \right. $$
Then, we  define  the spectral projectors $\Pi_{0, r}(|\bk|)$ for $r\in \{ 1,2\}$,   $\Pi_z(|\bk|)$ for $z\in \mathcal{Z}_s$ and $\Pi_{p,n}(|\bk|)$, $z\in\mathcal{Z}_-$ and $n\in \{ 1, \ldots, \mathfrak{m}_z \}$  associated to the eigenvalues  $\omega_{0,r}(|\bk|)$, $\omega_z(|\bk|)$ and $\omega_{z,n}(|\bk|)$ by:
\begin{equation}\label{eq.projdefbis}
\Pi_{0, r}(|\bk|)(x)= x_{ |\bk|, 0, r}, \
 \Pi_{z}(|\bk|)(x)=x_{|\bk|, z} \ \mbox{ and } \ \Pi_{z,n}(|\bk|)(x)=x_{|\bk|, z,n}. 
 \end{equation}
 From Proposition \ref{Prop.spec},  the geometric  multiplicity of each  eigenvalues in $\sigma(\bbA_{|\bk|,\perp})$ is two. Thus,  all these projectors  are rank two projectors.  We again emphasize  that the  dissipative operator $\bbA_{|\bk|,\perp}$ is not normal, thus its spectral projectors are not  orthogonal.
\\ [12pt] 
\noindent For $0<|\bk|\leq k_-$,  as $\bbA_{|\bk|, \perp}$ is diagonalizable (by Corollary \ref{eq.crit-diag}), one has
\begin{equation}
\bbA_{|\bk|,\perp}=\sum_{r=1}^2 \omega_{0,r}(|\bk|) \, \Pi_{0,r}(|\bk|) + \sum_{z\in \mathcal{Z}_s}   \omega_{z}(|\bk|) \, \Pi_{z}(|\bk|)+ \sum_{z\in \mathcal{Z}_-}\sum_{n=1}^{\mathfrak{m}_z} \omega_{z,n}(|\bk|)\, \Pi_{z,n}(|\bk|).
\end{equation}
Thus,  for  $0<|\bk|\leq k_-$, the solution $\bbU(\bk,t)$  given by \eqref{eq.refsolutionfourier2} can be expressed for all $ t \geq 0$    as

\begin{equation}\label{eq.decompositionthreeterm}
\bbU(\bk,t)= \bbU_{z,0}(\bk,t)+\bbU_{z,s}(\bk,t) + \bbU_{z,-}(\bk,t)
\end{equation}
where
\begin{equation} \label{decompUbis}
\left\{ \begin{array}{llll}  
\bbU_{z,0}(\bk,t)&=& \ds \sum_{r=1 }^2 \, \rme^{-\rmi \, \omega_{0,r}(|\bk|) \,t} \, \mathcal{R}_{\bk}^* \, \Pi_{0,r}(|\bk|) \, \mathcal{R}_{\bk} \,\bbU_0(\bk), & \quad(i)\\[19pt]
\bbU_{z,s}(\bk,t) &=& \ds \sum_{z\in \mathcal{Z}_s}  \rme^{-\rmi \, \omega_{z}(|\bk|) t} \, \mathcal{R}_{\bk}^* \, \Pi_{z}(|\bk|) \mathcal{R}_{\bk}\,\bbU_0(\bk), &\quad(ii)\\[19pt]
\bbU_{z,-}(\bk,t)&= & \ds \sum_{z\in \mathcal{Z}_-} \sum_{n=1}^{\mathfrak{m}_z} \, \rme^{-\rmi \, \omega_{z,n}(|\bk|) \,t} \, \mathcal{R}_{\bk}^* \, \Pi_{z,n}(|\bk|)\, \mathcal{R}_{\bk} \bbU_0(\bk) . &\quad (iii)
\end{array} \right. 
\end{equation}
\subsection{Estimates of the low frequency components of the solution}\label{sec.estmLF}
\subsubsection{Orientation} \label{orientation-LF}
Proceeding as for large $|\bk|$ in Section \ref{sec.estmHF}, we  estimate successively, in Lemmas \ref{LemEstizero}, \ref{LemEstizeros}  and \ref{EstimZminus}, the  three terms appearing in the decomposition \eqref{eq.decompositionthreeterm}.  An important difference with Section \ref{sec.estmHF} is that we consider bounded values of $|\bk|$, namely $0<|\bk| < k_-$, for some $k_- > 0$, which  will simplify some estimates via compactness arguments.\\ [12pt]
Two cases will have to be distinguished. \\ [12pt]
{\bf Case 1}:  estimates of  $\bbU_{z,0}(\bk,t)$ and $\bbU_{z,s}(\bk,t)$. This will be established in Sections \ref{estiU0} and  \ref{estiUZS},
by bounding separately each of the terms of the sums  in \eqref{decompUbis}(i) and  \eqref{decompUbis}(ii). \\ [12pt] 
 For each  eigenvalue  $\omega(|\bk|) \in \big\{  \omega_{0,r}(|\bk|), \omega_{z}(|\bk|)\}$, we shall first 
 estimate in Lemmas \ref{LemProzerodouble} and \ref{LemProszero}  the corresponding spectral projector $\Pi (|\bk|)$ uniformly in $|\bk|$ for $|\bk|$ small enough and positive. To this aim as in  Section  \ref{sec.estmHF}, we use the Dunford-Riesz functional
 calculus that says that, if $\Gamma_{|\bf k|}$ is a  (positively oriented) close contour in $\mathbb{C}$  that  encloses $\omega(|\bk|)$ but no other eigenvalue of $\bbA_{|\bk|,\perp}$,  so that 
  \begin{equation}\label{eq.RieszDunford2}
 	\Pi(|\bk|)=- \frac{1}{2\rmi \pi}\int_{\Gamma_{|\bf k|}} R_{|\bk|}(\omega) \, \rmd \omega \, .  
 \end{equation}
\noindent  For $ \Pi_{z}(|\bk|) $  in Lemma \ref{LemProszero}, $\omega_z (|\bk|) $ is a simple eigenvalues that converges to $z$ while remaining well separated from the other eigenvalues. Hence, using Proposition \ref{Prop.res}, we can choose  for $\Gamma_{|\bf k|}$ a fixed circle $\Gamma$ (i.e. independent of $|\bk|$), along which the operator $\mathcal{R}_{|\bk|}(\omega)$, by continuity in
$(|\bk|, \omega)$ and compactness of $[0,k_-] \times \Gamma$, remain bounded. Thus, the uniform estimate of $ \Pi_{z}(|\bk|) $ becomes obvious.  \\ [12pt]
For the  projector $\Pi_{0,1}(|\bk|)$ (this is the same for $\Pi_{0,2}(|\bk|)$ ), this is more complicated because the two distinct eigenvalues $ \omega_{0,r}(|\bk|)$ ($r=1,2$) converge both to $0$ while remaining far way from  the other eigenvalues  and any other point of the set ${\cal S}_{\cal T}$ than $0$. For this reason, we shall take as $\Gamma_{|\bf k|}$ the circle ${\cal C}_{0,1,|\bk|}$ centered at $\omega_{0,1}(|\bk|)$ with radius $\rho_{|\bk|}=|\omega_{0,2}(|\bk|)-\omega_{0,1}(|\bk|)|/4$ in order to prevent $\omega_{0,2}(|\bk|)$ from being inside ${\cal C}_{0,1,|\bk|}$. This circle will also not contain $0$ or any other point of  ${\cal S}_{\cal T}$ (by Proposition \ref{prop.dispersioncurvesLF}). As a consequence, in \eqref{eq.RieszDunford2}, $R_{|\bk|}(\omega)$ can be replaced by the product $\mathcal{V}_{|\bk|}(\omega) \, \mathcal{S}_{|\bk|}(\omega)$,  (see formula \eqref{eq.expressresolv} for $R_{|\bk|}(\omega)$ and the definition \eqref{defST} of ${\cal S}_{\cal T}$) leading,  similarly to \eqref{estiproj}, to
\begin{equation}\label{estiproj2}
	\|\Pi(|\bk|)\| \leq \rho_{|\bk|} \, \sup_{\omega \in \mathcal{C}_{ |\bk| }}  \big(\|  \mathcal{V}_{|\bk|}(\omega) \|  \, \| \mathcal{S}_{|\bk|}(\omega)  \|\big) \ .
\end{equation} Moreover the  product  $\|  \mathcal{V}_{|\bk|}(\omega) \|  \, \| \mathcal{S}_{|\bk|}(\omega)  \|$  will blow up when $|\bk|$ tend to $0$ but this will be compensated by the fact that the radius of $ \rho_{|\bk|}$ tends to 0. \\ [12 pt] 
\noindent In a second step, we concentrate on the exponentials appearing in each factor of the terms  involved in $\bbU_{z,0}(\bk,t)$ and $\bbU_{z,s}(\bk,t)$   whose estimates for small  $|\bk|$ rely on the asymptotic expansion of the eigenvalues (and more particularly their imaginary parts), see Lemmas  \ref{LemEigenzo} and \ref{LemEigenzs}.\\[6pt] 
\noindent{\bf Case 2}:  estimate of  $\bbU_{z,-}(\bk,t)$. This term will  be treated as the term $\bbU_{-}(\bk,t)$ (in Section   \ref{estiUminus}), because the set $\big\{ \omega_{z,n}(|\bk|)\}$ will remain far from the real axis for small enough $|\bk|$. Thanks to this property, the exponential decay of $|\bbU_{z,-}(\bk,t)|$ will be, contrarily to the previous terms $\bbU_{z,0}(\bk,t)$ and $\bbU_{z,s}(\bk,t)$, uniform with respect to $|\bk|$, so that it will not contribute in the end to the large time equivalent of $\bU(\cdot,t)$. For proving this, we  shall not use  \eqref{decompUbis}(ii) but an alternative  formula directly issued from the Riesz-Dunford functional calculus.  Compared to Section  \ref{estiUminus}, its estimate in Lemma \ref{EstimZminus} will  be simplified by a compactness argument.

\subsubsection{Estimates of $\bbU_{z,0}(\bk,t)$ for $0<|\bk|\leq k_-$} \label{estiU0}
In the following Lemma, we estimate the spectral projectors $\Pi_{0, r} (|\bk|)$ (for $r=1,2$) using  \eqref{estiproj2} with a  simple closed contour $ \Gamma_{|\bf k|}$ satisfying the two following properties:
\begin{equation} \label{condGamma}   
	\left\{ \begin{array}{lll}
		(i) & \Gamma_{|\bf k|} \mbox{ encloses $\omega(|\bk|)$ but no other eigenvalue of $\bbA_{|\bk|,\perp}$,} \\ [10pt]
		(ii) & \mbox{$\Gamma_{|\bf k|}$  does not enclose any point of the set ${\cal S}_{\cal T}$.}
	\end{array} \right.
\end{equation}
\begin{Lem} \label{LemProzerodouble}
There exists $k_->0$ such that  the spectral projectors  $\Pi_{0,r} (|\bk|)$, $r \in \mathcal\{1,2\}$ are uniformly bounded for $0<|\bk|\leq k_-$.  
\end{Lem}

\begin{proof}
Let $r\in \{ 1,2 \}$. We follow the approach of  the Section  \ref{sec.estmHF} and introduce as $\Gamma_{|\bf k|}$  the circle  $\mathcal{C}_{0,r,|\bk|}$ (positively oriented) of  center  $\omega_{0,r}(|\bk|)$ and radius
\begin{equation} \label{choicerhokbis}
 	\rho_{|\bk|} =  \mbox{$\frac{1}{4}$} \,  | \, \omega_{0,1}(|\bk|)-\omega_{0,2}(|\bk|)\, |,
 	\end{equation}
(see figure \ref{fig-disp-curv-LF} for an example of contour $\mathcal{C}_{0,r,|\bk|}$ with $r=1,2$.)
\\ [12pt]
From the aymptotics \eqref{eq.zero} applied with $\mathfrak{m}_0 = 2$   ($z=0$ is  a double zero), one deduces that
\begin{equation} \label{distancesbis}
	|\, \omega_{0,1}(|\bk|) \, | \sim  |A_0|^{\frac{1}{2}} \,  |\bk|, \quad | \, \omega_{0,1}(|\bk|)-\omega_{0,2}(|\bk)\, | \sim  2 \, |A_0|^{\frac{1}{2}} \,  |\bk|, \quad (|\bk| \to 0 )
\end{equation}  
so that, for $|\bk|$ small enough, $\Gamma_{|\bf k|}$ does not enclose $\omega_{0,2}(|\bk|)$ nor 0 (and of course any other point of ${\cal S}_{\cal T})$. 
As explained in Section \ref{orientation-LF}, we can apply the inequality \eqref{estiproj2}. 
As  $ \rho_{|\bk|} \lesssim |\bk|$, according to \eqref{choicerhokbis} and \eqref{distancesbis}, we thus obtain 
\begin{equation}\label{estipojero}
	\|\Pi_{0,r}(|\bk|)\| \lesssim |\bk| \, \sup_{\omega \in \mathcal{C}_{ 0, r, |\bk| }}  \big(\|  \mathcal{V}_{|\bk|}(\omega) \|  \, \| \mathcal{S}_{|\bk|}(\omega)  \|\big).
\end{equation} 
It remains to estimate the operators $\mathcal{V}_{|\bk|}(\omega)$ and $\mathcal{S}_{|\bk|}(\omega)$ given by the formulas \eqref{defS} and \eqref{defV}.
 Without any loss of generality, we can restrict ourselves  to  $r=1$ for the rest of the proof.  \\[4pt]
\noindent {\bf Step 1: Estimate of $\big(\mathcal{D}(\omega)-|\bk|^2\big)^{-1}$.}\\[4pt] 
This term is involved  in the expression \eqref{defS} of  $ \mathcal{S}_{|\bk|}(\omega)$.  Let us recall that (see \eqref{eq.expressioncalDinv})
$$\big(\mathcal{D}(\omega)-|\bk|^2\big)^{-1} = Q_e(\omega)\, Q_m(\omega)D_{|\bk|}(\omega)^{-1}.$$
where $D_{|\bk|}(\omega)$ is the polynomial defined by \eqref{eq.disppolynom}.
 Using the proposition \ref{prop.dispersioncurvesLF}, we know that $D_{|\bk|}(\omega)$ can be factorized as the following for $|\bk|$ sufficiently small and positive:
\begin{equation}\label{eq.polynomefact}
D_{|\bk|}(\omega)= \varepsilon_0 \, \mu_0 \, \big(\omega- \omega_{0,1}(|\bk|)\big) \, \big(\omega- \omega_{0,2}(|\bk|)\big) \, \prod_{z\in \mathcal{Z}_s} \big(\omega-\omega_{z}(|\bk|)\big)\, \prod_{z\in \mathcal{Z}_-} \prod_{n=1}^{\mathfrak{m}_z}\big(\omega-\omega_{z,n}(|\bk|)\big).
\end{equation}
Since $\omega_{z}(|\bk|), z \in {\cal Z}_s$ or $\omega_{z,n}(|\bk|), z \in {\cal Z}_-$ converge to $z \neq 0$, cf. \eqref{eq.zero}, the last two products in \eqref{eq.polynomefact} remain bounded from below in modulus when $\omega$ describes $\mathcal{C}_{ p, r, |\bk| } $  and $|\bk|$ is small enough. \\ [12pt]
On the other hand, it is easy to see that for $|\bk|$ small enough ($ \mathcal{C}_{ 0, 1, |\bk|}$ has been chosen for that purpose), 
\begin{equation*} 
\forall \;  \omega \in  \mathcal{C}_{ 0, 1, |\bk| }, \quad |\omega- \omega_{0,1}(|\bk|)|\geq C \, |\bk|, \quad  |\omega- \omega_{0,2}(|\bk|)|\geq C \, |\bk|,
\end{equation*}
Thus from \eqref{eq.polynomefact}, we infer that
\begin{equation}\label{estiD}
\forall \;  \omega \in  \mathcal{C}_{ 0, 1, |\bk| }, \quad |D_{|\bk|}(\omega)|^{-1} \lesssim  |\bk|^{-2}.
\end{equation}
Thus, as the polynomial $Q_e(\omega) Q_m(\omega)$ remains obviously bounded for bounded $\omega$, we deduce that  there exists $k_->0$ such that
\begin{equation}\label{eq.Dinvzero}
\forall \;  \omega \in  \mathcal{C}_{ 0, 1, |\bk| } , \quad |\mathcal{D}(\omega)-|\bk|^2|^{-1} \lesssim  |\bk|^{-2},  \quad \mbox{ for } 0<|\bk|\leq k_-.
\end{equation}
\noindent {\bf Step 2: Estimates of $\mathcal{S}_{|\bk|}(\omega)$ and $\mathcal{V}_{|\bk|}(\omega)$}\\[4pt]
\noindent First, it follows from  their definitions \eqref{operatorsApAm} that $\bbA_{e,j}(\omega)$, $\dot \bbA_{e,j}(\omega)$, $\bbA_{m,\ell}(\omega)$ and $ \dot  \bbA_{m,\ell}(\omega)$ are uniformly bounded in operator norm when $\omega$ does not approach $0$ (indeed $0$ is not a zero of the polynomials $q_{e,j}$ and $q_{m,\ell}$).  Thus $\|\bbA_{e}(\omega\| \lesssim 1$ and $\|\bbA_{m}(\omega\| \lesssim 1$ according to  \eqref{operatorsAeAh}.\\ [12pt]
Moreover,  $\mu(\omega)=\mu(0)+o(1)$ as $\omega \to 0$ and $|\omega|\lesssim |\bk|$ on $\mathcal{C}_{ 0, 1, |\bk| }$, thus,  for $|\bk|$ small enough, 
\begin{equation}\label{op.partSzero}
\forall \;  \omega \in  \mathcal{C}_{ 0, 1, |\bk| }, \quad  \|- \omega \mu(\omega ) \bbA_e(\omega) +|\bk| \, {\bf e_3} \times \bbA_m(\omega)\| \lesssim |\omega| + |\bk| \lesssim |\bk|.
\end{equation}
Combining  \eqref{eq.Dinvzero} and \eqref{op.partSzero} in  \eqref{defS}  yields that there exists  $k_->0$:
\begin{equation}\label{op.Szero}
\forall \;  \omega \in  \mathcal{C}_{ 0, 1, |\bk| }, \quad \|\mathcal{S}_{|\bk|}(\omega)\|\lesssim |\bk|^{-1} , \quad  \mbox{ for }0 <|\bk| \leq k_-.
\end{equation}
Finally, as, along  $\mathcal{C}_{ 0, 1, |\bk| }$, 
$|\omega|\geq \big| \rho_{|\bk|,1}- |\omega_{0,1}(|\bk|)|\big| \sim   \frac{1}{2} \;  |A_0|^{\frac{1}{2}} \; |\bk| $  as $|\bk|\to 0$, we deduce that 
 $$
 \forall \;  \omega \in  \mathcal{C}_{ 0, 1, |\bk| }, \quad  \frac{|\bk|}{|\omega| \, |\mu(\omega)|} \lesssim \frac{|\bk|}{|\omega|} \lesssim 1, 
\quad  \mbox{ for }0 <|\bk| \leq k_-. 
 $$
  Since $|q_{e,j} (\omega)|^{-1}$ and $|q_{m,\ell} (\omega)|^{-1}$ remain bounded along $\mathcal{C}_{ 0, 1, |\bk| }$, by the
  formula \eqref{defV} for $\mathcal{V}_{|\bk|}(\omega)$, 
  \begin{equation}\label{op.Vzero}
\|\mathcal{V}_{|\bk|}(\omega)\|\lesssim 1 , \quad \forall  \; \omega \in  \mathcal{C}_{ 0, 1, |\bk| } \ \mbox{ and } \ 0 <|\bk| \leq k_-.
\end{equation}
\noindent {\bf Conclusion:} 
  Combining \eqref{op.Szero} and  \eqref{op.Vzero} in  \eqref{estipojero} finally  implies  that there exists $k_->0$ such that $\Pi_{0,1} (|\bk|)$ is uniformly bounded for $0<|\bk|\leq k_-$.
\end{proof}
\noindent The next lemma is about the asymptotic expansion (in powers of $|\bk|$) of  $\omega_{0,r}(|\bk|)$ when $|\bk|\to 0$. It is important to push the expansion up to the first appearance of a negative imaginary part, since it will govern the decay  of $\bbU_{z,0}(\bk,t)$ for small values of $|\bk|$. 

\begin{Lem}  \label{LemEigenzo}
	For $r\in \{1,2\}$, the eigenvalue $\omega_{0,r}(|\bk|)$ satisfies the following asymptotic expansions, with $c_0=(\varepsilon(0)\mu(0))^{-1/2}>0$,
	\begin{equation}\label{eq.asymptz0}
	\omega_{z,r}(|\bk|)= (-1)^r \, c_0 \,   |\bk|  - \mbox{$\frac{1}{2}$}(\varepsilon \mu)'(0) \, c_0^4  \,  |\bk|^{2}    + o(|\bk|^{2}),\ \mbox{ as } |\bk|\to 0,
	\end{equation}  
with moreover 
\begin{equation}\label{eq.const0}
\operatorname{Im} \, \big( - (\varepsilon \mu)'(0) \big) <0.
 \end{equation}
\end{Lem}

\begin{proof}
Using \eqref{eq.equivalentzero}, one has $\mathcal{D}(\omega)=\omega^2 g(\omega)$ where the rational function $g=\varepsilon(\cdot) \,  \mu(\cdot)$ is analytic in the vicinity of $0$ and satisfies  $g(0)=\varepsilon(0) \mu(0)>0$. Thus applying Lemma \ref{Lem-implicte-function} with  $z=0$, $\mathfrak{m}_0=2$, $\zeta=|\bk|$, $a_1=-g(0)^{1/2}=-c_0^{-1}$ and $a_2=c_0^{-1}$ , the asymptotic formula  \eqref{eq.asympexpansion} yields:
 $$
 \omega_{z,r}(|\bk|)=a_{r}^{-1}   |\bk|^{1} -\frac{g'(0)}{2\,g(0)^2} |\bk|^{2}    + o(|\bk|^{2}), \quad \mbox{ as } |\bk| \to 0.
 $$
This is precisely the asymptotic \eqref{eq.asymptz0}. It remains to prove \eqref{eq.const0}.\\ [12pt]
As  $\mu(0)>0$ and $\varepsilon(0)>0$,  it is equivalent to show that  $ \operatorname{Im}g'(0)=\mu(0) \operatorname{Im}\varepsilon'(0)+ \varepsilon(0)\operatorname{Im}\mu'(0)$ is positive. The Taylor expansion of the rational functions  $\varepsilon$ and $\mu$ (see  \eqref{eq.permmitivity-permeabiity} and  \eqref{eq.polynom}) gives
$$
\operatorname{Im}\varepsilon'(0)=\varepsilon_0 \, \sum_{j=1}^{N_e} \frac{ \alpha_{e,j}\, \Omega_{e,j}^2 }{\omega_{e,j}^4} \geq 0 \  \mbox{ and }  \ \operatorname{Im}\mu'(0)= \mu_0 \, \sum_{\ell=1}^{N_m} \frac{  \alpha_{m,\ell}\, \Omega_{m,\ell}^2}{\omega_{m,\ell}^4} \geq 0.
$$
Furthermore,  by the weak dissipation condition \eqref{WD} at least one coefficient $\alpha_{e,j}$ or $\alpha_{e,m}$ is positive, thus one has  $\operatorname{Im}\varepsilon'(0)>0$ or $\operatorname{Im}\mu'(0)>0$ and we can conclude. 
\end{proof}

\noindent Combining the Lemma \ref{LemProzerodouble} and Lemma \ref{LemEigenzo} (formula  \eqref{eq.asymptz0} and  \eqref{eq.const0}) and proceeding as in the proof  of  Lemma \ref{LemEstiinfty} gives the following estimate for $\bbU_{z,0}(\bk,t)$.
 
 \begin{Lem}\label{LemEstizero}
There exists $k_->0$ and $C>0$  such that   the function  $\bbU_{z,0}(\bk,t)$   defined by  \eqref{eq.decompositionthreeterm} and \eqref{decompUbis},  satisfies the following estimate
	\begin{equation}\label{eq.estimateU0}
		| \bbU_{z,0}(\bk,t)| \lesssim    \rme^{- C \, |\bk|^2\,  t } \, |\bbU_0(\bk)|,   \quad \forall t \geq 0 \ \mbox{ and }    \ 0<|\bk|\leq k_-.
	\end{equation}
\end{Lem}

\subsubsection{Estimates of $\bbU_{z,s}(\bk,t)$ for $0<|\bk|\leq k_-$} \label{estiUZS}
\noindent This time we estimate $\bbU_{z,s}(\bk,t)$ in  \eqref{eq.decompositionthreeterm}  which involves in particular the projectors $\Pi_{z} (|\bk|)$.
\begin{Lem} \label{LemProszero}
There exists $k_->0$ such that  the spectral projectors  $\Pi_{z} (|\bk|)$, $z \in \mathcal{Z}_s$ are uniformly bounded for $0<|\bk|\leq k_-$.  
\end{Lem}
 \begin{proof} The proof of this lemma has already been sketched in Section \ref{orientation-LF} (paragraph about $\Pi_{z} (|\bk|)$). The details are left to the reader. 
 \end{proof}
\noindent We now give the asymptotic expansion of the  eigenvalues $\omega_{z}(|\bk|)$ for  small $|\bk|$.
\begin{Lem}  \label{LemEigenzs}
	Let $z\in \calZ_s$. The eigenvalue $\omega_{z}(|\bk|)$ satisfies the following asymptotic expansion 
	\begin{equation}\label{eq.asymptzs}
	\omega_{z}(|\bk|)=z+ A_{z}  |\bk|^{2}+o(|\bk|^{2}),\ \mbox{ as } |\bk|\to 0,
	\end{equation}  
	where  the complex number $A_z$ is defined by (two disjoint cases have to be distinguished):
	\begin{equation}\label{eq.constCzs}
\mbox{ if } z\in \calZ_e,  \quad  A_{z}=\frac{(\omega \varepsilon)'(z)^{-1} }{z \mu(z)}, \quad  \displaystyle \mbox{ if } z\in \calZ_m,\quad  A_{z}=\frac{\big(\omega \mu\big)'(z)^{-1} }{z \varepsilon(z)},  
 \end{equation}
and satisfies in all cases $\operatorname{Im} \, A_z < 0$.
 \end{Lem}
\begin{proof}
The elements of $\mathcal{Z}_s$ are simple zeroes ($\operatorname{m}_z=1$) of the rational function $\mathcal{D}$, it means that either $z\in \calZ_e$ or $z\in \calZ_m$. We will prove the Lemma by assuming that  $z\in \calZ_e$, the proof for   $z\in \calZ_m$ can be done with obvious ``symmetric arguments''. \\[6pt]
\noindent Using Proposition \ref{prop.dispersioncurvesLF} and more precisely the asymptotic expansion \eqref{eq.zero} (with here $\mathfrak{m}_z=1$) yields \eqref{eq.asymptzs} with $A_{z}=g(z)^{-1}$ where $g$ is defined by \eqref{eq.zerovoisinage}, that is to say
$$g(\omega)=\frac{\mathcal{D}(\omega)}{\omega-z}=\frac{\omega \varepsilon(\omega)}{\omega-z}  \, \omega \mu(\omega) \  \mbox{ for } \omega \in \mathbb{C} \setminus (\calP\cup \{ z\} ).$$
Since $z\in \calZ_e$ is a  real  simple zero of $\omega \varepsilon(\cdot)$ thanks to assumption $(\mathrm{H}_2)$, i.e. $\calZ_e \cap \calP_m=\varnothing$, it is not a zero nor a pole  $\omega \mu(\cdot)$.
Thus, $g$ can indeed be extended analytically for $\omega = z$ with 
$$
g(z)=(\omega \varepsilon)'(z) \,  z \mu(z).
$$
Thus, one has $ \ds A_{z}=\frac{1}{g(z)}=\frac{1}{z \mu(z) (\omega \varepsilon)'(z)}$ that is \eqref{eq.constCzs}.\\ [12pt] 
We  show now that $\operatorname{Im} \, A_z < 0$.  As $z\in  \mathcal{Z}_e$,  one observes on one hand  (see remark \ref{remZs}) that  $\mathcal{Z}_s= \mathcal{Z}_e$  and that all the $\alpha_{e,j}$ vanish. Thus, 
$\omega \varepsilon(\omega)$ is real-valued on the real axis  (outside the poles  $\calP_e$ of $\varepsilon(\omega)$) and from expression \eqref{eq.permmitivity-permeabiity}, one easily deduces  that  $(\omega \varepsilon)'(\omega)>0$ on  $\R\setminus \calP_e$. In particular, as $z \in \R\setminus \calP_e$, one has
$$
(\omega \varepsilon)'(z) > 0.
$$ 
On the other hand, as  all the $\alpha_{e,j}$ vanish, one knows by the weak dissipation condition \eqref{WD} that at least one coefficient $\alpha_{m,\ell}>0.$ Therefore  by the formula \eqref{eq.positvity}, $\operatorname{Im}(z \mu(z))>0$. Thus, one deduces that  $$\operatorname{Im}\big ((z \mu(z))^{-1}\big)=-\operatorname{Im}(z \mu(z))/ |z \mu(z)|^2<0.$$ Together with $(\omega \varepsilon)'(z) > 0$,  this implies that $\operatorname{Im}(  A_{z})$ is  negative.
\end{proof}
\noindent Combining the Lemma \ref{LemProszero} and Lemma \ref{LemEigenzs} (formula  \eqref{eq.asymptzs} and  \eqref{eq.constCzs}) and proceeding as in the proof  of Lemma \ref{LemEstiinfty} gives the following estimate for $\bbU_{z,s}(\bk,t)$.
 \begin{Lem}\label{LemEstizeros}
If  $\mathcal{Z}_s\neq \varnothing$, then there exists $k_->0$ and $C>0$  such that   the function  $\bbU_{z,s}(\bk,t)$   defined by defined by  \eqref{eq.decompositionthreeterm} and \eqref{decompUbis},  satisfies the following estimate
	\begin{equation}\label{eq.estimateUZS}
		| \bbU_{z,s}(\bk,t) | \lesssim    \rme^{- C \, |\bk|^2\,  t }  \, |\bbU_0(\bk)|,   \quad \forall t \geq 0 \ \mbox{ and }  \  0<|\bk|\leq k_-.
	\end{equation}
\end{Lem}

\subsubsection{Estimates of $\bbU_{z,-}(\bk,t)$ for $0<|\bk|\leq k_-$} \label{estiUZminus}
As announced in Section \ref{orientation-LF}, since we simply want to obtain a ``rough" exponential decay estimate for  $\bbU_{z,-}(\bk,t)$. We  give a direct proof  of it using Riesz-Dunford functional calculus.
\begin{Lem}\label{EstimZminus}
There exists  $\delta>0$ and $k_->0$ such that $\bbU_{z,-}(\bk,t)$, defined by  \eqref{eq.decompositionthreeterm} and \eqref{decompUbis},   satisfies
	\begin{equation}\label{eq.estimateUm2}
		| \bbU_{z,-}(\bk,t)| \lesssim \rme^{-\delta\, t} \ | \bbU_0(\bk)|, \quad  \forall \, t\geq 0, \quad  \forall \;  0<|\bk| \leq k_-.
	\end{equation}
\end{Lem}
\begin{proof}
The proof is similar to that of Lemma \ref{Estiminus} and even shorter since $|\bk|$ is  bounded here.
We introduce a fixed (positively oriented)  simple closed contour $\Gamma$, included in $\bbC^-$ such that  all the zeros of $\calZ_-$  lie inside $\Gamma$  (see figure \ref{fig-disp-curv-LF}). 
We choose $\Gamma$ such that $\Gamma\cap \mathcal{P}=\varnothing$.
We denote by  
\begin{equation}\label{eq.distLF}
\delta=\min \{-\operatorname{Im}(\omega), \, \omega \in \Gamma \}>0
\end{equation}
 the distance from $\Gamma$ to the real axis.  \\[12pt]
 \noindent For $|\bk|$ positive and small enough, $\bbU_{z,-}(\bk,t)$ is on one hand well-defined by   \eqref{eq.decompositionthreeterm} and \eqref{decompUbis} and on the other hand, by Proposition \ref{prop.dispersioncurvesLF},
$\Gamma$ encloses all eigenvalues $\omega_{p,n}(|\bk|)$ for $z\in \calZ_-$ and $n\in \{ 1, \ldots, \mathfrak{m}_z\}$ but no other elements of the spectrum of $\bbA_{|\bk|, \perp}$. Thus by the Riesz-Dunford functional calculus, one gets that there exists $k_->0$ such that
\begin{equation}\label{eq.ineqUZm}
\bbU_{z,-}(\bk,t)=- \frac{\, \mathcal{R}_{\bk}^*} {2\rmi \pi}\int_{\Gamma}  e^{- \rmi \omega t} \,  R_{|\bk|}(\omega) \, \, \mathcal{R}_{\bk} \bbU_0(\bk) \, \rmd \omega , \ \mbox{ for }0 <|\bk|\leq k_-,
\end{equation}
where  the resolvent $R_{\bk}(\omega)$ is well-defined by Proposition  \ref{Prop.res}   for $( |\bk|,\omega) \in K= [0,k_-] \times \Gamma $.  Moreover,   the function   $(|\bk|, \omega)\mapsto R_{\bk}(\omega) $, valued in $\mathcal{L}(\bC_{\perp})$, is continuous  on the compact  $K$.
Then, using the fact that $\mathcal{R}_{\bk}$ is unitary and the definition \eqref{eq.distLF} of $\delta$ (which implies that $|\rme^{-\rmi \omega t}|\leq \rme^{- \delta \,  t} $ on $\Gamma$), it follows from  \eqref{eq.ineqUZm} that
$$
| \bbU_{z,-}(\bk,t)| \leq C \,  \rme^{-\delta\, t} \ | \bbU_0(\bk)|, \  \forall \; |\bk|\in (0, k_-] \ \mbox{  where } C=  \max_{(|\bk|, \omega) \in K} \|R_{|\bk|}(\omega)\|.
$$

\end{proof}
\subsubsection{The global estimate} \label{EstiglobaleLF} 

\begin{Thm}\label{LF-estm}
There exists $k_->0,$ $C, \, \widetilde{C}> 0$  such that for $0<|\bk|\leq k_-$, the spatial Fourier components $|\bbU(\bk,t)|$ of the solution   of \eqref{eq.schro} with initial condition $\bU_0\in \mathcal{H}_{\perp}$ satisfy:
\begin{equation} \label{polynomial_decay-LF}
|\bbU(\bk,t)| \leq   \widetilde{C}  \, \rme^{-C |\bk|^2\, t} |\bbU_0(\bk)|,  \quad \forall t \geq 0    .
\end{equation}
\end{Thm}
\begin{proof}
The proof is similar to that  of Theorem \ref{HF-estm} using Lemmas  \ref{LemEstizero}, \ref{LemEstizeros}  and \ref{EstimZminus}. \end{proof} 
\noindent  We prove in the following result that the estimates of Theorem \ref{LF-estm}  is optimal for an infinite family of well chosen  initial conditions $\bU_0\in \mathcal{H}^{p}_{\perp,\operatorname{LF}}$ for any fixed  $p\geq 0$.
 \begin{Thm}\label{LF-estm-opt}
Let $p\geq 0$  and $ k_->0$ and  $\phi: \bbR^+ \mapsto \bbR$  be any    measurable and bounded function satisfying
\begin{equation}\label{eq.phik-LF}
\operatorname{supp} \phi \subset  [0,k_-] \mbox{ and  }  0< |\phi(|\bk|)|\lesssim |\bk|^p \, \mbox{ for } \, |\bk| \leq k_- .
\end{equation}
If the initial condtion $\bU_0$  of \eqref{eq.schro} is defined (for $k_-$ small enough) via its Fourier transform:
\begin{equation}\label{eq.defUOkoptLF}
\bbU_0(\bk)= \mathcal{F}(\bU_0)(\bk)=  \phi(|\bk|) \; \mathcal{R}_{\bk}^* \, \frac{\mathcal{V}_{|\bk|}\big(\omega_{0,1}(|\bk|)\big)  \mathbf{e}_{1}}{|\mathcal{V}_{|\bk|}\big(\omega_{0,1}(|\bk|)\big)  \bf{e}_1|}, \quad \forall \, \bk\in \R^{3,*},
\end{equation}
then $\bU_0\in  \mathcal{H}^{p}_{\perp,\operatorname{LF}}$ and     $\exists \; C,\,\widetilde{C}>0$  such that  the associated  solution $\bU$ of \eqref{eq.schro}  satisfy
	\begin{equation} \label{polynomial_decayncr-LF2}
	  \widetilde{C}  \, \rme^{C\, t\, |\bk|^2} \, |\bbU_0(\bk)|\leq |\bbU(\bk,t)|,  \quad \forall \, t \geq 0      \mbox{ and } \forall \, \bk\in \R^{3,*}.
	\end{equation}
	In other words, the  estimate   \eqref{polynomial_decay-LF} is optimal for an infinite family of solutions.
	\end{Thm}
 \begin{proof}
Using \eqref{eq.phik-LF} and \eqref{eq.defUOkoptLF}, it is clear that $\bU_0\in  \mathcal{H}^{p}_{\perp,\operatorname{LF}}$. The rest of the proof     uses the following  asymptotic expansion of $\operatorname{Im}\omega_{0,1}(|\bk|)$ (given by Lemma \ref{LemEigenzo}): $$\operatorname{Im}\omega_{0,1}(|\bk|)= \operatorname{Im} \, \big( -(\varepsilon \mu)'(0) \big)\, |\bk|^2+o(|\bk|^2), \ \mbox{ as }\ |\bk|\to 0 \ \mbox{ with } \ \operatorname{Im} \, \big( - (\varepsilon \mu)'(0) \big) <0$$  and is similar to the proof of Theorem \ref{HF-estm-opt}. Therefore, the details are left to the reader.
 \end{proof}

\section{Estimates of ``mid-range frequencies'' components of the solution}\label{sec_mid-frequencies}
The goal is to prove (see Theorem \ref{Th.intermediatefreq})  that for intermediate frequency  components, i.e. for $0<k_-\leq|\bk|\leq k_+$, the Fourier components of the solution  $\bbU(\bk,t)$ decay exponentially with a uniform exponent  depending only on the  compact set $K=[k_-,k_+]$. The proof  is based on three key ingredients: the Dunford  decomposition in linear algebra, the following perturbation Lemma and  finally a compactness argument.
\begin{Lem}\label{Pro-perturb}
Let $\bbM\in \mathcal{L}(\bbC_{\perp}^N)$. We assume that there  exists two constants $C, \alpha>0$ such that
$$
\|\rme^{t \, \bbM} \| \leq C e^{-\alpha t},  \quad \forall t \geq 0.
$$
Then,  for  any perturbation $\Delta \in  \mathcal{L}(\bbC_{\perp}^N)$, one has the following estimate:
$$
\|\rme^{t \, (\bbM+\Delta)} \| \leq C e^{(-\alpha+ C\, \|\Delta\| )\, t},  \quad \forall t \geq 0.
$$
 \end{Lem}
 \begin{proof}
The proof, based on the Duhamel formula and the Gr\"{o}nwall's lemma, is done  e.g. in lemma 1.6 page 97 of  \cite{Bec-Jol-02} or in the Proposition 4.2.18  page 406 of  \cite{Hin-10}.
\end{proof}
\begin{Thm}\label{Th.intermediatefreq}
Let $K=[k_-, k_+]$ be a compact interval of $\bbR^{+,*}$,  then  there exist  two constants $C>0$ and $\beta>0$ (depending only on $K$) such that the spatial Fourier components $|\bbU(\bk,t)|$ of the solution   of \eqref{eq.schro} with initial condition $\bbU_0\in \mathcal{H}_{\perp}$ satisfy
\begin{equation}\label{eq.freqintermestm}
|  \bbU( \bk,t)| \leq C\,  e^{-\beta t}  \, |\bbU_0( \bk)|, \ \forall  \;t\geq 0   \mbox{ and  a.e. }   \bk \in \bbR^3 \mid   |\bk| \in K .
\end{equation}
\end{Thm}

\begin{proof}
Let $k \in K$ be fixed. 
By virtue of   \eqref{loc_spectrum},  $\sigma(\bbA_{k, \perp})$ is included in the lower half-plane  $\mathbb{C}^-$, thus it follows from the Dunford Decomposition of $\bbA_{ k, \perp}$ (see  e.g. Corollary 2.26  page 106 of \cite{Hin-10}) that there exist $\alpha_k>0$ and $C_{k} >0$ such that
 \begin{equation}\label{eq.ineqinter}
 \|\rme^{-  \rmi \bbA_{k, \perp} t} \| \leq C_{k} \; e^{-\alpha_{k} t},  \quad  \forall \; t\geq 0.
 \end{equation}
As the function $|\bk| \mapsto -\rmi  \, \bbA_{|\bk|, \perp}$ is clearly continuous from $\bbR^{+,*}$ to $\mathcal{L}(\bC_{\perp}^N )$,
There exists an open interval $(k-\eta_{k},k+\eta_{k})$  such that 
\begin{equation}\label{eq.perutbsize}
\|-\rmi (\bbA_{|\bk|, \perp}- \bbA_{k, \perp}) \|< \frac{\alpha_{k}}{ 2\,  C_{k}}, \quad  \forall \; |\bk|\in (k-\eta_{k},k+\eta_{k})\cap K.
\end{equation}
Hence, applying  Lemma \ref{Pro-perturb}  with  $$\bbM=- \rmi  \, \bbA_{k, \perp} \ \mbox{ and } \ \Delta=- \rmi \, (\bbA_{|\bk|, \perp}- \bbA_{k, \perp}) $$ yields the following inequality for all $ t\geq 0 \ \mbox{ and } \forall |\bk| \in (k-\eta_{k},k+\eta_{k})\cap K $: 
\begin{equation}\label{eq.estmlocal}
\|\rme^{-  \rmi \bbA_{ |\bk|, \perp} t} \| \leq C_{k} \, e^{-\beta_{k} t} \quad \mbox{with}  \quad \beta_{k}=\alpha_{k} -C_{k} \|\Delta\|.
\end{equation}
Thus, one deduces   immediately from  \eqref{eq.perutbsize} and \eqref{eq.estmlocal} that
\begin{equation}\label{eqeq.estmlocal2}
\|\rme^{-  \rmi \bbA_{ \bk, \perp} t} \| \leq C_{k} \,  e^{-\frac{\alpha_{k}}{2}t}, \quad \forall \; t\geq 0  \ \mbox{ and } \  \forall \;  |\bk|\in (k-\eta_{k},k+\eta_{k})\cap K.  
\end{equation}
As the compact $K\subset \displaystyle \bigcup_{k\in K} (k-\eta_{k},k+\eta_{k}),$  there exists $N>0$ and $k_1, k_2, \ldots k_N\in K$ such that
\begin{equation}\label{eq.compactfiniterecovering}
K \subset \displaystyle \bigcup_{1\leq n\leq N}(k_n-\eta_{k_n},k_n+\eta_{k_n}).
\end{equation}
With $\alpha_{k_i}>0$ and $C_{k_i}>0$ given by \eqref{eq.ineqinter} for $k=k_i$, we define
$$C=\displaystyle \max_{i=1,\ldots, N} C_{k_i}>0\mbox{ and  } \quad  \alpha= \min_{i=1,\ldots, N} \alpha_{k_i}>0.  $$
Then, combining  \eqref{eqeq.estmlocal2} and  \eqref{eq.compactfiniterecovering} yields
$$
\|\rme^{-  \rmi \bbA_{| \bk|, \perp} t} \| \leq C e^{-\frac{\alpha}{2} t},  \quad \ \forall \; t\geq 0 \ \mbox{ and } \ \forall \; |\bk| \in K.
$$
As  the operator $\mathcal{R}_{\bk}$ is unitary, one finally deduces, with \eqref{eq.refsolutionfourier2}, the estimate \eqref{eq.freqintermestm}.
\end{proof}

\section{Proof of the main Theorems  of the paper}\label{proof-general-cases}

\subsection{Proof of Theorem \ref{thm_Lorentz} (decay rate estimates)}\label{sec-final-step}
This section constitutes  the last step of the proof of Theorem  \ref{thm_Lorentz}. For this final stage, we follow the approach developed in the proofs of Theorems 2.4 and 4.4 of  \cite{cas-jol-ros-22}.
For the sake of readability, we recall here the arguments, which rely on Plancherel's identity and the three decay estimates proved for high, low and intermediate frequencies  respectively in Theorems \ref{HF-estm}, \ref{LF-estm} and  \ref{Th.intermediatefreq}.\\[12pt]
\noindent 
Let $k_-$ and $k_+$ be two fixed positive real numbers satisfying $k_-<k_+$ such that the estimates \eqref{polynomial_decayncr-HF}  and \eqref{polynomial_decay-LF} of Theorems \ref{HF-estm} and \ref{LF-estm} hold. \\[12pt]
\noindent Using the Plancherel identity, any solution of \eqref{eq.schro}  with initial condition $\bU_0\in \mathcal{H}_{\perp}$ satisfies 
\begin{eqnarray}\label{eq.Plancherel}
\|\bU(t)\|_{\mathcal{H}}^2
&=&\int_{|\bk|<k_-}|\bbU(\bk,t)|^2 \, \rmd\bk+\int_{k_-\leq |\bk|\leq k_+}|\bbU(\bk,t)|^2 \, \rmd\bk+\int_{k_+<|\bk|}|\bbU(\bk,t)|^2 \, \rmd\bk.
\end{eqnarray}

\noindent {\bf Step 1: proof of  the convergence result \eqref{convergenceL}}\\[4pt]
In a non critical configuration,  using the low frequency  estimate \eqref{polynomial_decay-LF}, the intermediate frequency estimate \eqref{eq.ineqinter} and  the high frequency estimate \eqref{polynomial_decayncr-HF}  of  $\bbU(\bk,t)$, one  sees that  there  exist  three constants $C_1>0$, $C_2>0$ and $C_3>0$ (independent of $\bU$)  such that
\begin{equation}\label{eq.globestim}
\|\bU(t)\|_{\mathcal{H}}^2 \lesssim  \int_{|\bk|<k_-}  \rme^{-C_1  |\bk|^2 \, t } \,  |\bbU_0(\bk)|^2  \, \rmd\bk +   \rme^{-C_2  \, t } \, \|\bU_0 \|^2_{\mathcal{H}}+  \int_{ k_+<|\bk| }  \rme^{-C_3  |\bk|^{-2} \, t } \,  |\bbU_0(\bk)|^2  \, \rmd\bk.
\end{equation}
When $t\to +\infty$, the second term in \eqref{eq.globestim} converges exponentially to $0$,  whereas the first  and  third terms tend to $0$ by the Lebesgue's dominated convergence Theorem. \\ [12pt] 
In a critical configuration,  one obtains  \eqref{convergenceL} by the same reasoning  only after replacing the high frequency estimate \eqref{polynomial_decayncr-HF} by
 \eqref{polynomial_decaycr-HF} (thus it consists to substitute the factor $ \rme^{-C_3  |\bk|^{-2} \, t }$ by  $ \rme^{-C_3  |\bk|^{-4} \, t }$ in the second integral of   \eqref{eq.globestim}).\\[3pt]

 \noindent {\bf Step 2: Estimate of the low frequency term in \eqref{eq.globestim}}\\[12pt]
We assume now that  the initial condition $\bU_0\in \mathcal{H}_{\perp} \cap \boldsymbol{\cal L}_{p}^N $  for some  integer $p\geq 0$. \\ [12pt]
\noindent Thus, by  definition on the ${\boldsymbol{\cal L}_p^N}$ norm (see \eqref{normcalp}), 
$	|\bbU_0(\bk)| \leq |\bk|^p \;   \|{\bU}_0\|_{\boldsymbol{\cal L}_p^N}$ for a.e. $\bk\in \R^3$
and consequently
\begin{eqnarray}\label{eq.finalestimateLF}
  \int_{|\bk|<k_-}  \rme^{-C_1  |\bk|^2 \, t } \,  |\bbU_0(\bk)|^2  \, \rmd\bk 
&   \leq & \Big( \int_{\bbR^3}  \rme^{-C_1  |\bk|^2 \, t } \, |\bk|^{2p}  \, \rmd \bk\Big) \,  \big\| \mathbf{U}_0 \big\|_{\boldsymbol{\cal L}_p^N}^2.
\end{eqnarray}
Since,  with the change of variable $\mathbf{\xi }= \sqrt{C_1\, t} \; \bk$ ,
	$$
	\int_{\R^3} |\bk|^{2p}\, e^{- C_1 \, |\bk|^2 t} \; \rmd \bk = 
	\frac{1}{(C_1\, t)^{p + \frac{3}{2}}} \; \int_{\R^3} |\xi|^{2p}\, e^{- |\xi|^2} \; \rmd \mathbf{\xi } \equiv \frac{C(p)}{ t^{p + \frac{3}{2}}} ,
	$$
 one concludes from  \eqref{eq.finalestimateLF} that
\begin{equation}\label{eq.finalestimateLF-2}
  \int_{|\bk|<k_-}  \rme^{-C_1  |\bk|^2 \, t } \,  |\bbU_0(\bk)|^2  \, \rmd\bk 
  \leq  \frac{C(p)}{ t^{p + \frac{3}{2}}}  \  \big\| \mathbf{U}_0 \big\|_{\boldsymbol{\cal L}_p^N}^2.
\end{equation}

 \noindent {\bf Step 3: Estimate of the high frequency term in \eqref{eq.globestim}}\\[6pt]
 {\bf In a non critical configuration:}\\[6pt]
We assume  here that  the Maxwell's system is in a non critical configuration and that the initial condition $\bU_0\in  \mathcal{H}_{\perp}  \cap  \bH^m(\bbR^3)^N$ for some $m>0$.
We estimate here the third term of   \eqref{eq.globestim}  corresponding to the high frequency contribution.
To this aim, we denote by $\langle\bk\rangle$ the quantity: $\langle\bk\rangle:=(1+|\bk|^2)^{1/2}$. Then,   using the fact that $e^{- C_3 |\bk|^{-2} \,t} \leq e^{- C_3 \langle\bk\rangle^{-2} \,t}$, we rewrite  this third term  as follows (we simply introduce artificially the factor $\langle\bk\rangle^m/t^m$)
\begin{equation} \label{eq.LebesguesHm}
 \begin{array}{lll}
 \displaystyle  \int_{ k_+<|\bk| }  \rme^{-C_3  |\bk|^{-2} \, t } \,  |\bbU_0(\bk)|^2  \, \rmd\bk & \leq & \ds t^{-m} \int_{{ k_+<|\bk| }} \langle\bk\rangle^{2m}  |\bbU_0(\bk)|^2 \;  \Big( \frac{t}{\langle\bk\rangle^2}\Big)^m \; e^{- C_3 \,\langle\bk\rangle^{-2}  \,t} \; \rmd \bk \nonumber \\ [18pt] 
& \leq & \ds   t^{-m} \int_{\R^3} \langle\bk\rangle^{2m}  \;   |\bbU_0(\bk)|^2  \,  F_m\big(t/\langle\bk\rangle^{2}\big) \; \rmd \bk,
\end{array}
\end{equation} 
where we have set $F_m(r) := r^{m} \, e^{- C_3 \, r}, r \geq 0$ which satisfies: $\displaystyle\sup_{r\geq 0} F_m(r)=\tilde{C}_m:= (m/(C_3 \,\mathrm{e}))^{m}$. \\ [2pt]
Thus, by the Fourier characterization of Sobolev norms:
\begin{equation}\label{eq.finalestHFNcr}
\int_{ k_+<|\bk| }  \rme^{-C_3  |\bk|^{-2} \, t } \,  |\bbU_0(\bk)|^2  \, \rmd\bk \lesssim  \tilde{C}_m \,t^{-m} \|\bU_0\|^2_{\bH^m(\bbR^3)^N}, \quad \forall \; t > 0.
\end{equation}

 \noindent {\bf In a  critical configuration:}\\[6pt]
We assume  here that  the Maxwell's system is in a  critical configuration with again an initial condition $\bU_0\in  \mathcal{H}_{\perp} \cap  \bH^m(\bbR^3)^N$.
The proof is similar to that of the estimate \eqref{eq.finalestHFNcr}. We simply enlightens the differences.
 In a critical configuration, one has to replace  the  estimate \eqref{polynomial_decayncr-HF} by 
 \eqref{polynomial_decaycr-HF}. Thus, in  \eqref{eq.globestim}, one has to substitute the factor $ \rme^{-C_3  |\bk|^{-2} \, t }$ by  $ \rme^{-C_3  |\bk|^{-4} \, t }$ in the second integral of  \eqref{eq.globestim}. Then, one has
 \begin{eqnarray*}
 \displaystyle  \int_{ k_+<|\bk| }  \rme^{-C_3  |\bk|^{-4} \, t } \,  |\bbU_0(\bk)|^2  \, \rmd\bk & \leq & \ds t^{-m/2} \int_{{ k_+<|\bk| }} \langle\bk\rangle^{2m}  |\bbU_0(\bk)|^2 \;  \Big( \frac{t}{\langle\bk\rangle^4}\Big)^{m/2} \; e^{- C_3 \, \langle\bk\rangle^{-4}  \,t} \; \rmd \bk \nonumber \\ [12pt] 
& \leq & \ds   t^{-m/2} \int_{\R^3} \langle\bk\rangle^{2m}  \;   |\bbU_0(\bk)|^2  \,  F_{m/2}\big(t/\langle\bk\rangle^{4}\big) \; \rmd \bk.\end{eqnarray*}
Setting  $\tilde{C}_{m/2}:= \, \displaystyle\sup_{r\geq 0} F_{m/2}(r)$, this yields, by definition \eqref{Hmnorm} of the $\bH^m(\bbR^3)^N$ norm,
\begin{equation}\label{eq.LebesguesHmCR}
\int_{ k_+<|\bk| }  \rme^{-C_3  |\bk|^{-2} \, t } \,  |\bbU_0(\bk)|^2  \, \rmd\bk \leq \tilde{C}_{m/2} \;t^{-m/2} \; \|\bU_0\|^2_{\bH^m(\bbR^3)^N}, \quad \forall \; t > 0.
\end{equation}

\noindent {\bf  Step 4: Proof of the estimate  \eqref{polynomial_decayncr} and \eqref{polynomial_decaycr}}. \\[4pt]
In the non critical case (resp. the critical case), it suffices to substitute   \eqref{eq.finalestimateLF-2} and  \eqref{eq.finalestHFNcr} (resp.    \eqref{eq.finalestimateLF-2} and \eqref{eq.LebesguesHmCR}) into \eqref{eq.globestim} to obtain  \eqref{polynomial_decayncr} (resp. \eqref{polynomial_decaycr}).

\subsection{Proof of Theorem \ref{thm-optm} (optimality decay rate estimates)}\label{sec-optim}

\subsubsection{Optimality of the  high frequency  polynomial decay rate}\label{sec-optim-HF}
~\\ [-6pt]
{\bf Determination of  $\gamma_m^{\operatorname{HF}}$ in an non-critical configuration}\\[12pt]
\noindent We assume first that the Maxwell system is  in a non-critical configuration.\\[6pt]
\noindent First, we emphasize that  for any initial conditions $\bU_0 \in \mathcal{H}^{m}_{\perp, \operatorname{HF}}$ (see \eqref{eq.HperpHF} for the definition of  $ \mathcal{H}^{m}_{\perp, \operatorname{HF}}$) which does not contain any  low frequency  Fourier components (since as $\bU_0 \in \mathcal{H}^{m}_{\perp, \operatorname{HF}}$, $\operatorname{supp}(\bbU_0)\subset \bbR^3 \setminus B(0, k_+)$), we have the upper-bound given  by  \eqref{eq.globestim} and \eqref{eq.finalestHFNcr}, namely
\begin{equation}\label{upper-bound}
\|\bU(t)\|^2_{\mathcal{H}}\lesssim \displaystyle  \int_{ k_+<|\bk| }  \rme^{-C_3  |\bk|^{-4} \, t } \,  |\bbU_0(\bk)|^2  \, \rmd\bk \lesssim \frac{ \|\bU_0 \|_{ \bH^m(\R^3)^N}^2}{t^m}, \ \quad  \forall \,t>0.
\end{equation}
\noindent Clearly, \eqref{upper-bound} implies that $\gamma_m^{\operatorname{HF}}$  (defined  by \eqref{eq.gammaHF}) exists and satisfies  $\gamma_m^{\operatorname{HF}}\geq m$. To show that $\gamma_m^{\operatorname{HF}}=m$, we construct for any $\varepsilon>0$ an initial condition $\bU_{0,\varepsilon}\in \mathcal{H}^{m}_{\perp, \operatorname{HF}}$ such that  the associated solution of \eqref{eq.schro} satisfies
$$
\|\bU_{\varepsilon}(t)\|^2_{\mathcal{H}} \ \geq \frac{C}{t^{m+\varepsilon}},\, \quad \mbox{for some } C >0, \quad \forall t\geq1 .
$$
\noindent  To this aim,  according to Theorem \ref{HF-estm-opt}, we choose  $\bU_{0,\varepsilon}$  in the form \eqref{eq.defUOkopt}, more precisely
\begin{equation}\label{eq.defUOkopteps}
\left\{ 	\begin{array}{ll}
\ds	\bbU_{0,\varepsilon(\bk)}=  \phi_\varepsilon(|\bk|) \; \mathcal{R}_{\bk}^* \, \frac{\mathcal{V}_{|\bk|}\big(\omega_{+\infty}(|\bk|)\big)  \mathbf{e}_{1}}{|\mathcal{V}_{|\bk|}\big(\omega_{+\infty}(|\bk|)\big)  \bf{e}_1|}, 
\\ [18pt]
\ds \phi_{\varepsilon}(|\bk|) =  (1+|\bk|^2)^{-(\frac{3}{4}+\frac{m}{2}+\frac{\varepsilon}{2})} \; \mbox{ if } |\bk| \geq k_+, \mbox{ and }  \phi_{\varepsilon}(|\bk|) = 0 \mbox{ otherwise},
	\end{array} \right.
\end{equation}
where the exponent  $\frac{3}{4}+\frac{m}{2}+\frac{\varepsilon}{2}$ has been chosen just above in order to ensure the $H^m$ regularity of $\bU_{0,\varepsilon}$ (see \eqref{eq.phik}). As moreover $\bbU_{0,\varepsilon(\bk)}$ is supported in $\{|\bk| \geq k_+ \}$, $\bU_{0,\varepsilon}\in \mathcal{H}^{m}_{\perp, \operatorname{HF}}$ . \\ [12pt]
Therefore, applying the inequality \eqref{polynomial_decayncr-HF2} of Theorem  \ref{HF-estm-opt}, we get, as $|\bbU_{0,\varepsilon}(\bk)| =  \phi_\varepsilon(|\bk|) $, 
\begin{equation}\label{eq.optimal3}
\rme^{-\frac{2\,C\, t}{ |\bk|^2}}  (1+|\bk|^2)^{-(\frac{3}{2}+m+\varepsilon)} \lesssim  |\bbU_{\varepsilon}(\bk,t)|^2,   \   \mbox{ for a.e. } \bk\in \bbR^3 \mid  \; |\bk|\geq k_+ \mbox{ and } \, \forall \; t\geq 0.
\end{equation}
By Plancherel identity,  $ \ds 
\|\bU_{\varepsilon}(t)\|^2_{\mathcal{H}}=\int_{k_+<|\bk|}| \bbU_{\varepsilon}(\bk,t )|^2 \, \mathrm{d}\bk
$.
Therefore, using \eqref{eq.optimal3}, one gets 
\begin{equation}\label{eq.optimal4}
I_{\varepsilon}(t)=\int_{k_+<|\bk|}  \rme^{-\frac{2 C\, t}{ |\bk|^2}}  (1+|\bk|^2)^{-(3/2+ m+\varepsilon) }  \, \mathrm{d}\bk   \lesssim     \|\bU_{\varepsilon}(t)\|^2_{\mathcal{H}}.
\end{equation}
that is to say, the  integrand in \eqref{eq.optimal4} only depends on $|\bk|$, 
$$
I_{\varepsilon}(t)=4\pi \int_{k_+}^{+\infty}  \rme^{-\frac{2 \, C\, t}{ |\bk|^2}}  (1+|\bk|^2)^{-(3/2+  m+ \varepsilon) } \, |\bk|^2 \,  \mathrm{d}|\bk| .
$$
Using the change of variable $\xi =\sqrt{t}/|\bk|$ in $I_{\varepsilon}(t)$ with $ |\bk|^2 \,  \mathrm{d}|\bk| = -t^{\frac{3}{2}} \, \xi^{-4} \, \rmd \xi$   yields
\begin{eqnarray*}
I_{\varepsilon}(t) = 4\pi \, t^{\frac{3}{2}} \int_{0}^{\frac{\sqrt{t}}{k_+} } \rme^{- 2\,C\, \xi^{2} }\;  \xi^{3+2m+2\varepsilon} \; \big( \xi^{2} +t\big)^{-(3/2+  m+ \varepsilon) }\;  \frac{\mathrm{d}\xi}{\xi^4}   .  
\end{eqnarray*}
Then, using $\xi^2\leq t/k_+^2$ and thus that  $\xi^{2} +t \leq (1+k_+^{-2})  \, t $ on the domain of integration gives
$$
I_{\varepsilon}(t) \geq  \frac{{C}_m^\varepsilon(t)}{t^{m+\varepsilon}}, \quad  {C}_m^\varepsilon(t) :=  4 \pi  
(1+k_+^{-2})^{-(m + \frac{3}{2} + \varepsilon)} \int_{0}^{\frac{\sqrt{t}}{k_+} }  \frac{\rme^{-2\, C\, \xi^2} }{ \xi^{1-2m-2\varepsilon}}\,   \, \mathrm{d}\xi.
$$ 
As ${C}_m^\varepsilon(t) > 0$ is a strictly increasing, restricting ourselves to $t \geq 1$, \eqref{eq.optimal4} leads to
$$
\forall \; t\geq 1, \quad  \frac{{C}_m^\varepsilon(1)}{t^{m+\varepsilon}} \leq I_{\varepsilon}(t) \lesssim \|\bU_{\varepsilon}(t)\|^2_{\mathcal{H}}.
$$ 
Thus,  $\gamma_{m}^{HF}\leq m+\varepsilon$, for all $\varepsilon>0$ and this yields $\gamma_{m}^{HF}\leq m$. Thus, one concludes that $\gamma_{m}^{HF}=m$.\\[4pt]

\noindent {\bf Determination of  $\gamma_m^{\operatorname{HF}}$ in a  critical configuration}\\[12pt]
\noindent We assume now that the Maxwell system is in a critical configuration.
 \\[6pt]
 The proof is very similar to that of the non critical case and we shall only point out the differences.
First the inequality  \eqref{eq.LebesguesHmCR} shows that $\gamma_m^{\operatorname{HF}} \geq m/2$. \\ [12pt] 
To prove the reverse inequality, we proceed as for the non critical case.
The main difference lies in the choice of the initial data $\bU_{0,\varepsilon}$  which is now  chosen as in \eqref{eq.defUOkoptcr}, more precisely 
\begin{equation}\label{eq.defUOkoptcreps}
	\bbU_{0,\varepsilon}(\bk)= \mathcal{F}(\bU_0)(\bk)= \phi_\varepsilon(|\bk|) \; \mathcal{R}_{\bk}^* \, \frac{\mathcal{V}_{|\bk|}\big(\omega_{p}(|\bk|)\big)  \mathbf{e}_{1}}{|\mathcal{V}_{|\bk|}\big(\omega_{p}(|\bk|)\big)  \bf{e}_1|}, \quad \forall \, \bk\in \R^{3,*},
\end{equation}	
for some $p\in \calP_s$ such that $\omega_{p}(|\bk|)$ satisfies \eqref{eq.branchcirtique} (see again Theorem  \ref{HF-estm-opt}) and  the function $\phi_{\varepsilon}$ defined in  \eqref{eq.defUOkopteps}. In that case, we have to apply the inequality  \eqref{polynomial_decaycr-HF2}, instead of \eqref{polynomial_decayncr-HF2}, which leads to (instead of \eqref{eq.optimal3})
\begin{equation}\label{eq.optimal3cr}
	\rme^{-\frac{2\,C\, t}{ |\bk|^4}}  (1+|\bk|^2)^{-(\frac{3}{2}+m+\varepsilon)} \lesssim  |\bbU_{\varepsilon}(\bk,t)|^2,   \   \mbox{ for a.e. } \bk\in \bbR^3 \mid  \; |\bk|\geq k_+ \mbox{ and } \, \forall \; t\geq 0.
\end{equation}
the main difference with being \eqref{eq.optimal3} being that $|\bk|^{2}$ is replaced be $|\bk|^{4}$ in the exponential.\\ [12pt]
The rest of the calculations follows the same lines as in the non critical case modulo trivial adaptations : for instance the
change of variable  $\xi =\sqrt{t}/|\bk|$ is replaced by $\xi =\sqrt{t}/|\bk|^2$. The remaining details are left to the reader.

\subsubsection{Optimality of the  low frequency  polynomial decay rate}\label{sec-optim-LF}
We compute here the exponent $\gamma_p^{\operatorname{LF}}$ (defined  by \eqref{eq.gammaLF}) for $p\geq 0$ . The approach is similar to that  used in Section \ref{sec-optim-HF}  to compute $\gamma_m^{\operatorname{HF}}$.\\[12pt]
\noindent First, we underline that for any initial condition $\bU_0\in   \mathcal{H}^{p}_{\perp, \operatorname{LF}}$ (see \eqref{eq.HperpLF} for the definition of  $ \mathcal{H}^{p}_{\perp, \operatorname{LF}}$) which does not contain any  high frequency  Fourier components (since the support of  $\bbU_0$ is included in $\overline{  B(0, k_-)}$), we have the upper-bounds  \eqref{eq.globestim} and  \eqref{eq.finalestimateLF-2}, namely,
\begin{equation}\label{eq.finalestimateLF-3}
 \|\bU(t)\|^2_{\mathcal{H}}\lesssim \int_{|\bk|<k_-}  \rme^{-C_1  |\bk|^2 \, t } \,  |\bbU_0(\bk)|^2  \, \rmd\bk 
  \lesssim  \frac{C(p)}{ t^{p + \frac{3}{2}}}  \  \big\| \mathbf{U}_0 \big\|_{\boldsymbol{\cal L}_p^N}^2.
\end{equation}
Clearly, \eqref{eq.finalestimateLF-3} implies that $\gamma_p^{\operatorname{LF}}$   exists and satisfies  $\gamma_p^{\operatorname{LF}}\geq p+3/2$.
To show that $\gamma_p^{\operatorname{LF}}\leq p+3/2$, we construct an initial condition $\bU_{0}\in \mathcal{H}^{p}_{\perp, \operatorname{LF}}$ such that  
\begin{equation}\label{eq.finalestimateLF-4}
\|\bU(t)\|^2_{\mathcal{H}} \ \geq \frac{C(p)}{t^{p+\frac{3}{2}}},\, \quad \mbox{for some } C >0, \quad \forall t\geq1 .
\end{equation}
\noindent  To this aim,  according to Theorem \ref{LF-estm-opt}, we choose  $\bU_{0}$  in the form \eqref{eq.defUOkoptLF}, more precisely
\begin{equation}\label{eq.defUOkopteps-cr}
\left\{ 	\begin{array}{ll}
\ds	\bbU_{0}(\bk)=  \phi(|\bk|) \; \mathcal{R}_{\bk}^* \, \frac{\mathcal{V}_{|\bk|}\big(\omega_{0,1}(|\bk|)\big)  \mathbf{e}_{1}}{|\mathcal{V}_{|\bk|}\big(\omega_{0,1}(|\bk|)\big)  \bf{e}_1|}, 
\\ [18pt]
\ds \phi(|\bk|) = |\bk|^p\; \mbox{ if } |\bk| \leq k_-, \; \; = 0 \mbox{ otherwise.}
	\end{array} \right.
\end{equation}
Therefore, applying the inequality \eqref{polynomial_decay-LF} of Theorem  \ref{LF-estm-opt}, we get as $|\bbU_{0}(\bk)| = \phi(|\bk|) $:
\begin{equation}\label{eq.optimalLF}
 \rme^{-2\,C |\bk|^2\, t}  |\bk|^{2p} \lesssim  |\bbU(\bk,t)|^2,   \  \mbox{ for a.e. }   \bk\in \bbR^3 \mid  \; |\bk|\leq k_- \mbox{ and } \, \forall \; t\geq 0.
\end{equation}
By virtue of Plancherel identity:  $ \ds 
\|\bU(t)\|^2_{\mathcal{H}}=\int_{|\bk|<k_-}| \bbU(\bk,t )|^2 \, \mathrm{d}\bk
$,  it follows from \eqref{eq.optimalLF} that
\begin{equation}\label{eq.optimalLF2}
\int_{|\bk|<k-}\rme^{-2\,C |\bk|^2\, t}  \, |\bk|^{2p} \rmd \bk = 4\pi \, \int_{0}^{k_-} \rme^{-2\,C |\bk|^2\, t}  \, |\bk|^{2p}\, |\bk|^2 \rmd |\bk| \lesssim  \|\bU(t)\|_{\mathcal{H}}^2.
\end{equation}
Hence, using the change of variable $\mathbf{\xi }= \sqrt{2 \,C\, t} \; |\bk|$ , one obtains
$$
\frac{1}{(2\, C\, t)^{p + \frac{3}{2}}} \; \int_{0}^{ \sqrt{t} \, k_-}\xi^{2p+2}\, e^{- \xi^2} \; \rmd \mathbf{\xi } \lesssim   \|\bU(t)\|_{\mathcal{H}}^2.
$$
It leads to \eqref{eq.finalestimateLF-4} and more precisely to
$$
\frac{C(k_-, p)}{t^{p+\frac{3}{2}}} \lesssim   \|\bU(t)\|_{\mathcal{H}}^2  \quad \mbox{ with } \ C(k_-, p)=\frac{1}{(2\, C)^{p + \frac{3}{2}}} \int_{0}^{k_-}\xi^{2p+2}\, e^{- \xi^2} \; \rmd \mathbf{\xi },  \quad \forall \; t\geq 1.
$$
Thus, $\gamma_{p}^{LF}\leq p+3/2$ and one concludes that $\gamma_{p}^{LF}=p+3/2$.\\[4pt]

\subsection{Extensions of  Theorems \ref{thm_Lorentz} and  \ref{thm-optm}}\label{sec-extension-results}

\subsubsection{The case of a  non dispersive electric  or  magnetic media}\label{sec-extension1}
For some passive dispersive materials  (e.g.  in metals, see  for instance \cite{Mai-07}), the magnetization of the material is weak so that for the corresponding physical model, the magnetic permeability $\mu$ is assumed to be  a constant  function of the frequency: $\mu(\omega)=\mu_0$ whereas  the permittivity $\varepsilon(\omega)$ is given by a generalized Lorentz models.  This situation corresponds to the case $N_m=0$ and $N_e\geq 1$ in  formula \eqref{eq.permmitivity-permeabiity}.
In this setting, we would like to mention that our results sill hold  and   to point out  the simplifications  in their proof.\\[12pt]
\noindent Assume that $N_m=0$ and $N_e\geq 1$   in the evolution problem \eqref{eq.schro}.  It implies in particular that  $\bU=(\bE,\bH, \bP, \dot{\bP})$. 
Adapting the definition \ref{Critical_cases} to this particular case, the Maxwell's system is in a critical configuration if and only if one coefficient  $\alpha_{e,j}$ vanishes (which imposes that $N_e\geq 2$ since by \eqref{WD}  at least one coefficient $\alpha_{e,j'}$ has to be positive). Thus, it is in a non-critical configuration if  only if all the coefficients $\alpha_{e,j}$ are positive.
\\[12pt] 
 \noindent Under the same assumptions, the Theorems \ref{thm_Lorentz} and  \ref{thm-optm} still hold. The proof of these results remain the same as that developed in Sections \ref{sec-modal-analysis} to \ref{proof-general-cases} with the following two simplifications: 
 \begin{enumerate}
\item Concerning the  Section \ref{sec.HF},   $\mathcal{P}_s=\mathcal{P}_e\cap \bbR\neq \varnothing$  if and only if one coefficient  $\alpha_{e,j}$ vanishes which is precisely  the definition of  a critical configuration.
 Thus,  in  the Section \ref{estiUs} for the  estimate of the term $\bbU_s(\bk, t)$, one just needs to consider the critical configuration in Corollary \ref{CoroIm} and Lemma \ref{LemEstis}.  \\ [8pt]
As $\mathcal{P}_d=\varnothing$, $\bbU_d(\bk, t)=0$ in \eqref{decompU} and one  can skip  the Section \ref{estiUd} devoted to the estimate of this term.   
 \item In the Section \ref{sec.LF}, since by \eqref{WD}  at least one coefficient  $\alpha_{e,j}$ is positive, we deduce from  \eqref{eq.positvity} and \eqref{eq.consequenceWD}, that the set  $\mathcal{Z}_s$ of real zeros (see \eqref{Zs}) is empty : $\mathcal{Z}_s=\mathcal{Z}_e\cap \bbR= \varnothing$.  Hence, $\bbU_{z,s}(\bk, t)=0$ in \eqref{decompUbis} and one can skip  the Section \ref{estiUZS} devoted to the estimate of of this term.   
 \end{enumerate}
  \noindent   We finally point out all our results still hold in the ``symmetric case''  (more rare in physical models) where $N_e=0$ and $N_m\geq 1$ with obvious changes  compared to the above discussion.

\subsubsection{The case of bounded domains}
In \cite{cas-jol-ros-22} (Section 4.1), we show how to generalize our Lyapunov estimates technique to the case of bounded propagation domains  under the strong dissipation assumption \eqref{SD}.
Similarly, we show here how to extend our modal approach to bounded electromagnetic cavities but under the weak dissipation assumption  \eqref{WD}. \\[6pt]
\noindent One considers here  the  Maxwell's evolution system (\ref{planteamiento Lorentz},\ref{CI})  in a bounded Lipschitz-continuous domain $\Omega\subset \R^3$  completed with perfectly conducting  boundary conditions
\begin{equation} \label{PEC}
\bE \times {\bf n}=0,  \quad  \mbox{ where $\bf n$ is the outward normal vector  to $\partial \Omega$.}
\end{equation}
\noindent We introduce  the functional spaces (see  e.g. \cite{Dau-1990})
$$
\begin{array}{lr}
H(\mbox{rot} ;\Omega)=\{ {\bf u }\in \bL^2(\Omega) \mid \nabla \times {\bf u} \in \bL^2(\Omega)  \},\ \hspace*{0.2cm}  H_0(\mbox{rot} ;\Omega)=\{ {\bf u}  \in H(\mbox{rot} ;\Omega)\mid {\bf u} \times {\bf n }=0 \mbox{ on }\partial \Omega \}, \\[6pt]
H(\mbox{div} \,0 ;\Omega)=\{ {\bf u }\in \bL^2(\Omega) \mid  \nabla\cdot  {\bf u}=0  \},  \hspace*{0.7cm} H_0(\mbox{div} \,0 ;\Omega)=\{ {\bf u }\in H(\mbox{div} \,0 ;\Omega) \mid   {\bf u}\cdot {\bf n}=0  \mbox{ on } \partial \Omega \}.
\end{array}
$$
In this setting,  one can rewrite the Maxwell's system as the evolution equation \eqref{eq.schro} with the Hilbert space $\mathcal{H}=\bL^2(\Omega)^N$ (where $\bL^2(\Omega)= L^2(\Omega)^3 $) and a maximal dissipative unbounded operator $\bbA: D(\bbA)\subset \mathcal{H}\to \mathcal{H}$ defined by the relation \eqref{eq.defHamil} but with domain $$D(\bbA)=H_0(\mbox{rot};\Omega) \times H(\mbox{rot};\Omega) \times \bL^2(\Omega)^{N_e} \times \bL^2(\Omega)^{N_e} \times \bL^2(\Omega)^{N_m}  \times \bL^2(\Omega)^{N_m}.$$
Similarly, one introduces a reduced evolution equation on the Hilbert space
$$
\mathcal{H}_{\perp}=H(\mbox{div} \,0 ;\Omega)\times H_0(\mbox{div} \,0 ;\Omega)\times H(\mbox{div} \,0 ;\Omega)^{N_e}\times  H(\mbox{div} \,0 ;\Omega)^{N_e} \times H_0(\mbox{div} \,0 ;\Omega)^{N_m}  \times H_0(\mbox{div} \,0 ;\Omega)^{N_m} 
$$
with a reduced  maximal dissipative operator $\bbA_{\perp}:D(\bbA_{\perp})\subset \mathcal{H}_{\perp} \to  \mathcal{H}_{\perp} $   defined on its domain
$D(\bbA_{\perp})=\mathcal{H}_{\perp}\cap D(\bbA)$ (see (\ref{eq.stabAperp}, \ref{eq.reducA})). We point out that $\operatorname{ker}(\bbA_{\perp})$ has a finite dimension.  Moreover, $\bbA_{\perp}$  is  injective (or equivalently $\operatorname{ker}(\bbA_{\perp})^{\perp}=  \mathcal{H}_{\perp}$) if  $\Omega$ is simply connected. 
\\[4pt]
\noindent As in \cite{cas-jol-ros-22}, the main difference in the analysis, with respect to the case of $\R^3$, is  that the use of the Fourier transform in 
space has to be replaced by an adequate modal expansion. More precisely, we introduce the closed subspace of 
$\bL^2(\Omega)\times \bL^2(\Omega)$ given by $$\mathcal{H}_0 := H(\mbox{div} \,0 ;\Omega)\times H_0(\mbox{div} \,0 ;\Omega).$$
Then (see e.g. \cite{Dau-1990} chapter IX),  the  operator $\mathcal{M}: D(\mathcal{M})\subset \mathcal{H}_0\mapsto \mathcal{H}_0  $, defined by, 
\begin{equation*} \label{op_Maxwell}
\mathcal{M} \begin{pmatrix} \bf u \\ \bf v \end{pmatrix}  = \mathrm{i} \; \begin{pmatrix} \nabla \times {\bf v} \\ - \nabla \times \bf u \end{pmatrix} , \quad \forall \;  ({\bf u}, {\bf v})  \in D(\mathcal{M}) := \mathcal{H}_0\; \cap  \; 
{\big(H_0(\mbox{rot};\Omega)\times H(\mbox{rot} ;\Omega) \big)},
\end{equation*}
is self-adjoint.
One proves, see e.g.  \cite{Dau-1990} chapter IX, that, as  $ D(\mathcal{M})$ is compactly embedded in  $\bL^2(\Omega)\times \bL^2(\Omega)$,  $\mathcal{M}$ has a compact resolvent and a finite dimensional kernel (which reduces to $\{0 \}$ if  $\Omega$ is simply connected).
From the theory of self-adjoint operators with compact resolvent, and the symmetries of Maxwell's equations, which imply that the spectrum of $\mathcal{M}$ is symmetric with respect to the origin,  one can construct a countable infinity of normalized modes indexed by $ p \in \N^*$:
\begin{equation*} \label{cavity modes2}
\big( \pm \, k_p, ({\bf u}_p^\pm,  {\bf v}_p^\pm)\big) \in \R \times  D(\mathcal{M}) \ \mid \ \mathcal{M}  \begin{pmatrix} \bf u_p^\pm \\ \bf v_p^\pm \end{pmatrix} =\pm\, k_p\begin{pmatrix} \bf u_p^\pm \\ \bf v_p^\pm,  \end{pmatrix}  \mbox{ with } \ k_p > 0 
\end{equation*} 
where $(k_p)$ is an increasing sequences satisfying  $k_p \rightarrow + \infty$ (as $p\to \infty$).
These cavity modes form an orthonormal basis of $\operatorname{ker}(\mathcal{M})^{\perp}$.
Using this basis, one decomposes the components of the solution $\bU=(\bE, \bH, \bP, \dot{\bP}, \bM,\dot{\bM} )$  of  the  equation  \eqref{eq.schro} with an initial condition $\bU_0\in \mathcal{H}_{\perp}$ as follows:
$$
{\bf E}(\cdot,t) = \sum_\pm \sum_{p=1}^{+\infty}  \; \mathbb{E}(\pm k_p; t) \;  {\bf u}_p^\pm, \quad  {\bf H}(\cdot,t) = \sum_\pm \sum_{p=1}^{+\infty}  \; \mathbb{H}(\pm k_p; t) \;  {\bf v}_p^\pm,
$$
and the auxiliary field ${\bf P}_j$  and ${\bf M}_{\ell}$   accordingly 
$$
{\bf P}_j(\cdot,t) = \sum_\pm \sum_{p=1}^{+\infty}  \; \mathbb{P}_{j}(\pm k_p; t) \;  {\bf u}_p^\pm, 
\ {\bf M}_{\ell}(\cdot,t) = \sum_\pm \sum_{p=1}^{+\infty}  \; \mathbb{M}_{\ell}(\pm k_p; t) \; {\bf v}_p^\pm.
$$
Injecting these modal  expansions in the equation \eqref{eq.schro} yields  that  $$\bbU(\pm k_p;t)=\big( \mathbb{E}(\pm k_p; t), \mathbb{H}(\pm k_p; t),  \mathbb{P}(\pm k_p; t ), \dot{\mathbb{P}}(\pm k_p; t ), \mathbb{M}_\ell(\pm k_p; t ),\dot{\mathbb{M}}_\ell(\pm k_p; t )\big) \in \bbC^N$$
satisfies the following system of ODE's indexed by $\pm k_p$ (which plays the same role here as  the Fourier parameter $\bk$):
\begin{equation*}\label{eq.cauchypbkp}
 \frac{ \rmd \bbU(\pm k_p,t)}{\rmd  t} + \rmi \, \bbA_{\perp, \pm k_p} \bbU(\pm k_p,t)=0 \quad  \mbox{ for } t\geq 0    \ \mbox{ with }\
\bbU( \pm k_p,0)=\bbU_0(\pm k_p).
\end{equation*}
where   $\bbA_{\perp, \pm k_p}\in \mathcal{L}(\bbC^{N})$ is defined  for all  $\bbU =(\bbE, \bbH, \bbP, \dot{\bbP},\bbM , \dot{\bbM})\in
 \bbC^N$ as $\bbA_{\bk}$ in \eqref{eq.defAk} but by substituting in the two first lines $-(\bk \times \bbH)/\varepsilon_0$ and $(\bk \times \bbE)/ \mu_0$ by $\pm \, k_p  \, \bbH /\mu_0$ and $\pm \, k_p \,  \bbE/\varepsilon_0$.\\[12pt]
One also remarks  on  \eqref{eq.defAk} that $\bbA_{\perp,  k_p}$ and $ \bbA_{\perp, - k_p}$ are unitarily equivalent via the unitary map $$\mathcal{R}: (\bbE, \bbH, \bbP, \dot{\bbP},  \bbM, \dot{\bbM}) \in \bbC^{N}\longmapsto (\bbE,- \bbH, \bbP, \dot{\bbP}, - \bbM,- \dot{\bbM}) \in \bbC^{N} .$$
Thus, one only has to estimate  the solution $\rme^{-\rmi \bbA_{\perp,  k_p}t} \, \bbU_0(k_p)$ of \eqref{eq.cauchypbkp} for $k_p>0$. 
A computation similar to that in the Appendix \ref{sec-app-spec2} shows the spectrum $\sigma(\bbA_{k_p})$ is given by the set of solutions  of the dispersion relation 
\begin{equation*}
\omega^2 \varepsilon(\omega) \mu(\omega)=k_p^2.
\end{equation*}

\noindent As the set $\{  k_p, \, p \in \mathbb{N}^* \}$ is discrete and contains only non-zero elements, one does not need to derive low  frequencies estimates  (as in Section \ref{sec.LF})
and the exponential  uniform decay of the mid-range frequencies is trivial since there are only finitely many of them.  
Thus,  one only  estimates the high frequencies components $\bbU(\pm k_p, t)$  for $k_p\gg 1$. To this aim, one uses the results of Sections \ref{sec-modal-analysis} and  \ref{sec.HF} by replacing $|\bk|$ by $k_p$.\\[12pt]
\noindent 
Let $m$ be a positive integer. One defines the Hilbert space $D(\bbA^m)$ by  $\{ \bU \in \mathcal{H} \mid \bbA^k \bU \in \mathcal{H} \mbox{ for } k=1, \ldots, m \}$  endowed with its usual norm $\|\cdot\|_{D(\bbA^m)}$ defined by 
\begin{equation*}
\|\bU\|_{D(\bbA)^m}^2=\sum_{k=0}^m\|\bbA^k \bU\|^2_{\mathcal{H}}, \   \forall \; \bU \in D(\bbA^m).
\end{equation*}
Then, the main result is as follows. 
\begin{Thm} \label{thm_Lorentz_bounded}
Let assume  that  the Maxwell's system  (\ref{planteamiento Lorentz},\ref{CI},\ref{PEC})  satisfies  the weak dissipation assumption \eqref{WD} and the irreducibility assumptions $(\mathrm{H}_1)$ and $(\mathrm{H}_2)$. Then, for any initial condition $\bU_0\in  \mathcal{H}_{\perp}$ of this system, one has 
	\begin{equation} \label{convergenceLbouned}
	\lim_{t \rightarrow + \infty} \|\bU(t)\|^2_{\mathcal{H}}= 0.
	\end{equation}
In addition,  if the initial data  satisfies 
$$\bU_0 \in  D(\bbA^m) \cap  \operatorname{ker}(\bbA_\perp)^{\perp},  \quad \mbox{for  any integer }m> 0 \quad \mbox{(see remark \ref{rem_reg})} , $$
then, in the non-critical (resp. critical) case , see definition \ref{Critical_cases}, one has, for any $t>0$,
		\begin{equation} \label{polynomial_decay-bounded}
			\|\bU(t)\|^2_{\mathcal{H}} \lesssim \frac{\|\bU_0\|^2_{D(\bbA^m)}}{t^m} ,  \quad  \quad \Big( \mbox{resp. } \|\bU(t)\|^2_{\mathcal{H}} \lesssim \frac{\|\bU_0\|^2_{D(\bbA^m)}}{t^\frac{m}{2}} \;  \Big) .
		\end{equation}
		Moreover, the decay estimate \eqref{polynomial_decay-bounded} are  optimal  since in  the non-critical case 
	\begin{equation*} \label{polynomial_decay-opt-bounded}
		m=\sup \{\gamma \in \bbR^+ \mid   \forall \; \bU_0\in D(\bbA^m) \cap  \operatorname{ker}(\bbA_\perp)^{\perp} ,\, \exists \; C(\bU_0)>0 \mid \|\bU(t)\|^2_{\mathcal{H}} \ \lesssim \frac{C(\bU_0)}{t^{\gamma}}, \ \forall t\geq1 \}
\end{equation*}
and  in the critical case,   the above equality holds  with $m$ replaced by $m/2$ in the left hand side. 		
\end{Thm}

\begin{Rem} \label{rem_reg}
For $C^{\infty}$ domains $\Omega$, the condition $\bU_0 \in  D(\bbA^m) \cap \operatorname{ker}(\bbA_{\perp})^{\perp}$ (see e.g. \cite{Dau-1990}, chapter IX) is equivalent to $\bU_0\in   \bH^m(\Omega)^N  \cap \operatorname{ker}(\bbA_{\perp})^{\perp}$  for any integer $m>0$ and the  norms $\|\cdot\|_{D(\bbA^m)}$ and $\|\cdot \|_{ \bH^m(\Omega)^N}$ on these spaces are equivalent.
\end{Rem}

\appendix
\renewcommand{\thesection}{\Alph{section}} 
\makeatletter
\renewcommand\@seccntformat[1]{\appendixname\ \csname the#1\endcsname.\hspace{0.5em}}
\makeatother
\section{}\label{sec-appendix}
\subsection{Spectrum   and resolvent of $\bbA_{|\bk|,\perp}$}\label{sec-spectrum}
\subsubsection{Proof of Proposition \ref{Prop.res}:  resolvent of $\bbA_{|\bk|, \perp}$} \label{sec-expression-resolv}  
Given $\bbF=(\be, \bh, \bp, \dot{\bp}, \bm,\dot{\bm})^t\in \bC_{\perp}^N$ and $\omega\in \bbC\setminus (S(|\bk|) \cap {\cal S}_{\cal T}) $, we look at the problem
\begin{equation}\label{eq.res}
\mbox{Find } \bbU = (\bbE, \bbH, \bbP, \dot \bbP, \bbM, \dot \bbP ) \in \bC_{\perp}^N \mbox{ such that} \quad (\bbA_{|\bk|,\perp}- \omega \mathrm{I}) \, \bbU= \bbF, 
\end{equation}
which is equivalent to finding $\bbU = (\bbE, \bbH, \bbP, \dot \bbP, \bbM, \dot \bbP ) \in \bC_{\perp}^N $ solution of the system \\
\begin{equation}\label{systemEH} 
	\hspace*{0.1 cm} 	\left\{  \begin{array}{lll}
		\displaystyle -\frac{|\bk| \, {\bf e_3} \times \bbH}{\varepsilon_0}- \rmi \,  \sum \Omega_{e,j}^2 \,\dot{\bbP}_j - \omega \,  \bbE= \be,  & \quad \quad (i)  \\[12pt] 
		\displaystyle \frac{|\bk| \, {\bf e_3} \times \bbE}{\mu_0}- \rmi \, \sum \Omega_{m,\ell}^2\, \dot{\bbM}_\ell-  \omega  \, \bbH= \bh,  & \quad \quad (ii)
	\end{array} \right.
\end{equation}	
\vspace*{0.1cm} 
\begin{equation}\label{systemP} 
\hspace*{0.1cm} \left\{  \begin{array}{lll}
	\; 	\rmi \,  \dot {\bbP}_j- \omega \, \bbP_j= \bp_j,   & \quad \; \; \;  \; \; (i) \\[12pt]
		- \rmi \, \alpha_{e,j} \, \dot{\bbP}_j- \rmi  \, \omega_{e,j}^2 \, \bbP_j+\rmi \, \bbE - \omega \, \dot{\bbP}_j= \dot{\bp}_j,    & \quad \; \;  \; \; (ii)
\end{array} \right. 
\end{equation}	
\vspace*{0.3cm} 
\begin{equation} \label{systemM}
	\hspace*{0.3 cm}  \left\{  \begin{array}{lll}
\; \rmi \, \dot{ \bbM}_\ell - \omega \bbM_\ell= \bm_\ell , & \quad (i) \\[12pt]
- \rmi \, \alpha_{m,l} \, \dot{\bbM}_{\ell} - \rmi \, \omega_{m,\ell}^2 \, \bbM_{\ell}+\rmi \, \bbH- \omega \, \dot{\bbM}_{\ell}=\dot{\bm}_{\ell} .   & \quad (ii)
\end{array} \right.
\end{equation}	
~\\
We want to show that this problem admits a unique solution $\bbU$ defining $\bbU = R_{|\bk|}(\omega) \, \bbF$. The proof of Proposition \ref{Prop.res} is thus purely computational and reduces to solving explicitly the system (\ref{systemEH}, \ref{systemP}, \ref{systemM}). We provide some details for the ease of the reader. \\[12pt]
Substituting the expression of $\dot \bbP_{j}$ and $\dot \bbM_{\ell}$ given by \eqref{systemP}(i) and \eqref{systemM}(i), i.e.  
\begin{equation}\label{PPdot}
	\dot {\bbP}_j=-\rmi \, \omega \, \bbP_j -\rmi \, \bp_j, \quad \dot {\bbM}_\ell =-\rmi \, \omega  \, \bbM_\ell -\rmi \, \bm_{\ell}
	\end{equation} 
into
\eqref{systemP}(ii) and \eqref{systemM}(ii), we get, 
by definition of $q_{e,j}(\omega)$ and $q_{m,\ell}(\omega)$ and by definition \eqref{operatorsApAm} of the operators $\bbA_{e,j}(\omega), \bbA_{m,\ell}(\omega)$ (the division by  $q_{e,j}(\omega)$ and $q_{m,\ell}(\omega)$ is allowed because $\omega \notin {\cal P}),$
\begin{equation}\label{eq.defP-M}
	\bbP_{j}= - \frac{\bbE}{q_{e,j}(\omega)}+\bbA_{e,j}(\omega) \, \bbF  \quad\mbox{ and } \quad  \bbM_{ \ell}= - \frac{\bbH}{q_{m,\ell}(\omega)}+\bbA_{m,\ell}(\omega) \, \bbF .
\end{equation}
Thus, going back to $\dot \bbP_{j}$ and $\dot \bbM_{\ell}$, via  \eqref{systemP}(i) and \eqref{systemM}(i), gives by definition \eqref{operatorsApAm} of $\dot \bbA_{e,j}(\omega), \dot \bbA_{m,\ell}(\omega)$:
\begin{equation}\label{eq.defdotP-M}
	\dot \bbP_{j}=  \frac{ \rmi \, \omega  \,\bbE}{q_{e,j}(\omega)} + \dot \bbA_{e,j}(\omega) \, \bbF  \quad\mbox{ and } \quad \dot  \bbM_{ \ell}= \frac{ \rmi \, \omega  \,\bbH}{q_{m,\ell}(\omega)} + \dot \bbA_{m,\ell}(\omega) \, \bbF .
\end{equation}
Substituting  $\dot {\bbP}_{j}$, from \eqref{eq.defdotP-M},  into \eqref{systemEH}(ii) we get, by definitions (\ref{eq.permmitivity-permeabiity}, \ref{operatorsApAm}) of $\mu(\omega)$ and  $\bbA_{m}(\omega)$:
\begin{equation} \label{Eq1}
|\bk| \; {\bf e_3 }\times \bbE - \omega \mu(\omega) \, \bbH= \mu_0\, \Big( \bh+ \rmi  \sum\Omega_{m,\ell}^2 \; \dot \bbA_{m,\ell}(\omega) \, \bbF \Big) =  - \, \bbA_{m}(\omega) \, \bbF .
\end{equation}
Similarly, from  $\dot {\bbM}_{\ell}$ in \eqref{eq.defdotP-M}, we obtain, by definitions (\ref{eq.permmitivity-permeabiity}, \ref{operatorsApAm}) of $\varepsilon(\omega)$ and  $\bbA_{e}(\omega)$,  
\begin{equation} \label{Eq2}
-|\bk| \; {\bf e_3} \times \bbH - \omega \, \varepsilon (\omega) \, \bbE=\varepsilon_0 \Big(\, \be+ \rmi\sum\Omega_{e,j}^2  \; \dot{\bbA}_{e,j} (\omega) \, \bbF \Big) =  - \, \bbA_{e}(\omega) \, \bbF .
\end{equation}
In order to eliminate $\bbH$ between \eqref{Eq1} and  \eqref{Eq2}, we perform the combination 
$$|\bk| \, {\bf e_3 }\times\eqref{Eq1} - \omega \mu(\omega) \eqref{Eq2}\, .$$
This gives, since $-|\bk| \, {\bf e_3}\times (|\bk| \, {\bf e_3}  \times \bbE)=  |\bk|^2 \,\bbE$, (as  $\bbU \in \bC_{\perp}^N$, ${\bf e_3} $ and $\bbE$ are orthogonal)
	\begin{equation}\label{eq.defFe}
	\big(\mathcal{D}(\omega)-|\bk|^2\big) \, \bbE= \omega \mu(\omega ) \, \bbA_e(\omega)\bbF - |\bk| \, {\bf e_3} \times \bbA_m(\omega) \bbF.
\end{equation}
	As $\omega\notin  S(|\bk|)$, one has $\mathcal{D}(\omega)-|\bk|^2\neq 0$ and therefore, by definition \eqref{defS} of $\mathcal{S}_{|\bk|}(\omega)$, 
\begin{equation}\label{eq.e3}
	\bbE=\frac{ -|\bk| \, {\bf e_3} \times \bbA_m(\omega) \, \bbF +  \omega \mu(\omega ) \, \bbA_e(\omega) \, \bbF}{\mathcal{D}(\omega)-|\bk|^2} = \mathcal{S}_{|\bk|}(\omega) \, \bbF.
\end{equation}
Thus, since $\omega \mu(\omega)\neq 0$ (as $\omega\notin \mathcal{Z}_m\cup \{ 0\}$),  from \eqref{Eq1} and by definition \eqref{operatorsAeAh} of $\bbA_h(\omega)$:
	\begin{equation}\label{eq.h4}
		\bbH=\frac{1}{\omega \mu(\omega)}  \Big(|\bk| \; {\bf e_3}  \times  \mathcal{S}_{|\bk|}(\omega) \, \bbF+ \bbA_m(\omega) \,  \bbF \Big).
	\end{equation}
Finally, substituting (\ref{eq.defFe}, \ref{eq.h4}) in (\ref{eq.defP-M}, \ref{eq.defdotP-M}) we get  the expression of the other 
	\begin{equation}\label{eq.pm}
		\left\{ \begin{array}{lll}
\ds 	\bbP_{j}= - \frac{\mathcal{S}_{|\bk|}(\omega) \, \bbF}{q_{e,j}(\omega)}+\bbA_{e,j}(\omega) \, \bbF, \quad  & 	\ds \bbM_{\ell} =  \bbA_{m,\ell}(\omega) \, \bbF - \frac{\bbA_m(\omega) \, \bbF + |\bk| \, {\bf e_3}  \times \mathcal{S}_{|\bk|}(\omega) \, \bbF}{ \omega \mu(\omega) q_{m,\ell}(\omega)}, \\ [18pt]
	\ds		\dot \bbP_{j}=  \frac{ \rmi \, \omega  \,\mathcal{S}_{|\bk|}(\omega) \, \bbF}{q_{e,j}(\omega)} + \dot \bbA_{e,j}(\omega) \, \bbF,   \quad &  \ds\dot {\bbM}_{\ell}= \dot  \bbA_{m,\ell}(\omega) \, \bbF + \rmi \;   \frac{\bbA_m(\omega) \, \bbF + |\bk| \, {\bf e_3}  \times \mathcal{S}_{|\bk|}(\omega) \, \bbF}{\mu(\omega) q_{m,\ell}(\omega)} .
	\end{array} 
	\right.
\end{equation}
The reader will then easily verify that the formulas (\ref{eq.defFe}, \ref{eq.h4}, \ref{eq.pm}) lead to the expression (\ref{eq.expressresolv}, \ref{defS}, \ref{defS}, \ref{defT}) given in Proposition \ref{Prop.res} for the resolvent $R_{|\bk|}(\omega)$.
\subsubsection{The Spectrum of $\bbA_{|\bk|,\perp}$ : proof of Proposition \ref{Prop.spec}}\label{sec-app-spec2}
From Proposition \ref{Prop.res}, we already know that
$$
\sigma(\bbA_{|\bk|,\perp}) \subset S(|\bk|) \cup {\cal S}_{\cal T}.
$$
To prove the Proposition   \ref{Prop.spec}, it suffices to show that
$$
\begin{array}{ll} 
{\bf  (i)} &  \quad \sigma(\bbA_{|\bk|,\perp}) \cap {\cal S}_{\cal T} = \varnothing. \\ [12pt]
{\bf  (ii)} & 	\quad \mbox{$S(|\bk|) \subset \sigma(\bbA_{|\bk|,\perp}) $ and for any $\omega \in S(|\bk|) $, $\mbox{dim } \mbox{ker } \big( \bbA_{|\bk|,\perp} - \omega \, \mathrm{I}d   \big) = 2$.}
\end{array}
$$
{\bf Preambulus}. For both steps, we shall use the fact that $\bbU = \big(\bbE, \bbH, (\bbP_j), (\dot \bbP_j), (\bbM_\ell), (\dot \bbP_\ell)	\big)\in \mbox{ker } \big( \bbA_{|\bk|,\perp} - \omega \, \mathrm{I}d  \big)$ means that $\bbU$ is a solution of
(\ref{systemEH}, \ref{systemP}, \ref{systemM})  with $\bbF=(\be, \bh, \bp, \dot{\bp}, \bm,\dot{\bm}) = 0$. Therefore, proceeding as for obtaining \eqref{eq.defP-M} and  \eqref{eq.defdotP-M}, before division by  $q_{e,j}(\omega)$ and $q_{m,\ell}(\omega)$, we deduce from  (\ref{systemP}, \ref{systemM}) that 
\begin{equation}\label{eq.PM}
q_{e,j}(\omega)\, 	\bbP_{j}= -\bbE , \quad  q_{m,\ell}(\omega) \, \bbM_{ \ell}= -\bbH , \quad q_{e,j}(\omega)\, 	\dot \bbP_{j}=  \rmi \, \omega \, \bbE , \quad  q_{m,\ell}(\omega) \, \dot \bbM_{ \ell}=  \rmi \, \omega \, \bbH .
\end{equation}
\noindent {\bf Proof of Step (i)}.  Assume that $\omega \in {\cal S}_{\cal T} = {\cal P}_e \cup {\cal P}_m  \cup {\cal Z}_m \cup \{0\}$.\\ [12pt]
(a) Let $\omega \in {\cal P}_e$ and $\bbU \in \mbox{ker } \big( \bbA_{|\bk|,\perp} - \omega \, \mathrm{I}d   \big)$, for some $j_0$, $q_{e,j_0}(\omega)= 0$, thus, by \eqref{eq.PM}, $\bbE= 0$. 
Thus, by \eqref{eq.PM} and assumption $(\mathrm{H}_1)$  that yields $q_{e,j}(\omega)\neq  0$ jor $j \neq j_0$, we conclude that  $\bbP_{j}= 0$ and $\dot \bbP_{j}= 0$ for any $j \neq j_0$.  \\ [12pt]
Next, according to assumption $(\mathrm{H}_1)$ again,
\begin{itemize} 
	\item Either $q_{m,\ell}(\omega) \neq 0$ for any $m$, we conclude from \eqref{eq.PM}  $\bbM = 0, \dot \bbM = 0$. 
	 Thus, going back to \eqref{systemEH}(ii) with ${\bf h} = 0$, we  have $\bbH=0$ (since $\omega\neq 0$).
	Then, from \eqref{systemEH}(i) with ${\bf e} = 0$, we have $\dot \bbP_{j_0}= 0$. Thus, one gets $\bbP_{j_0}= 0$ since, by \eqref{PPdot} for $ {\bf p} = 0$, $- \rmi \, \omega \,  \bbP_{j_0}= \dot \bbP_{j_0}$ and $\omega \neq  0$ (as $0\notin \calP_e$).

	\item  Or there exists a unique $\ell_0$ for which $q_{m,\ell}(\omega) = 0$. In that case \eqref{eq.PM}  implies that $\bbH= 0$ and $\bbM_\ell=\dot \bbM_{\ell}= 0$ for any $\ell \neq \ell_0$.
	 Finally, going back to \eqref{systemEH} (i) and (ii) with ${\bf e}={\bf h} = 0$, we have $\dot\bbP_{j_0} =\dot \bbM_{\ell_0}= 0$,  thus $\bbP_{\ell_0}= \bbM_{\ell_0}= 0$ since \eqref{PPdot} for ${\bf m} =0$ and $  {\bf p} = 0$ gives  $- \rmi \, \omega \,  \bbP_{j_0}= \dot \bbP_{j_0}$ and  $- \rmi \, \omega \,  \bbM_{\ell_0}= \dot \bbM_{\ell_0}$ with $\omega \neq  0$.
	\end{itemize}
In all cases, $\bbU = 0$, thus $\omega \notin \sigma(\bbA_{|\bk|,\perp}).$ \\ [12pt]
(b) In a symmetric manner, we prove that if $\omega \in {\cal P}_m$, $\omega \notin \sigma(\bbA_{|\bk|,\perp}).$ \\ [12pt]
\noindent Now, for the rest of the proof, we point out that for $\omega \notin \mathcal{P}$, one shows by proceeding exactly as for obtaining \eqref{Eq1} and \eqref{Eq2} that
\begin{equation} \label{EqEH}
(1) \quad 	|\bk| \; {\bf e_3 }\times \bbE - \omega \mu(\omega) \, \bbH=0 , \quad  (2) \quad 	-|\bk| \; {\bf e_3} \times \bbH - \omega \, \varepsilon (\omega) \, \bbE= 0.
\end{equation}
More precisely, one checks that as soon as $\omega\notin \mathcal{P}$,  (\ref{eq.PM}, \ref{EqEH})  is equivalent to (\ref{systemEH}, \ref{systemP}, \ref{systemM}) for $\bbF = 0$.\\[4pt]

\noindent (c)  It remains to look at $\omega \in {\cal Z}_m \cup \{0\}$, i.e. $\omega \,  \mu (\omega) = 0$.
By $(\mathrm{H}_2)$, one has $\mathcal{P}\cap ( {\cal Z}_m \cup \{0\})=\varnothing$. Thus, from \eqref{EqEH}(1) and since $\bbE\cdot {\bf e}_3 = 0$, we deduce that $\bbE = 0$ and it follows  from  \eqref{EqEH}(2) that $\bbH = 0$ (as $\bbH\cdot {\bf e}_3 = 0$). Then, going back to \eqref{eq.PM}, since $\omega \notin {\cal P}$, we deduce that $\bbP = \dot \bbP = 0 \mbox{ and }  \bbM = \dot \bbM = 0$, thus $\bbU = 0$. This proves that $\omega \notin \sigma(\bbA_{|\bk|,\perp}).$ \\ [12pt]
{\bf Proof of Step (ii)}. 
Let $\omega \in S(|\bk|)$ and $\bbU \in \operatorname{ker } \big( \bbA_{|\bk|,\perp} - \omega \, \mathrm{I}d   \big)$. First note that, since  by assumption $\bk\neq 0$ and thus  $\omega \mu(\omega) \neq 0$ (since  $\omega$ satifies the dispersion relation \eqref{eq.spec}), according to \eqref{EqEH}(a), we have
\begin{equation} \label{EV1}
\bbH=\frac{|\bk|\, {\bf e_3}\times \bbE}{\omega \mu(\omega)}, \quad
\end{equation}
so that, it yields with \eqref{eq.PM}:
\begin{equation}  \label{EV2}
 \bbP_{j}= - \frac{\bbE}{q_{e,j}(\omega)}, \quad
 \dot {\bbP}_{j}=\frac{\rmi \,\omega \,\bbE}{q_{e,j}(\omega)}, \quad \bbM_{\ell}= - \frac{|\bk|\, {\bf e_3}\times \bbE}{\omega \, \mu(\omega) \, q_{m,\ell}(\omega)}, \quad \dot {\bbM}_{\ell}=\frac{\rmi \,\omega \,|\bk|\, {\bf e_3}\times \bbE}{\omega \, \mu(\omega) \, q_{m,\ell}(\omega)} .
\end{equation}
Of course  \eqref{EV1} and  \eqref{EV2} mean that $\mbox{dim } \mbox{ker } \big( \bbA_{|\bk|,\perp} - \omega \,\mathrm{I}d   \big) \leq \mbox{dim } \bC_{\perp}= 2$.  \\ [12pt]
To prove (ii), it suffices to check that 
$\forall \; \bbE \in \bC_\perp$, $\bbU = \big(\bbE, \bbH, \bbP, \dot \bbP, \bbM, \dot \bbM\big)\in \operatorname{ker } \big( \bbA_{|\bk|,\perp} - \omega \, \mathrm{Id} \big)$, when $\omega \in S(|\bk|)$ and $\big(\bbH, \bbP, \dot \bbP, \bbM, \dot \bbM\big)$ is given by (\ref{EV1}, \ref{EV2}). Thus, it consists to show that such vectors  $\bbU$ are solutions to (\ref{systemEH}, \ref{systemP}, \ref{systemM}) with $\bbF = 0$. 
This is equivalent (since $\omega\notin\calP$) to checking  (\ref{eq.PM}, \ref{EqEH}). This is rather immediate since \eqref{EV2} is nothing but \eqref{eq.PM} with $\bbH$ given by \eqref{EV1} while  \eqref{EV1} is nothing but \eqref{EqEH}(1).
It remains to check \eqref{EqEH}(2). However, substituting 
\eqref{EV1} into the left hand side  of \eqref{EqEH}(2) gives, as $-|\bk| \, {\bf e_3}\times (|\bk| \, {\bf e_3}  \times \bbE)=  |\bk|^2 \,\bbE$ for $\bbE$ in $\bC_\perp$, 
\begin{equation}\label{eq.eigenE}
	\big(|\bk|^2- \mathcal{D}(\omega) \big) \, \bbE=0,
\end{equation}
which is true for any $\bbE \in \bC_\perp$ since $\mathcal{D}(\omega) = |\bk|^2$  for $\omega \in S(|\bk|)$.

\subsection{Technical result for the analysis of the dispersion curves}\label{sec-asymptotic}
Let $z\in \mathbb{C}$. We are interested in the solutions $\omega\in \mathbb{C}$ of the  nonlinear equation 
\begin{equation} \label{eq_implicite}
 (\omega - z)^{\mathfrak{m}} \; g(\omega) =\zeta^{\mathfrak{m}}, \
\end{equation}
parametrized by $\zeta$, where $g$ is a given analytic function near $z$ satisfying $g(z)\neq 0$.
More precisely, we want to prove  that \eqref{eq_implicite} implicitly  defines $\omega$ in function of $\zeta$, via $\mathfrak{m}$ analytic branches of solutions $\omega_n(\zeta)$, for small $|\zeta|$. 
\begin{Lem}\label{Lem-implicte-function}
Let $\mathcal{G}$ be an analytic function on a domain $\Omega \subset \bbC$ and $z\in \Omega$ a zero of multiplicity $\mathfrak{m}\in\mathbb{N}^*$ of $\mathcal{G}$, so that $\mathcal{G}$ can be rewritten  as
\begin{equation}\label{eq.factoriz}
\mathcal{G}(\omega)=(\omega-z)^{\mathfrak{m}} g(\omega) \mbox{ with $g$ analytic in a vicinity of $z$ and }  g(z)=A\neq 0.
\end{equation}
 Then, there exist  an open neighborhood $U$  of $\zeta=0$  and $\mathfrak{m}$ distinct   analytic functions $\zeta \mapsto \omega_{n}(\zeta)$ for $n=1\ldots, m$ defined  on $U$ such that  
 \begin{equation}\label{eq.branchecondition}
\mathcal{G}(\omega_n(\zeta))=\zeta^\mathfrak{m}\quad\mbox{on $U$. }
\end{equation}
Moreover the functions $\omega_{n}(\zeta)$ have the following Taylor expansion
\begin{equation}\label{eq.asympexpansion}
\omega_{n}(\zeta)=z+  a_{n}^{-1} \zeta   -\frac{a_n^{-2} g'(z)}{\mathfrak{m}\,  g(z)}  \,\zeta^2+O(\zeta^3),  \mbox{ as } \zeta\to 0, \, \mbox{ where } a_{n}=|A|^{1/\mathfrak{m}} \; \rme^{\rmi \,\frac{\theta}{\mathfrak{m}}} \; \rme^{2 \, \rmi \,\frac{n \pi}{\mathfrak{m}}}, 
\end{equation}
$\theta\in (-\pi, \pi]$ is the  principal value of the argument of $A$ and the complex coefficients  $a_{n}$ are the $\mathfrak{m}$ distinct roots of the polynomial  equation $X^{\mathfrak{m}}=A$.\\ [12pt]
Furthermore, one can find an open neighborhood $D$  of $z$  included in $\Omega$ such that the set of solutions of $G(\omega)=\zeta^m$ in $D$ for $\zeta\in U$ is exactly given by $\{ \omega_{n}(\zeta) \mbox{ for } n=1, \ldots, \mathfrak{m}\}$.
\end{Lem}
\begin{proof}
We want to find branches of solutions  $\zeta \mapsto \omega_{n}(\zeta)$ of  the equation 
\begin{equation}\label{eq.implicite}
\mathcal{G}(\omega)= \zeta^\mathfrak{m},
\end{equation}
where $\zeta$ is a parameter that lies  in  an open neighborhood of $0$. 
To this aim, fix $1\leq n\leq \mathfrak{m}$ and make the following  change of unknown $\omega \rightarrow  \eta$, for $\zeta \neq 0$:
\begin{equation}\label{eq.substitution}
\eta=-1+ (\omega-z)\, a_n \, \zeta^{-1} \quad \Longleftrightarrow \quad \omega=z+a_n^{-1}\,\zeta \,(1+\eta) 
\end{equation}
where the coefficient $a_n$ is given by \eqref{eq.asympexpansion}.
Therefore by replacing \eqref{eq.substitution} into \eqref{eq.implicite} and  using  \eqref{eq.factoriz} and  the fact  that $g(z)=A=a_n^{\mathfrak{m}}$,
we notice that on  an open neighborhood of $\zeta=0$, finding an analytic branches  $\zeta\mapsto \omega_n(\zeta)$  of solutions of \eqref{eq.implicite} is equivalent to find an analytic branch  $\zeta \mapsto \eta_n(\zeta)$ satisfying 
\begin{equation}\label{eq.implicitfunction}
H_n(\zeta,\eta)= 0 \  \mbox{ with } \  H_n(\zeta,\eta):= (1+\eta)^\mathfrak{m}\,g(z+a_n^{-1}\,\zeta \,(1+\eta))-g(z),
\end{equation}
with $\omega_n$ given in term of $\eta_n$ by \eqref{eq.substitution}.
As $g$ is analytic in a vicinity of $z$, by the Hartogs Theorem (see   \cite{Muj-86}, Theorem 36.8 page 271),  the function $H_n(\zeta,\eta)$ is  well-defined  and analytic in a neighbourhood of $(0,0)$ in $\bbC\times \bbC$ and satisfies $H_n(0,0)=0$. We then want to use the analytic implicit function Theorem (see  \cite{Fri-2002}, Theorem 7.6 page 34) to solve $H_{n}(\zeta,\eta)= 0$. Thus,  one first has to check that $\partial_{\eta }H_{n}(0,0)\neq 0$. Indeed,  from \eqref{eq.implicitfunction}, we compute
$$
\partial_{\eta }H_{n}(\zeta,\eta)= \mathfrak{m}\,(1+\eta)^{\mathfrak{m}-1}\,g(z+a_n^{-1}\,\zeta \,(1+\eta))+(1+\eta)^\mathfrak{m}\,  a_n^{-1}\,\zeta  \, \,g'(z+a_n^{-1}\,\zeta \,(1+\eta))
$$
so that in particular
\begin{eqnarray}\label{eq.derivimpliceta}
	\partial_{\eta }H_{n}(0,0) =  \mathfrak{m}\,  g(z)\neq 0.
\end{eqnarray}
Hence, the  analytic implicit function Theorem proves the existence of two open neighbourhoods of the origin $U_n \subset \bbC$ (the $\zeta$-complex plane) and $V_n  \subset \bbC $ (the $\eta$-complex plane) and of an analytic function $\eta_n (\zeta): U_n\to V_n$   satisfying $\eta(0)=0$ such that 
$$\{ (\zeta, \eta)\in U_n\times V_n \mid H_n(\zeta,\eta_n)=0\}=\{  (\zeta, \eta_n(\zeta)), \, \zeta \in  U_n \}.$$ 
Moreover, one has \begin{equation}\label{eq.etaderivzero}
	\eta'(0)=-\frac{\partial_\zeta H_n(0,0)}{\partial_\eta H_n(0,0)}=-\frac{a_n^{-1} g'(z)}{\mathfrak{m}\,  g(z)}.
\end{equation}
Thus  $\zeta \mapsto \omega_n(\zeta) := z+a_n^{-1}\,\zeta \,(1+\eta_n(\zeta))$ is an analytic branch of solution of  \eqref{eq.implicite} in $U_n$.
Since $\eta_n(0)=0$, as  $\zeta \to 0$, $\eta_n(\zeta)=\zeta \,\eta_n'(0)+O(\zeta^2),$  which, put  in the definition of  $\omega_n(\zeta)$, gives \eqref{eq.asympexpansion}. \\ [12pt]
We point out that, as the coefficients $a_n^{-1}$ in \eqref{eq.asympexpansion} are all distinct, we have  constructed $\mathfrak{m}$ distinct  analytic functions $\zeta \to \omega_{n}(\zeta)$ for $n=1\ldots, \mathfrak{m}$ defined on the open neighborhood $\tilde{U}$, given by the intersection of the sets $U_n$ for $n\in \{ 1,\ldots, \mathfrak{m}\}$,  which satisfy  \eqref{eq.branchecondition} and \eqref{eq.asympexpansion}. \\ [12pt] 
The last point is proven with Rouch\'e Theorem (see for e.g. Theorem 10.43 page 225 of \cite{Rud-87}).  As $g$ does not vanish in the vicinity of $z$, we can find $r>0$ such that, if $D$  is the open disk of radius $r$ and center $z$, $\overline D \subset \Omega$, $g$  is analytic  in an open set containing $\overline D$   and  $g$ does not vanish in   $\overline D$. Hence, by \eqref{eq.factoriz}, the only zero of $G(\omega)$ in $D$ is $z$, with multiplicity $\mathfrak{m}$. Let us define the function $G_{\zeta}:\omega\mapsto G(\omega)-\zeta^m$ which is analytic in $\Omega$ and $\partial D$ the circle of center $z$ and radius $r$. Then, one has:
$$\forall \; \omega \in \partial D,  \quad  |G_{\zeta}(\omega)- G(\omega)|=|\zeta|^m \mbox{ and by \eqref{eq.factoriz}} , \ |G(\omega)|\geq r^m \min_{\omega\in  \partial D} |g(\omega)|>0.$$
Thus for $R=r \, \min_{\omega \in  \partial D } |g(\omega)|^{1/m}$, one gets 
$$ \forall \; |\zeta| < R, \quad \forall \; \omega \in \partial D, \quad \big |G_{\zeta}(\omega)- G(\omega) \big|<|G(\omega)|. $$  
Therefore by the Rouch\'e Theorem, $G$ and  $G_{\zeta}$ have the same number of zeros (counted with multiplicity) in $D$. Hence $G_{\zeta}$ has $\mathfrak{m}$ zeros on $D(z,r)$ for $|\zeta|<R$. Let $D(0,R)$ be the open disk of center $0$ and radius $R$. Thus, if we define $U$ as $\tilde{U}\cap D(0,R)$, then there are no other solutions of $G(\omega)=\zeta^m$ in $D$  when $\zeta \in U$ than $\omega_n(\zeta)$ for $n\in \{1,\ldots, \mathfrak{m} \}$.
\end{proof}
\begin{Rem} The  result of Lemma \ref{Lem-implicte-function} is  related to Puiseux series. Puiseux series have been  derived in \cite{Hin-10,Kat-80,Wel-11} for  the perturbation analysis of eigenvalues of non-selfdjoint matrices. However the results  of \cite{Wel-11} do not apply directly here since  in some situations we are perturbing derogatory eigenvalues (i.e. non simple eigenvalues). Moreover, we need to compute explicitly the two first coefficients of the Puiseux  series which is not done in \cite{Kat-80}, Section 1.2 nor in \cite{Hin-10}, see Corollary 4.2.9.
\end{Rem} 
	

\end{document}